\documentclass[11pt,letterpaper]{article}

\usepackage[T1]{fontenc}
\usepackage{newpxtext}
\usepackage{mathpazo}
\usepackage{amsthm, amsfonts, amssymb, graphicx}
\usepackage{mathtools}
\usepackage{tikz}
\usepackage{blkarray}
\usepackage{enumitem}
\usepackage{dsfont}
\usepackage[margin=1truein]{geometry}
\usepackage[full]{complexity}
\usepackage{enumitem}

\usepackage[pdfstartview=FitH,pdfpagemode=UseNone,colorlinks=true,citecolor=blue,linkcolor=blue,backref=page,hypertexnames=false]{hyperref}
\definecolor{blueviolet}{RGB}{60,50,200}
\definecolor{oliveg}{RGB}{40,200,30}
\hypersetup{colorlinks=true,       
    linkcolor=blueviolet,
    citecolor=oliveg,
}

\usepackage{comment, bbm, enumitem}
\usepackage{xcolor}
\usepackage{pmat}
\usepackage{algorithm}
\usepackage{algpseudocode}
\usepackage{underscore}
\usetikzlibrary{matrix}
\usepackage{amsmath} 
\usepackage[capitalize]{cleveref}
\usepackage{authblk}

\crefname{observation}{Observation}{observations}
\crefname{ALC@line}{Step}{Steps}
\Crefname{ALC@line}{Step}{Steps}

\usepackage{nicematrix}
\usepackage{mathtools}
\definecolor{DeepGreen}{RGB}{6,64,43}

\DeclarePairedDelimiterX{\Set}[2]{\{}{\}}{\,{#1}\,:\,{#2}\,}
\DeclarePairedDelimiter{\card}{|}{|}

\DeclareMathOperator*{\tw}{ct}

\DeclareMathOperator*{\diag}{diag}
\DeclareMathOperator*{\rank}{rank}

\newcommand{\comp}[1]{\overline{#1}}

\newcommand{\rg}[1]{\textcolor{red}{\guillemotleft \textbf{Rohit:} #1\guillemotright}}

\newcommand{\RR}[1]{{\color{blue}[{\tiny \textbf{Roshan:}\bf #1}]\marginpar{\color{blue}*}}}
\newcommand{\sgc}[1]{\textcolor{orange}{\guillemotleft \textbf{Sumanta:} #1\guillemotright}}

\newcommand{\F}{\mathbb{F}}

\newtheorem{theorem}{Theorem}[section]
\newtheorem{lemma}{Lemma}[section]

\newtheorem{observation}{Observation}[section]
\newtheorem{claim}{Claim}[section]

\newtheorem{corollary}{Corollary}[section]
\newtheorem{definition}{Definition}[section]
\newtheorem{remark}{Remark}[section]

\newtheorem{problem}{Problem}[section]

\DeclareMathOperator{\calP}{\mathcal{P}}
\newcommand{\pp}[2]{\{\{#1\},\{#2\}\}}
\DeclareMathOperator*{\ct}{ct}

\DeclareMathOperator{\ad}{adj}
\newcommand{\adj}[1]{(#1)^{\ad}}

\newcommand{\DS}{\stackrel{\mathclap{\scalebox{.5}{\mbox{DS}}}}{=}}
\newcommand{\nDS}{\stackrel{\mathclap{\scalebox{.5}{\mbox{DS}}}}{\neq}}
\newcommand{\DE}{\stackrel{\mathclap{\scalebox{.5}{\mbox{DE}}}}{=}}
\newcommand{\PME}{\stackrel{\mathclap{\scalebox{.4}{\mbox{PME}}}}{=}}
\newcommand{\nPME}{\stackrel{\mathclap{\scalebox{.4}{\mbox{PME}}}}{\neq}}
\newcommand{\kPME}[3]{\mathrm{PME}(#1, #2, #3)}
\newcommand{\newPME}[4]{\mathrm{PME}(#1, #4, #2#3)}
\newcommand{\nsubset}[1]{[n]\setminus\{#1\}}

\newcommand{\calD}{\mathcal{D}}
\newcommand{\calA}{\mathcal{A}}
\newcommand{\calB}{\mathcal{B}}
\newcommand{\prop}{\mathcal{R}}

\newcommand{\PMoracle}[1]{\textsc{PM}_{#1}}

\newcommand{\submatrixfamily}[1]{\mathcal{F}_{#1}}
\newcommand{\fourPMEfamily}[1]{\mathcal{S}_{#1}}

\DeclareMathOperator*{\pol}{poly}

\algdef{SE}{PFor}{EndPFor}[1]{\textbf{for all} \(\mbox{#1}\) \textbf{do in parallel}}{\textbf{end for}}

\algnewcommand{\algorithmicinput}{\textbf{Input:}}
\algnewcommand{\Input}{\item[\algorithmicinput]}
\algnewcommand{\algorithmicoutput}{\textbf{Output:}}
\algnewcommand{\Output}{\item[\algorithmicoutput]}

\newenvironment{repstatement}[1]{\noindent {\bf \cref{#1} (restated).} \em}{}

\newcommand{\probA}{\text{Learning RODs}}
\newcommand{\probB}{\text{Black-box PMAP}}
\newcommand{\calC}{\mathcal{C}}

\newcommand{\N}{\mathbb{N}}

\setlist[description]{font=\normalfont\itshape\space}
\newlength{\jeroenlen}
\newenvironment{remarks}
{\settowidth{\jeroenlen}{\textit{Remarks.}}%
	\begin{description}[leftmargin=\jeroenlen,labelwidth=0pt,labelsep=0pt]
		\item[\textit{Remarks.\qquad}]%
		\begin{enumerate}[leftmargin=0em,labelsep=0.5em]}
		{\end{enumerate}\end{description}}
\crefname{remarks}{Remark}{Remarks}
\Crefname{remarks}{Remark}{Remarks}
\renewcommand{\char}{\textnormal{char}}

\title{Learning Read-Once Determinants and the Principal Minor Assignment Problem}
\author[1]{Abhiram Aravind}
\author[2]{Abhranil Chatterjee}
\author[3]{Sumanta Ghosh}
\author[4]{Rohit Gurjar}
\author[5]{Roshan Raj}
\author[1]{Chandan Saha}

\affil[1]{Indian Institute of Science, Bangalore}
\affil[2]{Indian Institute of Technology Kanpur}
\affil[3]{Indian Statistical Institute, Kolkata}
\affil[4]{Indian Institute of Technology Bombay}
\affil[5]{Tata Institute of Fundamental Research, Mumbai}
 \date{}

\begin{document}

\maketitle

\begin{abstract}
A symbolic determinant under rank-one restriction computes a polynomial of the form $\det(A_0 + A_1y_1 + \ldots + A_ny_n)$, where $A_0, A_1, \ldots, A_n$ are square matrices over a field $\F$ and $\rank(A_i) = 1$ for each $i \in [n]$. This class of polynomials has been studied extensively, since the work of Edmonds (1967), in the context of linear matroids, matching, matrix completion and polynomial identity testing. We study the following learning problem for this class: Given black-box access to an $n$-variate polynomial $f = \det(A_0 + A_1y_1 + \ldots + A_ny_n)$, where $A_0, A_1, \ldots, A_n$ are unknown square matrices over $\F$ and $\rank(A_i) = 1$ for each $i \in [n]$, find a square matrix $B_0$ and rank-one square matrices $B_1, \ldots, B_n$ over $\F$ such that $f = \det(B_0 + B_1y_1 + \ldots + B_ny_n)$. In this work, we give a randomized $\pol(n)$ time algorithm to solve this problem; the algorithm can be derandomized in quasi-polynomial time. To our knowledge, this is the first efficient learning algorithm for this class. As the above-mentioned class is known to be equivalent to the class of \emph{read-once determinants} (RODs), we will refer to the problem as \emph{learning RODs}. An ROD computes the determinant of a matrix whose entries are field constants or variables and every variable appears at most once in the matrix. Thus, the class of RODs is a rare example of a well-studied class of polynomials that admits efficient \emph{proper} learning. 

The algorithm for learning RODs is obtained by connecting with a well-known open problem in linear algebra, namely the \emph{Principal Minor Assignment Problem} (PMAP), which asks to find (if possible) a matrix having prescribed principal minors. PMAP has also been studied in machine learning to learn the kernel matrix of a determinantal point process. Here, we study a natural \emph{black-box} version of PMAP: Given black-box access to an $n$-variate polynomial $f = \det(A + Y)$, where $A \in \F^{n \times n}$ is unknown and $Y = \diag(y_1, \ldots, y_n)$, find a $B \in \F^{n\times n}$ such that $f = \det(B + Y)$. We show that black-box PMAP can be solved in randomized $\pol(n)$ time, and further, it is randomized polynomial-time equivalent to learning RODs. The algorithm and the reduction between the two problems can be derandomized in quasi-polynomial time. To our knowledge, no efficient algorithm to solve this black-box version of PMAP was known before.

We resolve black-box PMAP by investigating a crucial property of dense matrices that we call the \emph{rank-one extension} property. Understanding ``cuts'' of matrices with this property and designing a \emph{black-box cut-finding} algorithm to solve PMAP for such matrices (using only principal minors of order $4$ or less) constitute the technical core of this work. The insights developed along the way also help us give the first NC algorithm for the \emph{Principal Minor Equivalence} problem, which asks to check if two given matrices have equal corresponding principal minors. 

\end{abstract}

\thispagestyle{empty}
\newpage

\setcounter{tocdepth}{2}
\tableofcontents
\thispagestyle{empty}
\newpage
\setcounter{page}{1}

\section{Introduction}
Learning arithmetic circuits is a fundamental problem in algebraic complexity theory. An arithmetic circuit, which is an arithmetic analogue of a Boolean circuit, consists of addition and multiplication operations and computes a multivariate polynomial in the input variables. The \emph{learning problem} (a.k.a. \emph{reconstruction problem}) for a circuit class $\calC$ asks to find a circuit $C$ from black-box access to a polynomial $f$ that is computable by a circuit in class $\calC$ such that $C$ computes $f$. Black-box access to $f$ allows us to obtain the value of $f$ at any chosen point in one time step. If the output circuit $C$ belongs to class $\calC$, then we say that the learning is \emph{proper}.

Despite research spanning over two decades, efficient proper learning algorithms are known for only a handful of circuit classes. This is not surprising, as proper learning is considered a computationally hard problem for most circuit classes. Some of the important classes for which proper learning is known are sparse polynomials (that is, depth two circuits) \cite{Ben-OrT88, KlivansS01}, read-once formulas (ROFs) \cite{HancockH91, BshoutyHH95, ShpilkaV14, MinahanV18}, and read-once oblivious algebraic branching programs (ROABPs) with known variable ordering \cite{BeimelBBKV00, KlivansS06}. Algebraic branching programs (ABPs), which can be viewed as skew circuits, form a powerful model (class) of arithmetic computation. One of the most important polynomials in algebraic complexity, namely the \emph{determinant} polynomial, can be computed by an ABP of polynomial size \cite{Berkowitz84, MahajanV97}. Due to a known reduction from ABPs to symbolic determinants \cite{Valiant79a}, it can be assumed that an ABP computes a polynomial of the form $\det(A_0 + A_1y_1 + \ldots + A_ny_n)$, where $y_1, \ldots, y_n$ are variables and $A_0, A_1, \ldots, A_n \in \F^{r \times r}$ are $r \times r$ matrices over an underlying field $\F$. 

A subclass of ABPs that has received a lot of attention in the past is \emph{symbolic determinants under rank-one restriction}. A polynomial $f$ is computable by this class if $f$ can be expressed as $\det(A_0 + A_1y_1 + \ldots + A_ny_n)$, where $A_1, \ldots, A_n \in \F^{r \times r}$ are rank-one matrices; there is no restriction on $A_0 \in \F^{r \times r}$. This class has been studied extensively over the past several decades in the context of matroid problems, matrix completion, and polynomial identity testing \cite{Edmonds67, Edmonds68, Edmonds79, Lovasz89, Murota93, Geelen99, HarveyKM05, IvanyosKS10, DBLP:journals/jcss/IvanyosKQS15, GT20}. The border complexity of this class has also been recently investigated \cite{CGGR23}. In this work, we study the learning problem for this class.

\begin{problem} \label{prob:learningROD}
Given black-box access to an $n$-variate polynomial $f = \det(A_0 + A_1y_1 + \ldots + A_ny_n)$, where $A_0, A_1, \ldots, A_n$ are unknown square matrices over $\F$ and $\rank(A_i) = 1$ for each $i \in \{1, \ldots, n\}$, find a square matrix $B_0$ and rank-one square matrices $B_1, \ldots, B_n$ over $\F$ such that $f = \det(B_0 + B_1y_1 + \ldots + B_ny_n)$.

\end{problem}

It is known that a polynomial computable by a symbolic determinant under rank-one restriction is also computable by a \emph{read-once determinant} (ROD) of small size (see Lemma 4.3 in \cite{GT20}). \\

\noindent{\textbf{Read-once determinants.}}  An ROD is the determinant of a matrix $M$ whose entries are variables or field constants such that every variable appears at most once in $M$. The dimension of $M$ defines the \emph{size} of the ROD. In the converse direction, it is easy to see that an ROD is (by definition) a symbolic determinant under rank-one restriction. As the two classes, RODs and symbolic determinants under rank-one restriction, are essentially the same, we will refer to Problem \ref{prob:learningROD} as \underline{\emph{Learning RODs}}. Both ROABPs and RODs are more expressive classes than ROFs \cite{Valiant79a}. But, unlike ROABPs, RODs are \emph{not universal} in the sense that not every polynomial can be expressed as an ROD. For example, the elementary symmetric polynomial and the permanent polynomial cannot be expressed as ROD \cite{AravindJ15}. On the other hand, the determinant and the iterated matrix multiplication polynomials (both of which define complete polynomial families for the important algebraic complexity class $\mathsf{VBP}$) admit small RODs \cite{Valiant79a}. It is worth noting that any ROABP computing the determinant has exponential size \cite{Nisan91}, and so, learning ROABPs is not a viable way to learn RODs, even improperly. Thus, Problem \ref{prob:learningROD} is interesting in its own right in algebraic complexity theory. Besides, there is another strong motivation to study this problem which comes from its connection with the \emph{Principal Minor Assignment Problem} (PMAP).  \\

\noindent{\textbf{Principal Minor Assignment Problem.}} PMAP is a well-known open problem in linear algebra \cite{HS02, GriffinTsatsomeros2006}. Given a list $\calP = (p_S)_{\varnothing \neq S \in 2^{[n]}}$ of $2^n - 1$ field elements, PMAP asks to find (if possible) a matrix $A \in \F^{n \times n}$ such that the principal minor of $A$ defined by the rows and columns indexed by $S$ is $p_S$ for all nonempty $S \in 2^{[n]}$. If $\calP$ is promised to be the list of principal minors of some unknown matrix $A$, then it is typically assumed that $\calP$ is given as an oracle that returns $p_S$ when queried with $S$; the goal is to minimize the number of queries to the oracle $\calP$. Efficient algorithms for PMAP that ask only $\pol(n)$ queries to the oracle are known if the matrix $A$ to be reconstructed is special, such as a symmetric or magnitude-symmetric matrix \cite{RisingKuleszaTaskar2015, Brunel2018, BrunelUrschel2024} or a certain kind of dense matrix \cite{GriffinTsatsomeros2006}. But the problem of learning an arbitrary $A$ in $\pol(n)$ time using oracle access to the list $\calP$ remains open. Although PMAP is interesting in its own right in linear algebra, another motivation to study it comes from learning \emph{determinantal point processes} (DPPs).

A DPP associated with an $n\times n$ kernel matrix $K$ is a probability distribution over subsets of a set of $n$ items, where the probability of getting any subset $T\subseteq [n]$ is given by the principal minor of $K$ indexed by $T$ (see \cite{KuleszaTaskar2012}). 
DPPs are useful in settings where one needs to sample a diverse set of objects. The learning question for DPP asks to learn the associated kernel matrix, given samples from the distribution. Efficient PMAP algorithms for symmetric matrices \cite{Oeding2011-jo,RisingKuleszaTaskar2015} played a crucial role in making significant progress on learning symmetric DPPs~\cite{Urschel2017} as the principal minors of the kernel can be well estimated from the observed samples using a method of moments. A similar progress via PMAP is also made for learning signed DPPs in which the underlying kernel is magnitude-symmetric \cite{Brunel2018, BrunelUrschel2024}.   \\

Now, observe that the coefficients of a polynomial $f = \det(A + Y)$, where $Y = \diag(y_1, \ldots, y_n)$ and $A \in \F^{n \times n}$, are the principal minors of $A$. This naturally motivates us to study the following \emph{black-box version} of the Principal Minor Assignment Problem.   

\begin{problem} \label{prob:black-boxPMAP}
Given black-box access to an $n$-variate polynomial $f = \det(A + Y)$, where $A \in \F^{n \times n}$ is unknown and $Y = \diag(y_1, \ldots, y_n)$, find a $B \in \F^{n\times n}$ such that $f = \det(B + Y)$.
\end{problem}
As the black-box allows us to recover a principal minor of $A$ in $\pol(n)$ time using interpolation (assuming $|\F| > n$), we will refer to Problem \ref{prob:black-boxPMAP} as \underline{\emph{Black-box PMAP}}. Note that $\det(A + Y)$ is an ROD, and so, Problem \ref{prob:black-boxPMAP} is clearly related to Problem \ref{prob:learningROD}. Also, both are proper learning problems. 

Observe that the output $B$ in Problem \ref{prob:black-boxPMAP} has the same corresponding principal minors as $A$. Two matrices are said to be \emph{principal minor equivalent} (PME) if they have equal corresponding principal minors of all orders. The decision problem of testing if two given matrices are PME is naturally related to PMAP. A characterization of PME was known for some special cases like symmetric matrices~\cite{ES80}, matrices with no cuts~\cite{Hartfiell84,Loewy86}, and few other classes related to skew-symmetry~\cite{BOUSSAIRI201547,BC16,BoussairiChaichaaCherguiLakhlifi2021}. 
Recently, a characterization of PME was proven for the general case~\cite{Chatterjee25},
where a polynomial-time deterministic algorithm was also given to test if two given matrices are PME. 
Informally speaking, their characterization proves that two matrices are PME if and only if they are related by a sequence of simple operations -- called \emph{cut-transpose} (see Section \ref{sec: cut transpose}). 
The algorithm starts from one of the matrices, keeps applying appropriate cut-transpose operations and reaches the other matrix in at most $2n$ steps or says that they are not PME. 
The algorithm seems to be inherently sequential, raising an interesting question: 
\begin{problem} \label{prob:PMEinNC}
Is there an efficient parallel (NC) algorithm for testing PME? 
\end{problem}
The class NC contains problems for which we have parallel algorithms running in poly-logarithmic time using only a polynomial number of processors. Interestingly, if we allow randomization then there is a simple parallel algorithm to test PME via a reduction to polynomial identity testing (PIT). As indicated before, two matrices $A$ and $B$ are PME if and only if the polynomial identity $\det(A+Y) = \det(B+Y)$ holds.  
A simple randomized algorithm to check a polynomial identity is to plug in random numbers for variables and check equality~\cite{Sch80,DL78,Zip79}, which succeeds with high probability.
Note that the determinant can be computed efficiently in parallel (NC) (see, e.g., \cite{Csanky76, Berkowitz84}), and thus, the above algorithm is in RNC.
While derandomization of PIT is a challenging open question, we can ask whether we can derandomize the specific instance arising from PME and get an NC algorithm.

In this work, we give randomized polynomial-time algorithms to solve Problem \ref{prob:learningROD} (Learning RODs) and Problem \ref{prob:black-boxPMAP} (Black-box PMAP) (see Theorem \ref{thm:main}). In fact, we show a randomized polynomial-time equivalence between the two problems (in Section \ref{sec:equiv bet probA and probB}), and then solve the latter in randomized polynomial time (in Sections \ref{sec:reduction to matrices with prop R}-\ref{sec: PMAP for matrices with prop R}). The steps involving randomization (in the equivalence and the algorithm) can be derandomized in quasi-polynomial time. Thus, the class defined by RODs is one of the rare classes of polynomials for which efficient proper learning is possible. Also, to our knowledge, no efficient algorithm for any natural black-box version of PMAP was known prior to this work. Further, we give an NC algorithm for testing PME, thereby solving Problem \ref{prob:PMEinNC} (see Theorem \ref{thm:main-NC-algorithm} and Section \ref{sec:NC-algo-for-PME}). These algorithms are obtained by studying the principal minors of order at most $4$ of dense matrices satisfying a key property that we call the ``\emph{rank-one extension}'' property (see Theorem \ref{thm:main-characterization} and \ref{thm:main-reconstruction-PM-with-Q-property}). 

\subsection{Our results} \label{sec:results}
We now state our results formally. The results hold under a mild condition on the size of the underlying field $\F$. We also assume that a field operation takes unit time, and square root computation over $\F$ can be done in polynomial time\footnote{Over $\mathbb{Q}$, square roots can be computed in deterministic polynomial time. Over finite fields, there are multiple randomized polynomial time square root finding algorithms. Moreover, if a quadratic non-residue over $\mathbb{F}_q$ is known (which is a well justified assumption under GRH) then square roots can be computed in deterministic polynomial time over $\mathbb{F}_q$}.  

\begin{theorem}[Main theorem] \label{thm:main}
Let $n \in \N$ and $|\F| > n^6$. 
\begin{enumerate}
\item \emph{(Learning RODs)} There is a $\poly(n)$ time randomized algorithm that takes input black-box access to an $n$-variate polynomial $f = \det(A_0 + A_1y_1 + \ldots + A_ny_n)$, where $A_0 \in \F^{r \times r}$ is unknown and $A_1, \ldots, A_n \in \F^{r \times r}$ are unknown rank-one matrices and $r\in \N$ is also unknown, and outputs (with high probability) a matrix $B_0 \in \F^{n\times n}$ and rank-one matrices $B_1, \ldots, B_n \in \F^{n \times n}$ such that $f = \det(B_0 + B_1y_1 + \ldots + B_ny_n)$. The algorithm can be derandomized in quasi-polynomial time.

\item \emph{(Black-box PMAP)} There is a $\poly(n)$ time randomized algorithm that takes input black-box access to an $n$-variate polynomial $f = \det(A + Y)$, where $A \in \F^{n \times n}$ is unknown and $Y = \diag(y_1, \ldots, y_n)$, and outputs (with high probability) a matrix $B \in \F^{n\times n}$ such that $f = \det(B + Y)$. The algorithm can be derandomized in quasi-polynomial time.
\end{enumerate}
\end{theorem}

\noindent{\textbf{An algorithm for PMAP.} The algorithm for Black-box PMAP immediately implies a $2^n \cdot \pol(n)$ time randomized algorithm for PMAP as the input list $\calP = (p_S)_{\varnothing \neq S \in 2^{[n]}}$ can be used to implement a query to the black-box in $2^n \cdot \pol(n)$ time. Note that the complexity is almost linear in the size of the input $\calP$ which is the best one can hope for if it is not promised that $\cal{P}$ is consistent with the principal minors of some matrix. In the absence of such a promise, one needs to check if $\cal{P}$ is consistent with the principal minors of the reconstructed matrix. As the size of $\cal{P}$ is $2^n - 1$ and principal minors of a matrix can be computed in $O(2^n)$ time \cite{GriffinTsatsomeros2006b}, any algorithm for PMAP takes $\Omega(2^n)$ time. In this sense, we give a \emph{near optimal} time algorithm for PMAP. To our knowledge, no $\tilde{O}(2^n)$ time algorithm was known for PMAP before this work. A naive reduction to polynomial solvability incurs a cost of $n^{O(n^4)}$ over $\mathbb{C}$ \cite{I89}, as there are $n^2$ unknown entries of a matrix. 

If it is promised that $\cal{P}$ is the list of principal minors of some matrix, then it is conceivable that far fewer queries to the oracle $\cal{P}$ suffice. For instance, if $\cal{P}$ is the list of principal minors of a symmetric matrix, then $O(n^2)$ queries are necessary and sufficient \cite{RisingKuleszaTaskar2015}. The algorithm is based on \cite{Oeding2011-jo} that characterizes the variety of principal minors of symmetric matrices. Unfortunately, the ideal of polynomial relations satisfied by the principal minors of a generic $n \times n$ matrix is not well-understood (except for $n = 4$ \cite{LinSturmfels2008}). Nevertheless, we show that PMAP for matrices satisfying a certain genericity condition can be solved using $\pol(n)$ queries (see Theorem \ref{thm:main-reconstruction-PM-with-Q-property}). \\

A proof sketch of Theorem \ref{thm:main} is given in Section \ref{sec:ideas thm main}. A couple of remarks are in order:
\begin{remarks}
\item  The condition on the size of $\F$ in Theorem \ref{thm:main} originates from the need to apply the Polynomial Identity Lemma (see Lemma \ref{lemma:random-diag-shift-satisfies-R}). If the black-box allows us to make queries at points from a field extension, then we can drop this assumption on the field size. 

\item The assumption on the efficiency of square root finding arises from the need to solve quadratic equations (see Algorithm \ref{algo:reconstruction-of-size-four-matrices} in Section \ref{sec: PMAP for matrices with prop R}). If $\char(\F) = 2$, then solving quadratic equations reduces to solving a system of linear equations over $\F_2$.

\end{remarks}

A proof sketch of the following theorem on PME is given in Section \ref{sec:ideas thm NC}.
\begin{theorem}[PME is in NC]
\label{thm:main-NC-algorithm}
Testing principal minor equivalence of two given matrices is in NC.
\end{theorem}

The algorithms in Theorem \ref{thm:main} and \ref{thm:main-NC-algorithm} are obtained by studying a crucial property of matrices which we define next. A proof sketch of Theorem \ref{thm:main-characterization} below is provided in Section \ref{sec:ideas characterization of matrices with prop Q}.

\begin{definition}[Rank-one extension and property $\prop$] \label{defn:rank one extension}
A matrix is \underline{\emph{dense}} if all its off-diagonal entries are nonzero. We say that a matrix $A$ satisfies the \underline{\emph{rank-one extension}} property if the following condition holds: if $\rank(A[\{i,j\},\{k,\ell\}])=1$ for any four distinct indices $i,j,k,\ell \in [n]$, then 
there exists a subset $S \subseteq [n]$ with $i,j \in S$ and $k,\ell \in \comp{S}$ such that $\rank(A[S,\comp{S}])=1$. We say that a matrix satisfies \underline{\emph{property $\prop$}} if it is dense and satisfies the rank-one extension property.
\end{definition}

Understanding and exploiting the rank-one extension property is at the heart of the technical contributions of this work. The following theorems demonstrate the usefulness of this property. 

\begin{theorem} [Sufficiency of PME up to order $4$]
\label{thm:main-characterization}
Let $A$ be a matrix satisfying property $\prop$. 
Let $B$ be another matrix of the same size such that $A$ and $B$ have equal corresponding principal minors of order at most 4. 
Then $A$ and $B$ have equal corresponding principal minors of all orders.
\end{theorem}

What makes property $\prop$ so effective is that Black-box PMAP for arbitrary matrices randomly reduces to that for matrices satisfying property $\prop$ (see Section \ref{sec:reduction to matrices with prop R}). This fact, coupled with the next theorem (whose proof sketch appears in Section \ref{sec:ideas thm main}), gives us a proof of Theorem \ref{thm:main}. 

\begin{theorem} [PMAP for matrices with property $\prop$]
\label{thm:main-reconstruction-PM-with-Q-property}
There exists a polynomial-time algorithm that, given oracle access to all the principal minors of order at most $4$ of a matrix $A \in \F^{n \times n}$ satisfying property $\prop$, reconstructs a matrix $B$ such that the corresponding principal minors of $A$ and $B$ are equal.
\end{theorem}

\begin{remarks}
\item The above theorem solves PMAP for matrices satisfying property $\prop$ in $\pol(n)$ time as we need oracle access to only the principal minors of order at most $4$. In comparison, we use black-box access to $\det(A + Y)$ (in Theorem \ref{thm:main}) to solve $\probB$ for an \emph{arbitrary} $A$. It turns out that the proof of Theorem \ref{thm:main} requires only oracle access to the principal minors of $A$ and that of $A+D$ for a random diagonal matrix $D$.

\item A random matrix satisfies property $\prop$ as it does not have a $2 \times 2$ submatrix of rank one with high probability. So, property $\prop$ defines a genericity condition. In \cite{GriffinTsatsomeros2006}, a $\poly(n)$ time PMAP algorithm was given for matrices satisfying another genericity condition, namely \emph{off-diagonal full (ODF)}. Property $\prop$ and ODF seem incomparable (see Section \ref{sec:comparison}). However, Black-box PMAP for aribitrary matrices reduces to PMAP for matrices with property $\prop$ (see Section \ref{sec:reduction to matrices with prop R}). It is unclear if the same is true for matrices satisfying ODF.
\item A \emph{random kernel} of a DPP can be defined as follows: Let $n > 1$, $\lambda \in (0, 1/2)$ and $\mu \in [0, \lambda/(n-1))$. Let $D \in \mathbb{R}^{n \times n}$ be a diagonal matrix with $D[i,i] \in [\lambda, 1-\lambda]$ for all $i \in [n]$, and $A\in \mathbb{R}^{n \times n}$ whose entries are randomly and independently chosen from $[-1,1]$. As there is a DPP with kernel $K = D + \mu A$ (see Section 2.2 in \cite{Brunel2018}), we can say that $K$ is a random kernel. Note that $K$ satisfies property $\prop$ as $A$ is random. In this sense, a random kernel can be learnt from its principal minor oracle in $\poly(n)$ time using Theorem \ref{thm:main-reconstruction-PM-with-Q-property} .
\end{remarks}



\subsection{Proof ideas} \label{sec:ideas}
\subsubsection{Proof sketch for Theorem \ref{thm:main} and Theorem \ref{thm:main-reconstruction-PM-with-Q-property}} \label{sec:ideas thm main}
As an initial first step towards proving Theorem \ref{thm:main}, we establish a randomised polynomial-time reduction from Learning RODs to Black-box PMAP (in fact, we show an equivalence between the two problems in Section \ref{sec:equiv bet probA and probB}). Then, we use the three broad steps outlined in the algorithm below to solve Black-box PMAP. The algorithm uses the notion of \emph{cut} of a matrix: A set $X \subset [n]$ of size between $2$ and $n-2$ is a cut of a matrix $A \in \F^{n \times n}$ if the ranks of the submatrices $A[X, \comp{X}]$ and $A[\comp{X}, X]$ are at most one. The \emph{principal minor oracle} of a matrix $A \in \F^{n \times n}$ takes input a nonempty set $S \subseteq [n]$ and outputs the principal minor of $A$ defined by the rows and columns indexed by $S$.

\begin{algorithm}[H]
\caption{Black-box PMAP}
\label{algo:Black-box PMAP}
\textbf{Input:} Black-box access to $f = \det(A + Y)$, where $Y = \diag(y_1, \ldots, y_n)$ and $A \in\F^{n\times n}$ is unknown.\\
\textbf{Output:} A matrix $B \in \F^{n\times n}$ such that $f = \det(B + Y)$. 
\begin{algorithmic}[1]
\State Reduce the problem to learning $M \in \F^{n \times n}$ given access to the principal minor oracle of $M$, where $M$ \underline{\emph{satisfies property $\prop$}} (see Algorithm \ref{algo:PMAP_redn_to_prop_R} in Section \ref{sec:reduction to matrices with prop R}).
\State Use the \underline{\emph{black-box cut finding}} algorithm (Algorithm \ref{algo:cut-detection-finding-algorithm} in Section \ref{sec: cut discovery}) to reduce the problem further to learning $N \in \F^{m \times m}$ from the principal minor oracle of $N$, where $N$ satisfies property $\prop$ and has no cut. ~(See Algorithm \ref{algo:reconstruction-with-Q} in Section \ref{sec: PMAP for matrices with prop R}.)
\State Solve PMAP for matrices with property $\prop$ and \underline{\emph{having no cuts}} (using Algorithm \ref{algo:reconstruction-with-no-cut-&-Q} in Section \ref{sec: PMAP for matrices with prop R}).
\end{algorithmic}
\end{algorithm}

\noindent The cut finding algorithm in Step 2 accesses only the principal minors of order at most $4$ and helps reduce PMAP for matrices with property $\prop$ to the no-cut case. So, the last two steps of the above algorithm, which run in deterministic polynomial time and constitute the main technical part of this work, give us a proof of Theorem \ref{thm:main-reconstruction-PM-with-Q-property}. We briefly discuss the steps below.  \\


\noindent{\textbf{Step $0$: Reducing $\probA$ to $\probB$}.} We have black-box access to $f = \det(A_0 + \sum_{i\in [n]} A_i y_i)$, where $A_0, \dots, A_n \in \F^{r \times r}$ and $A_1, \dots, A_n$ are rank one. We reduce the problem to learning $h = \det(A + Z)$, where $Z = \diag(z_1, \dots, z_n)$, in randomised polynomial-time as follows:

\begin{enumerate}
\item First, homogenise $f$ to get the degree $n$ multilinear polynomial $f' = t_1 t_2\dots t_nf(\frac{y_1}{t_1}, \frac{y_2}{t_2}, \dots \frac{y_n}{t_n})$.

\item Use the Isolation Lemma \cite[Section 3]{MVV87} to find a monomial (with non-zero coefficient) in $f'$. Label the variables in this monomial $z_1,z_2,\dots,z_n$, and the rest of the variables $\bar{z}_1,\bar{z}_2, \dots, \bar{z}_n$.

\item Let $\gamma$ be the coefficient of $z_1z_2 \cdots z_n$ in $f'$. We then show (in the proof of Theorem \ref{aa_thm_redn_rank1_to_pma}) that
\begin{align*}
    f' &= \gamma \det\left(\begin{bmatrix}I_n & I_n\end{bmatrix}\begin{bmatrix}Z & O_{n,n} \\ O_{n,n} & \bar{Z}\end{bmatrix} \begin{bmatrix}I_n \\ A\end{bmatrix}\right) = \gamma \det(Z + \bar{Z}A),
\end{align*}
for some $n \times n$ matrix $A$ over $\mathbb{F}$. Here, $Z = \diag(z_1,z_2, \dots, z_n)$ and $\bar{Z} = \diag(\bar{z}_1,\bar{z}_2, \dots, \bar{z}_n)$. Set $\bar{z}_1,\bar{z}_2, \dots, \bar{z}_n$ to $1$, scale by $\gamma^{-1}$, and learn the polynomial $h = \det(A + Z)$.

\item If $h = \det(A' + Z)$, then
$f' = \gamma \det\left(\begin{bmatrix}I_n & I_n\end{bmatrix}\begin{bmatrix}Z & O_{n,n} \\ O_{n,n} & \bar{Z}\end{bmatrix} \begin{bmatrix}I_n \\ A'\end{bmatrix}\right)\text{.}$
Set $t_1,\dots,t_n$ to $1$ and simplify to get $f = \det(B_0 + \sum_{i \in [n]} B_i y_i)$, where $B_0,\dots, B_n \in \F^{n \times n}$ and $B_1, \dots, B_n$ are rank $1$.
\end{enumerate}
The step involving the Isolation Lemma can be derandomised in quasi-polynomial time \cite{GT20}. \\




\noindent{\textbf{Step $1$: Reduction to PMAP for matrices with property $\prop$}.} We have black-box access to $f = \det(A+Y)$ where $A$ is an arbitrary $n \times n$ matrix and $Y = \diag(y_1, \dots, y_n)$. We reduce it to the case where $A$ satisfies property $\prop$ in three steps:

\begin{enumerate}
    \item We first demonstrate that for a diagonal matrix $D$ of randomly chosen entries, the maximal irreducible submatrices (see \cref{def:Reducible} and \cref{ob:irreducible-matrix-and-directed-graph}) of $B = (A+D)^{-1}$ satisfy property $\prop$ with high probability (\cref{lemma:random-diag-shift-satisfies-R-for-irred-comp}). This is where we require the field to be of size at least $n^6$. Using \cref{lem:bba_to_adj}, we obtain black-box access to $\det(B+Y)$.
    \item Since the irreducible blocks of $B$ have non-zero off-diagonal entries, we see that $i,j \in [n]$ are indices of the same irreducible submatrix if and only if $B[\{i,j\}]$ is irreducible (which can be checked using Claim \ref{clm:checking 2x2 reducibility}). We then compare enough pairs of indices $i,j$ in \cref{algo:PMAP_redn_finding_irred_blocks} to find the indices corresponding to the irreducible submatrices $B_1,B_2,\dots, B_s$ of $B$ and obtain their principal minor oracles (using Claim \ref{claim:bba-to-pm-oracle}).
    \item Each irreducible submatrix $B_i$ of $B$ satisfies property $\prop$, and learning them separately reduces to the principal minor assignment problem for matrices satisfying property $\prop$. Let $C_1 \PME B_1,C_2 \PME B_2, \dots, C_s \PME B_s$ be the matrices we learn. Then, by \cref{lem:redToIr}, there exists a permutation matrix $P$ such that $C = P^T \diag(C_1,\dots,C_s)P \PME B = (A+D)^{-1}$. The matrix $P$ is known to us due to our knowledge of the indices corresponding to the irreducible submatrices. We then show in \cref{lem:pme_A_from_pme_adj_A} that $C^{-1} - D \PME A$.
\end{enumerate}

\noindent This step can be derandomised in quasi-polynomial time using the strategy outlined in \cref{sec:deranomisation}. \\

\noindent{\textbf{Step $2$: Using black-box cut finding to reduce to the no-cut case.} Given an unknown matrix $M$ with property $\prop$ via its principal minor oracle, we show that if $M$ admits a cut, then the learning problem for $M$ can be reduced to learning a couple of its principal submatrices, whose total size is approximately that of $M$ (\cref{cl:reductionStep}). We apply this reduction recursively until the matrix has no cut, and then use the learning algorithm described in Step 3. To get these smaller principal submatrices, we have to find a cut of $M$ or any other matrix that is principal minor equivalent to it. Given an explicit matrix, \cite{Chatterjee25} reduces the problem of finding a cut to submodular minimization. It is not clear how to get access to the submodular function that they minimize using the principal minor oracle (or $\det(M+Y)$ evaluation oracle).

Given a matrix $M$ with non-zero off-diagonal entries, it is easy to see that a set $S$ is a cut if and only if for each $\{i,j\}\subset S$ and $\{k,l\}\subset \comp{S}$, $M[\{i,j,k,l\}]$ has $\{i,j\}$ (equivalently $\{k,l\}$) as a cut. With only principal minor access to $M$, we cannot decide directly whether $M[\{i,j,k,l\}]$ has such a cut. However, we can efficiently check whether there exists another $4\times4$ matrix, principal-minor equivalent to it, that has $\{i,j\}$ as a cut (\cref{obs:checkP}). We show that this suffices for matrices with property $\prop$, using some structural properties of cuts from \cite{Chatterjee25}. Precisely, a set $S$ is a cut of a matrix $M$ with property $\prop$, or of some matrix principal-minor equivalent to it, if and only if for all $\{i,j\}\subset S$ and $\{k,l\}\subset \comp{S}$, there exists a $4\times4$ matrix principal minor equivalent to $M[\{i,j,k,l\}]$ having the corresponding cut (\cref{lem:CutChecking}). We call such sets \emph{plausible sets} (\cref{def:plausibleSet}).

Using this idea, we give \cref{algo:cut-detection-finding-algorithm} to find a plausible set. To illustrate the idea, we describe a relatively simple algorithm for finding a cut of an explicitly given matrix. The same idea extends to finding a plausible set with slight modifications (see \cref{rem:cut-finding}). Given a matrix $M$, a set $\{i,j,k,l\}$ of four elements of the index set and a set $S$ with $i,j\in S$ and $k,l\in \comp{S}$, suppose we have to decide whether $S$ is a cut of $M$. First of all, $M[\{i,j,k,l\}]$ must have $\{i,j\}$ as a cut. If $S$ is such a cut, then an element $e$ cannot belong to $S$ if $M[\{i,e,k,l\}]$ does not have $\{i,e\}$ as a cut. Similarly, $e$ cannot be in $\comp{S}$ if $M[\{i,j,k,e\}]$ does not have cut $\{i,j\}$.  Also, for two elements $p,q$, it cannot be the case that $p\in S, q\in \comp{S}$ if $M[\{i,p,q,l\}]$ does not have $\{i,p\}$ as a cut. It can be shown that these local conditions suffice to test whether $S$ is a cut. \cref{lem:CorrectnessOfCutFinding} shows that similar local conditions are sufficient for a set to be a plausible set for matrices with property $\prop$. Such a set can be found by reducing the problem to a 2-SAT instance, where each variable corresponds to an index and each clause enforces one of the above conditions (see \cref{algo:cut-detection-finding-algorithm} for details). \\

\noindent{\textbf{Step $3$: Solving PMAP for matrices with property $\prop$ and having no cuts.} Suppose that $A\in\F^{n\times n}$ is a matrix satisfying property $\prop$ and has \emph{no} cut. We have access to the principal minors of $A$. We say two matrices $M, N\in\F^{n\times n}$ are \emph{diagonally similar}, if there exists an invertible diagonal matrix $D$ such that $N=D^{-1}MD$. Since all off-diagonal entries of $A$ are nonzero, there exists a (unique) matrix $M$ such that $M$ is diagonally similar to $A$ and $M[1, i]=1$ for all $i\in[n]\setminus\{1\}$. Moreover, $M\PME A$, that is, the corresponding principal minors of $M$ and $A$ are equal. Therefore, without loss generality, we may assume that $A[1, i]=1$ for all $i\in[n]\setminus \{1\}$. Our goal is to compute a matrix $B$ such that $B\PME A$. Moreover, for all $i\in[n]\setminus\{1\}$, the entry $B[1, i]$ will be $1$.

We construct $B$ iteratively. However, we do not follow an arbitrary sequence of indices in the iteration. Instead, we compute a sequence $\mathcal X=(i_n, i_{n-1}, \ldots, i_3, i_2)$ of distinct elements from $I\setminus\{1\}$ such that for all $j\in\{4,5,6, \ldots, n\}$, $A[I_j]$ has no cut, where $I_j=\{1, i_2,i_3, \ldots, i_j\}$. Since $A$ satisfies property $\prop$, such a sequence exists and can be computed in polynomial time (see~\cref{cor:no-cut-from-big-to-small-matrix} and~\cref{algo:finding-no-cut-sequence}). 

The usefulness of the sequence $\mathcal X$ can be explained as follows: At the $j$-th iteration, we construct the submatrix $B[I_j]$ such that $B[I_j]\PME A[I_j]$.  Since $A[I_j]$ has no cut, by~\cite[Theorem~1]{Loewy86} (or, see~\cref{lem:no-cut-to-PME}) and from additional structural properties on both $A$ and $B$, we have
\begin{equation}
\label{eqn:proof-idea-one}
B[I_j]=A[I_j], \text{ or } B[I_j] \text{ is diagonally similar to } A^T[I_j]
\end{equation}
This will be helpful to correctly run the iteration. In contrast, if $A[I_j]$ had a cut, by~\cite[Theorem~1.1]{Chatterjee25}, we obtain a weaker relation between $A[I_j]$ and $B[I_j]$, but it is not clear how to use that for our purpose. 

At the beginning of $j$-th iteration, the only unknown entries of $A[I_j]$ lie in the row and the column indexed by $i_j$. Note that $B[i_j, i_j]=A[i_j, i_j]$, hence this diagonal entry is known. Since $B[1, i_j]=1$, by the condition $B[\{1, i_j\}]\PME A[\{1, i_j\}]$, the entry $B[i_j, 1]$ is also determined.  Next, we iteratively compute $B[i_j, s]$ for all $s\in I_j\setminus\{1, i_j\}$. 

As before, we avoid following an arbitrary sequence. Instead, we find a sequence $(k_j, k_{j-1}, \ldots, k_3)$ of distinct elements from $I_j\setminus\{1, i_j\}$ such that for all $\ell\in\{4,5,6,\ldots, j\}$, the submatrix $A[L_\ell]$  has no cut, where $L_\ell=\{1, i_j, k_3, \ldots, k_\ell\}$. Since $A$ satisfies property $\prop$, the submatrix $A[L_\ell]$ also satisfies it (see~\cref{prop:matrix-property-large-to-small}). Moreover,  as $A[I_j]$ has no cut, such a sequence exists and can be computed in polynomial time (see~\cref{cor:no-cut-from-big-to-small-matrix} and~\cref{algo:finding-no-cut-sequence}).

Observe that $L_j=I_j$. We iteratively compute the matrix $B[L_j]$ such that $B[L_j]\PME A[L_j]$. At $\ell$-th iteration, the only unknown entries are $B[i_j, k_\ell]$ and $B[k_\ell, i_j]$.  Furthermore, once $B[i_j, k_\ell]$ is determined, $B[k_\ell, i_j]$ is uniquely fixed by the condition $A[\{i_j, k_\ell\}]\PME B[\{i_j, k_\ell\}]$. Note that $A[\{1, i_j, k_\ell\}]$ must be principal minor equivalent to $B[\{1, i_j, k_\ell\}]$. This restricts the value of $B[i_j, k_\ell]$  to at most two possible values, which corresponds to roots of a suitable quadratic equation (see~\cref{lem:reconstruction-size-three-matrix}). Additionally, we note the following:
\begin{enumerate}
\item Since $L_\ell-i_j$ is a subset of $I_{j-1}$, applying~\cref{eqn:proof-idea-one} for $(j-1)$-th iteration yields that $ B[L_\ell-i_j]=A[L_\ell-i_j]$, or it is diagonally similar to $A^T[L_\ell-i_j]$. 

\item Since $A[L_{\ell-1}]\PME B[L_{\ell-1}]$ and $A[L_{\ell-1}]$ has no cut, again by \cite[Theorem~1]{Loewy86} (or, see \cref{lem:no-cut-to-PME}), $ B[L_{\ell-1}]=A[L_{\ell-1}]$, or it is diagonally similar to $A^T[L_{\ell-1}]$.
\end{enumerate}
Combining the above facts, we obtain four possible cases, based on which we will compute $B[L_\ell]$ satisfying $B[L_\ell]\PME A[L_\ell]$. To identify a correct candidate for $B[L_\ell]$,~\cref{thm:main-characterization} plays a crucial role. For each possible choice of $B[L_\ell]$, we verify whether the corresponding principal minors of $A[L_\ell]$ and $B[L_\ell]$ up to order $4$ are equal. Since $A[L_\ell]$ satisfies property $\prop$, by~\cref{thm:main-characterization}, this verification is sufficient to ensure whether $A[L_\ell]\PME B[L_\ell]$. Thus, we can find a correct $B[L_\ell]$ in polynomial time. 

For more details, see~\cref{algo:reconstruction-with-no-cut-&-Q} and its correctness in~\cref{sec:reconstruction-with-no-cut-&-Q-correctness}.

\subsubsection{Proof sketch for Theorem \ref{thm:main-characterization}}
\label{sec:ideas characterization of matrices with prop Q}

Consider two matrices $A, B\in \F^{n\times n}$ that have equal corresponding principal minors of order at most 4. Suppose further that $A$ satisfies property~$\prop$. Our goal is to show that $A$ and $B$ have equal corresponding principal minors of all orders up to $n$. We begin with the case where $A$ has no cut, which constitutes the main technical component of the proof. 
\begin{enumerate}
	\item Suppose $A$ has no cut. We prove the principal minor equivalence of $A$ and $B$ by induction on~$n$. The base case $n=4$ is immediate. 
    
    For any $k \in [n]$, let $A_k$ and $B_k$ denote the principal submatrices of $A$ and $B$, respectively, obtained by deleting the $k$-th row and column. 
    Let $I \subseteq [n]$ be the set of indices for which $A_i$ is diagonally similar to $B_i$ or $B_i^T$.
    Following Lemma~\ref{lem:partial-DS-to-complete-DS} \cite[Lemma~1]{Hartfiell84}, it suffices to prove $|I| \geq 5$ as it would imply that $A$ is diagonally similar to $B$ or $B^T$. Consequently, $A$ and $B$ are principal minor equivalent \cite[Theorem 1]{Loewy86} (also see Lemma~\ref{lem:no-cut-to-PME}).
    
    For $n \geq 5$, we show that there exist at least three indices $i_1, i_2, i_3 \in [n]$ such that each $A_{i_j}$ has no cut. We crucially use property $\prop$ to show the existence of these indices (see~\cref{cor:no-cut-from-big-to-small-matrix}). By the induction hypothesis, $A_{i_j}$ is principal minor equivalent to $B_{i_j}$. Since $A_{i_j}$ has no cut, it again follows from \cite[Theorem~1]{Loewy86} (also see~\cref{lem:no-cut-to-PME}) that $A_{i_j}$ is diagonally similar to $B_{i_j}$ or its transpose. 
    
     We have $|I| \geq 3$ as $i_1, i_2, i_3 \in I$. We then use the principal minor equivalence of $A$ and $B$ up to order 4 and property $\prop$ of $A$ to prove that if $|I|$ were $3$ or $4$, then $A$ would necessarily have a cut (see~\cref{lem:small-matrices-DE-to-cut-4}, and~\cref{lem:small-matrices-DE-to-cut-3}). It contradicts our assumption that $A$ has no cut. A similar result appears in \cite[Theorem~2 and Theorem~4]{Loewy86} with the assumption that $A$ and $B$ are dense and principal minor equivalent of all orders. On the other hand, we show the result still holds even if we assume the matrices to be principal minor equivalent up to only order $4$ and additionally $A$ satisfies rank-one extension property.

	\item If $A$ has a cut, we consider a minimal cut $S$ of $A$. We first show that either $S$ is also a cut of $B$ or there is another matrix $\widetilde{B}$ principal minor equivalent to $B$ such that $S$ is a cut of $\widetilde{B}$ (see \cref{lem:Cut-and-PME-upto-4}, \cref{lem:common-cut}). Some weaker versions of~\cref{lem:Cut-and-PME-upto-4} and~\cref{lem:common-cut} appear in \cite[Lemma~3.5 and Lemma~3.6]{Chatterjee25}. We observe that almost the same proofs even work for our case. The argument then relies on the decomposition lemma (see \cref{lem:cut-decomposition}) that asserts the equivalence of the following two statements:
\begin{enumerate}
\item $A$ is principal minor equivalent to $B$.
\item For some common cut $S$ of $A$ and $B$,  $A[S+t]$ is principal minor equivalent to $B[S+t]$ and $A[\comp{S}+s]$ is principal minor equivalent to $B[\comp{S}+s]$ for any $s\in S$, and $t\in \comp{S}$.
\end{enumerate}
We use ideas from the proof of \cite[Lemma~2.12]{Chatterjee25} to prove the decomposition lemma. With the above setup, the inductive proof proceeds where the base case is either $n=4$, or when $A$ has \emph{no cut}.
\end{enumerate}

\subsubsection{Proof sketch for Theorem \ref{thm:main-NC-algorithm}} \label{sec:ideas thm NC}

To design an NC algorithm for principal minor equivalence, we first consider the special case where one of the matrices satisfies property $\prop$. In this case, \cref{thm:main-characterization} immediately yields an NC algorithm, since we can, in parallel, verify whether $\det(A[S])=\det(B[S])$ for all $S\subseteq [n]$ of size at most $4$ to determine the principal minor equivalence of $A$ and $B$. To solve the general case, we provide an NC reduction to the above special case through the following steps. 

\begin{enumerate}
     \item The matrices $A$ and $B$ are principal minor equivalent if and only if their maximal irreducible submatrices induce the same partition of the index set $[n]$, and the corresponding principal submatrices are themselves principal minor equivalent (see~\cref{lem:redToIr}). The maximal irreducible submatrices correspond to the strongly connected components of a graph associated with the input matrix (see~\cref{ob:irreducible-matrix-and-directed-graph}), which can be computed in NC using standard parallel algorithms for finding strongly connected components~\cite{Gazit88, Cole89}.

     \item Let $M^{\mathrm{adj}}$ denote the \emph{adjugate} of a matrix $M$ and $Y$ be an $n\times n$ diagonal matrix with distinct indeterminates on its diagonal. Then, $A$ is principal minor equivalent to $B$ if and only if $\adj{A+Y}$ and $\adj{B+Y}$ are principal minor equivalent (\cite[Lemma~4]{Hartfiell84}). Furthermore, we show that $\adj{A+Y}$ satisfies property $\prop$ (see~\cref{clm:AplusYsatisfy}). 

    Our goal is now to compute, in NC, a diagonal matrix $D$ such that $\adj{A+D}$ satisfies property $\prop$. Consequently, it reduces the principal minor equivalence testing of two irreducible matrices $A$ and $B$ to testing principal minor equivalence of $\adj{A+D}$ and $\adj{B+D}$. However, we do not know of an NC procedure to compute such $D$ directly. Instead, we first compute, in NC, a diagonal matrix $D_1$ such that all entries of  $A_1 = \adj{A+D_1}$ are nonzero ~(\cref{lem:ShiftAdjugatePreservesPME}). Next we compute another diagonal matrix $D_2$ in NC such that $A_2 = \adj{A_1+D_2}$ satisfies property~$\prop$ (\cref{lem:findDwithQ}). 
    Since the determinant of a matrix can be computed in NC~\cite{Csanky76, Borodin83}, all these matrices can be computed in NC. 
\end{enumerate} 
\subsection{Comparison with previous works} \label{sec:comparison}
In this section, we provide comparisons with known results that are directly relevant to our work. 

Randomized polynomial-time learning algorithms for ROFs were given in \cite{HancockH91, BshoutyHH95}. These algorithms were derandomized in polynomial-time in \cite{ShpilkaV14, MinahanV18} using hitting-sets for ROFs. Similarly, the randomized polynomial-time learning algorithm for ROABPs (with known variable ordering) in \cite{BeimelBBKV00, KlivansS06} was derandomized in quasi-polynomial time in \cite{ForbesS13} using hitting-sets for ROABPs. If the variable ordering in unknown, then ROABP learning is hard \cite{BDGT24}. The class of RODs is more powerful than ROFs (by a reduction in \cite{Valiant79a}). It is easy to show that the degree-$3$ determinant polynomial, which has an ROD of size $3$, cannot be computed by an ROF. But, RODs do not capture ROABPs -- the elementary symmetric polynomial, which has an ROABP of polynomial size, cannot be expressed as an ROD \cite{AravindJ15}. On the other hand, any ROABP that computes the degree-$n$ determinant polynomial has size $2^{\Omega(n)}$ \cite{Nisan91}. So, the learning algorithms for ROFs and ROABPs do not imply that RODs are efficiently learnable. Furthermore, the partial derivatives and evaluation dimension based techniques used for learning ROFs and ROABPs do not seem effective for learning RODs. However, as is the case for ROF and ROABP learning, the hitting-set for RODs in \cite{DBLP:conf/stoc/GurjarT17} can be used to derandomize our ROD learning algorithm in quasi-polynomial time (see Section \ref{sec:deranomisation}).

The PMAP problem was known to have a polynomial-time solution for the following classes of matrices: symmetric matrices \cite{Oeding2011-jo, RisingKuleszaTaskar2015}, magnitude-symmetric matrices ($|A(i,j)| = |A(j,i)|$)~\cite{Urschel2017,BrunelUrschel2024}, and ODF (off-diagonal full) matrices~\cite{GriffinTsatsomeros2006}.    
%
The property of symmetry or magnitude-symmetry is crucially used in the first two algorithms, and thus, their approaches do not seem directly applicable for the non-symmetric case.
The ODF property, introduced in \cite{GriffinTsatsomeros2006}, involves a technical condition which states that at every step of their algorithm, the solution should be unique.
Though this condition is true generically, it does not seem to correspond to any structural property of the given matrix that will gurantee an efficient solution for PMAP. 
On the other hand, property $\prop$, for which we solve PMAP (\cref{thm:main-reconstruction-PM-with-Q-property}) is a generic structural condition on the matrix.
Also note that by definition, ODF matrices cannot have a cut, while property~$\prop$ does not have such a restriction. 
Moreover, we show that blackbox PMAP for arbitrary matrices reduces to PMAP for matrices with property~$\prop$. 
It is not clear if there is such a reduction for ODF matrices.

As mentioned in Section~\ref{sec:results}, identifying property~$\prop$ and exploiting it for PMAP is one of our key technical ideas. 
Though property~$\prop$ does not appear explicitly in the literature, its definition is inspired by the proofs 
found in the works of Hartfiel-Loewy~\cite[Theorem 2]{Hartfiell84} and
Loewy~\cite[Lemma 6]{Loewy86}. 
The same is true for the result that for any irreducible matrix $A$, $(A+D)^{-1}$ satisfies property~$\prop$ for a random diagonal matrix $D$ (\cref{lemma:random-diag-shift-satisfies-R}).

Our solution for PMAP builds on  the fact that for matrices with property~$\prop$, PME up to order 4 implies PME (Theorem~\ref{thm:main-characterization}). 
Similar results were known for other classes of matrices, namely, skew-symmetric matrices with some other conditions~\cite{BC16, BOUSSAIRI201547}, generalized tournament matrices~\cite{BoussairiChaichaaCherguiLakhlifi2021}, and magnitude symmetric matrices~\cite{BrunelUrschel2024}. 
In contrast to our algebraic approach, most of these results rely fundamentally on graph-theoretic techniques.

These results also imply  efficient parallel algorithms for testing PME for these classes: in parallel, one can test equality of all corresponding principal minors up to order 3 or 4.
For the general case, no such parallel algorithms were known. 
The only known efficient algorithm for PME testing~\cite{Chatterjee25} involves finding a sequence of cut-transpose operations, which does not seem parallelizable. 
Here we make use of property~$\prop$ and Theorem~\ref{thm:main-characterization} to get a parallel algorithm for PME testing for two arbitrary matrices (\cref{thm:main-NC-algorithm}). 

\subsection{Other related works} \label{sec:related works}
Proper learning is known for a few other classes of polynomials. These classes are defined by restricted depth three and depth four arithmetic circuits having bounded top fan-in \cite{Shpilka09a, KarninS09, GuptaKL12, Sinha16, BSV20, BSV21}. Like RODs, these models are not universal. The learning algorithm for ROABPs implies improper learning algorithms for two important subclasses of depth three circuits, namely depth three powering circuits (symmetric tensors) and set-multilinear depth three circuits (tensors). However, learning general depth three or depth four circuits (even improperly) in the worst case is presumably hard (see the discussions about hardness of learning in \cite{KayalS19, GKS20}). Consequently, the complexity of learning has been investigated for various circuit models, including depth three and depth four circuits, in the average-case \cite{GuptaKL11, GKQ14, KayalNS19, KayalS19, GKS20, BGKS22}. Note that a natural way to define an average-case version of the PMAP problem is to let the input matrix be chosen randomly. Since a random matrix satisfies property $\prop$ with high probability (see the second remark after Theorem \ref{thm:main-reconstruction-PM-with-Q-property}), we have a deterministic polynomial-time algorithm for this average-case version of PMAP (by Theorem \ref{thm:main-reconstruction-PM-with-Q-property}).
\section{Preliminaries}
We use $[n]$ to denote the set of positive integers $\{1,2,3,\ldots, n\}$. For a set $S$ and a subset $T$ of $S$, $\comp{T}$ denotes the complement of $T$, i.e., $\comp{T}=S\setminus T$. For a set $S$ and an element $s$, we use $S+s$ and $S-s$ to denote the sets $S\cup\{s\}$ and $S\setminus\{s\}$, respectively.

For an $n\times n$ matrix $A$ and subsets $S, T\subseteq [n]$, we use $A[S,T]$ to denote the submatrix of $A$ whose rows are indexed by $S$ and columns are indexed by $T$. For $S=T$, we use $A[S]$ to denote $A[S, S]$. For a square matrix $A$, $A^{\mathrm{adj}}$ denotes the adjugate matrix of $A$. For a vector $(a_1,a_2,\ldots, a_n)$, $\diag(a_1, a_2, \ldots, a_n)$ denotes the $n\times n$ diagonal matrix $D$ such that for $i\in[n]$, $D[i, i]=a_i$. $I_n$ denotes the $n \times n$ identity matrix and $0_{m,n}$ denotes the $m \times n$ zero matrix.

\subsection{Reducible and Irreducible matrix}
\begin{definition}[Reducible and Irreducible matrix] 
\label{def:Reducible}
A matrix is called \emph{reducible} if it can be written as a block upper triangular matrix after permuting the rows and the corresponding columns. A matrix that is not reducible is called \emph{irreducible}. 

Equivalently, if we replace the nonzero off-diagonal entries with one and the diagonal entries with zero, then a reducible matrix corresponds to the adjacency matrix of a directed graph having more than one strongly connected component.
\end{definition}
From the above definition, it is easy to see that any matrix $A$ with all nonzero off-diagonal entries is an irreducible matrix. The above definition directly gives us the following observation.
\begin{observation}
\label{ob:irreducible-matrix-and-directed-graph}
Let $A$ be an $n\times n$ matrix over a field $\F$ such that the row and columns of $A$ are indexed by $[n]$. Let $G_A$ be a directed graph defined as follows: the vertex set in $[n]$, and a tuple $(i,j)$ is an edge of $G_A$ if and only if $i\neq j$ and $A[i,j]\neq 0$. Let $I_1, I_2,\ldots, I_s$ be the strongly connected components of $A$. Then, after permuting the rows and the corresponding columns, the matrix $A$ can be made a block upper triangular matrix, and the diagonal blocks $A[I_1], A[I_2],\ldots, A[I_s]$ are irreducible matrices.
\end{observation}

For two reducible matrices $A$ and $B$, the next lemma helps to reduce the testing of whether $A\PME B$ to multiple instances of testing whether two irreducible matrices have the same corresponding principal minors. The following lemma is a direct consequence of~\cite[Corollary~5.4]{Ahmadieh23}.
\begin{lemma}
\label{lem:redToIr}
Let $A$ and $B$ be two $n\times n$ matrices over a field $\F$. Suppose that after permuting the rows and the corresponding columns, $A$ can be written as a block upper triangular matrix with $s$ diagonal blocks $A_1,A_2,\dots,A_s$ where each $A_i$ is irreducible and the rows and columns of $A_i$ are indexed by set $T_i\subseteq [n]$. Then, $A\PME B$ if and only if the following holds:
\begin{enumerate}
    \item After permuting some rows and the corresponding columns, $B$ can be written as a block upper triangular matrix with $s$ diagonal blocks $B_1,B_2,\dots,B_s$ such that each $B_i$ is irreducible and the rows and columns of $B_i$ are indexed by the same sets $T_i$.
    \item For each $i\in [s]$, $A_i\PME B_i$.
\end{enumerate}
\end{lemma}

\subsection{Diagonal similarity}
\label{sec:diagonal-similarity}
Suppose that $A$ and $B$ are $n\times n$ matrices over $\F$. We say that $A$ is \emph{diagonally similar} to $B$, denoted by $A\DS B$, if there exists an invertible diagonal matrix $D$ such that $B=D^{-1}AD$. If no such matrix exists, we use $A\nDS B$ to denote that $A$ is \emph{not} diagonally similar to $B$.  We say that $A$ is \emph{diagonally equivalent} to $B$, denoted by $A\DE B$, if $A\DS B$ or $A\DS B^T$. For a matrix with nonzero off-diagonal entries, we make the following observation.
\begin{observation}
\label{obs:uniqueness-of-DS}
Let $A\in\F^{n\times n}$  with nonzero off-diagonal entries. Then, for every $i\in[n]$, there exists a unique $n\times n$ matrix $M_i$ such that $$M_i\DS A \text{ and } M_i[i, j]=1\ \ \forall\, j\in \nsubset{i}.$$ 
\end{observation}

For a fixed $i\in[n]$, the matrix $M_i$ in the above observation can be identified as a \emph{canonical representative} of all matrices that are diagonally similar to $A$. We next define two subsets $\calD_1(A, B)$ and $\calD_2(A, B)$ of $[n]$.
\begin{align}
\calD_1(A,B) &:= \{i\in[n]\,\mid\, A[\nsubset{i}] \DS B[\nsubset{i}]\}\label{eqn:ds-one}\\
\calD_2(A,B) &:= \{i\in[n]\,\mid\, A[\nsubset{i}]\DS B[\nsubset{i}]^T\}\label{eqn:ds-two}
\end{align}
From these two subsets, we have the following result.
\begin{lemma}[\cite{Hartfiell84}, Lemma~1]
\label{lem:partial-DS-to-complete-DS}
Let $n\geq 4$ be an integer. Let $A$ and $B$ be two $n\times n$ matrices over a field $\F$. Suppose that all the off-diagonal entries of $A$ are nonzero. Then, 
\begin{enumerate}
\item if $|\calD_1(A,B)|\geq 3$ then $A\DS B$.
\item if $|\calD_2(A,B)|\geq 3$ then $A\DS B^T$.
\item if $|\calD_1(A,B)\cup\calD_2(A,B)|\geq 5$ then $A\DE B$
\end{enumerate}
\end{lemma}

\subsection{Principal minor equivalence}
Suppose that $A$ and $B$ are two $n\times n$ matrices over $\F$. Then, the matrices $A$ and $B$ are said to be \emph{principal minor equivalent} if all the corresponding principal minors of $A$ and $B$ are equal, i.e., for all $S\subseteq [n]$, $\det A[S]=\det B[S]$. We use $A\PME B$ to denote that $A$ is principal minor equivalent to $B$, and $\kPME{A}{B}{k}$ to denote that all the corresponding principal minors of $A$ and $B$ of order \emph{at most} $k$ are equal, that is, for all $S\subseteq [n]$ with $|S|\leq k$, $\det A[S]=\det B[S]$. By $\newPME{A}{\geq}{k}{B}$, we denote that all the corresponding principal minors of $A$ and $B$ of order \emph{at least} $k$ are equal, that is, for all $S\subseteq [n]$ with $|S|\geq k$, $\det A[S]=\det B[S]$.

We now discuss the relationship between diagonal equivalence and principal minor equivalence. To begin, we state the following observation.
\begin{observation}
\label{obs:DE-to-PME}
Let $A$ and $B$ be an $n\times n$ matrices over $\F$ such that $A\DE B$. Then, $A\PME B$. 
\end{observation}

We next consider the converse of the above observation. Hartfiel and Loewy~\cite[Theorem~3]{Hartfiell84} showed that if $A$ and $B$ are irreducible matrices of size $n=2$ or $3$, then $A\PME B$ implies that $A\DE B$. Subsequently, Loewy~\cite[Theorem~1]{Loewy86} proved that for irreducible matrices $A$ and $B$ with $A$ having no cut, $A\PME B$ implies that $A\DE B$. Combining these results, we obtain the following lemma.
\begin{lemma}
\label{lem:no-cut-to-PME}
Let $A$ and $B$ be $n\times n$ matrices over $\F$ such that $A$ is irreducible and $A\PME B$. Then, 
\begin{enumerate}
\item if $n=2$ or $3$, then $A\DE B$
\item if $n\geq 4$ and $A$ has no cut, then $A\DE B$.
\end{enumerate}
\end{lemma}

\subsection{Cut of a matrix}
\begin{definition}[Cut of a matrix]
\label{def:cut-of-matrix}
Let $A$ be an $n\times n$ matrix over $\F$ with $n\geq 4$. A subset $X$ of $[n]$ is called a \emph{cut} of the matrix $A$, if $2\leq |X|\leq n-2$, and the rank of the submatrices $A[X, \comp{X}]$ and $A[\comp{X}, X]$ are at most one.

In particular, if $A$ is an irreducible matrix, then $\rank A[X, \comp{X}]=\rank A[\comp{X}, X ]=1$.
\end{definition}

The following lemma demonstrates how to derive new cuts from an existing collection of cuts.
\begin{lemma}
\label{prop:old-to-new-cut}
Let $A$ be an $n\times n$ matrix over a field $\F$ with $n\geq 2$ and all the off-diagonal entries of $A$ are nonzero. Let $X_1$ and $X_2$ be two subsets of $[n]$ such that both $X_1\cap X_2$ and $\comp{X}_1\cap \comp{X}_2$ are nonempty. Let $\rank A[X_1, \comp{X}_1]=\rank A[X_2, \comp{X}_2]=1$. Then, $$\rank A[X_1\cap X_2, \comp{X}_1\cup \comp{X}_2]=\rank A[X_1\cup X_2, \comp{X}_1\cap \comp{X}_2]=1.$$
\end{lemma}
\begin{proof}
From the lemma hypothesis, we have $$\rank A[X_1\cap X_2, \comp{X}_1]=\rank A[X_1\cap X_2, \comp{X}_2]=1$$
Let $k\in \comp{X}_1\cap\comp{X}_2$. Then, the above condition implies that for all $j\in\comp{X}_1\cup\comp{X}_2$, the column $A[X_1\cap X_2, j]$ is a scalar multiple of $A[X_1\cap X_2, k]$; that is, $$A[X_1\cap X_2, j]=cA[X_1\cap X_2, k] \text{ for some }c\in\F.$$ Hence, it follows that $\rank A[X_1\cap X_2, \comp{X}_1\cup\comp{X}_2]=1$. A similar argument shows that $\rank A[X_1\cup X_2, \comp{X}_1\cap \comp{X}_2]=1.$
\end{proof}

Suppose that $A$ is an $n\times n$ matrix over $\F$. Then, a family $\calA(A)$ of ordered pair of disjoint subsets of $[n]$ is defined as follows
\begin{equation}
\label{eqn:cut-of-matrix-one}
\calA(A):=\left\{(S, \comp{S})\,\mid\, S\subseteq [n],\ 2\leq|S|\leq n-2,\ \rank A[S,\comp{S}]=1\right\}
\end{equation}
We next consider the following two subsets of $[n]$.
\begin{align}
\calB_1(A) &:= \left\{i\in[n] \,\mid\, \exists\, X\subseteq \nsubset{i} \text{ s.t. } (X+i, \comp X)\in \calA(A) \text{ and } (\comp X+i, X)\in\calA(A)\right\}\label{eqn:cut-of-matrix-two}\\
\calB_2(A) &:= \left\{i\in[n] \,\mid\, \exists\, X\subseteq \nsubset{i} \text{ s.t. } (X, \comp X+i)\in \calA(A) \text{ and } (\comp X, X+i)\in\calA(A)\right\}\label{eqn:cut-of-matrix-three}
\end{align}
Note that in the above equations, $\comp{X}=[n]\setminus (X+i)$. From the above definitions, Loewy~\cite{Loewy86} showed the following lemma:
\begin{lemma}[\cite{Loewy86}, Theorem~A.1]
\label{lem:no-cut-property}
Let $A$ be an $n\times n$ matrix over $\F$ such that $A$ is irreducible and has no cut. Let $\calB_1(A)$ and $\calB_2(A)$ be defined as~\cref{eqn:cut-of-matrix-two} and~\cref{eqn:cut-of-matrix-three}, respectively. Then, $$|\calB_1(A)\cup\calB_2(A)|\leq n-3.$$
\end{lemma}

\subsection{Cut-transpose operation} \label{sec: cut transpose}
Motivated by~\cite[Lemma~4.5]{Ahmadieh23}, \cite{Chatterjee25} introduced the concept of the \emph{cut-transpose operation} and provided a characterization of principal minor equivalence using this operation. Below, we describe the definition and some results that will be useful in our setting. For further details, see~\cite[Section~2.5]{Chatterjee25}.   

\begin{definition}[\cite{Chatterjee25}, Definition~2.10]
\label{def:cut-transpose-operation}
Let $A$ be an $n\times n$ irreducible matrix, and $X\subseteq [n]$ be a cut of $A$. Let  $q, u\in\F^{|\comp{X}|}$ such that $q^T$ is the first nonzero row of $A[X,\comp{X}]$, and $u$ is the first nonzero column of $A[\comp{X},X]$. Let $p,v\in\F^{|X|}$  such that $A[X,\comp{X}]=p\cdot q^T$ and $A[\comp{X},X]=u\cdot v^T$. That is,
\[ 
A=\begin{pmatrix}
A[X]        & p\cdot q^T\\
&\\
u\cdot v^T  & A[\comp{X}]
\end{pmatrix}
\] 
Then, the \emph{cut-transpose operation on $A$ with respect to $X$}, denoted by $\tw(A,X)$, transforms $A$ to the following matrix:
\[
\tw(A,X) =
\begin{pmatrix}
A[X]          & p\cdot u^T\\
&\\
q\cdot v^T  & A[\comp{X}]^T
\end{pmatrix}
 \] 
\end{definition}
Like diagonal equivalence (see~\cref{sec:diagonal-similarity}), the principal minors of a matrix preserved under the cut-transpose operation.
\begin{lemma}
\label{lem:PME-under-cut-transpose}
Let $A$ be an $n\times n$ irreducible matrix over a field $\F$. Let $X\subseteq [n]$ be a cut of the matrix $A$. Then, $A\PME \tw(A, X)$.
\end{lemma}
For a proof see~\cite[Appendix~A]{Chatterjee25-full-version} (a full version of~\cite{Chatterjee25}). The next lemma describes a relationship between the cuts of two matrices whose corresponding principal minors up to order $4$ are equal.
\begin{lemma}
\label{lem:Cut-and-PME-upto-4}
Let $A$ and $B$ be two $n\times n$ matrices with nonzero off-diagonal entries such that $\kPME{A}{B}{4}$ and $S=\{s_1, s_2\}$ is a cut of $A$. Then either $S$ is a cut of $B$, or for each $i\in\{1,2\}$, the set $X_i$, defined as $$X_i=\{t\in\comp{S}\,\mid\, A[S+t]\DS B[S+t]\}\cup\{s_i\},$$ is a cut of $B$ and $S$ is a cut in $\tw(B, X_i)$.
\end{lemma}
A similar result was shown in~\cite[Lemma~3.5]{Chatterjee25}, with the key difference that their assumption is stronger: they require $A\PME B$, whereas the lemma above only assumes $\kPME{A}{B}{4}$. Nonetheless, the proof remains essentially the same. A close reading of their argument reveals that only the equality of principal minors of order at most 4 is used. For full details, refer to the proof of \cite[Lemma~3.5]{Chatterjee25}.

\subsection{Black-box evaluation with inverted inputs} \label{sec:bb with inverted inputs}

In \cref{aa_thm_redn_rank1_to_pma} and \cref{lem:bba_to_adj}, we have black-box access to polynomials such that some of the inputs to the black-box are inverted. As an example, given black-box access to a degree $d$ polynomial $f(y)$, consider the black-box for the polynomial $g(y,z) = z^d f(\frac{y}{z})$ constructed using the black-box for $f$. The inverse of $z$ is an input to this black-box for $g$. Since $g$ is a polynomial that can be evaluated at all points, we need to be able to evaluate the black-box for $g$ at points where $z$ is zero.

To generalise, suppose we have a degree $d$ polynomial $f(y_1,\dots,y_m,z_1,\dots,z_n)$ such that the inverses of $z_1,\dots,z_n$ are inputs to a black-box. Suppose further that we are trying to evaluate $f$ at a point where $k \leq n$ of the $z_i$ variables are set to zero. Assume without loss of generality the point is $(a_1,a_2,\dots,a_m,b_1,b_2,\dots,b_{n-k},0,0,\dots,0)$. We can set the variables $z_{n-k+1}, z_{n-k+2},\dots,z_n$ to a new variable $t$ to get
\[f(a_1,a_2,\dots,a_m,b_1,b_2,\dots,b_{n-k},t,t,\dots,t) = \sum_{i=0}^{d} c_i t^i\text{.}\]

Let $g(t) = f(a_1,\dots,a_m,b_1,\dots,b_{n-k},t,\dots,t)$. Note that $f(a_1,\dots,a_m,b_1,\dots,b_{n-k},0,\dots,0) = g(0) = c_0$ is the value we need to compute. If $\card{\F} > d+1$, we can find $d+1$ distinct non-zero values $\beta_0,\dots,\beta_{d}$ and use interpolation to compute $c_0$:

\[\begin{pmatrix}g(\beta_0) \\ g(\beta_1) \\ \vdots \\ g(\beta_d)\end{pmatrix} = \begin{pmatrix}
    1 & \beta_0 & \beta_0^2 & \cdots & \beta_0^d \\
    1 & \beta_1 & \beta_1^2 & \cdots & \beta_1^d \\
    \vdots & \vdots & \vdots & \ddots & \vdots \\
    1 & \beta_d & \beta_d^2 & \cdots & \beta_d^d \\
\end{pmatrix}\begin{pmatrix} c_0 \\ c_1 \\ \vdots \\ c_d\end{pmatrix}\text{.}\]

On the right is a Vandermonde matrix. We can then express $c_0$ as a linear combination $\sum_{i = 0}^d \alpha_i g(\beta_i)$ for $\alpha_i$'s computable from $\beta_i$'s. This requires $d+1$ queries to the black-box for $f$.

\begin{observation}\label{obs:black-box-var-inversion}
    Suppose we have black-box access to a degree $d$ polynomial $f(y_1,\dots,y_n)$ where some of the inputs to the black-box are inverted variables. Then if $\card{\F} > d+1$, we can evaluate the polynomial at any point using $d+1$ queries to the black-box.
\end{observation}

\begin{remarks}
    \item For multilinear polynomials, $d \leq n$ and we would require at most $n+1$ queries.
    \item The values $\alpha_0,\dots,\alpha_d$ only depend on $\beta_0,\dots,\beta_{d}$ and thus can be computed in advance.
\end{remarks}
\section{Equivalence between $\probA$ and $\probB$} \label{sec:equiv bet probA and probB}

In this section, we show a randomised polynomial-time equivalence between learning read-once determinants and the black-box principal minor assignment problem. We first show show a randomised polynomial-time reduction from learning polynomials of the form $f(y_1,\dots, y_n) = \det(A_0 + \sum_{i \in [n]} A_i y_i)$ where $A_1, \dots, A_n$ are rank-one matrices to learning a polynomial of the form $h(y_1,\dots,y_n) = \det(A+Y)$ where $Y = \diag(y_1,\dots,y_n)$. Observe that for the polynomial $f$, if $A_i = u_iv_i^T$ for some vectors $u_i$ and $v_i$, then $f(y_1,\dots,y_n) = \det(UYV^T + A_0)$ where $u_i$ and $v_i$ are the columns of $U$ and $V$ respectively. In addition to using these two representations interchangeably, we shall also use the following lemma to represent $f$ as a read-once determinant:

\begin{lemma}[{\cite[Lemma 4.3]{GT20}}]\label[lemma]{aa_lemma_rank1_to_read_once} For matrices $U,V \in \F^{r\times n}$, $A_0 \in \F^{r\times r}$ and $Y = \diag(y_1,\dots,y_n)$,
    \[\det(UYV^T + A_0) = \det\begin{pmatrix}
        I_n & Y & 0_{n,r} \\
        0_{n,n} & I_n & V^T \\
        U & 0_{r,n} & A_0
    \end{pmatrix}\text{.}\]
\end{lemma}

We show that for $f(y_1,\dots, y_n) = \det(A_0 + \sum_{i \in [n]} A_i y_i)$, we can assume the matrices $A_i$ are no larger than $n \times n$.

\begin{lemma}\label[lemma]{aa_lemma_rank1_size_redn}
For a non-zero polynomial $f(y_1, \dots, y_n) = \det(UYV^T + A_0)$ with $U,V \in \F^{r \times n}$ and $A_0 \in \F^{r \times r}$, we can assume without loss of generality that $r \leq n$ and $U$ is full-rank.
\end{lemma}
\begin{proof}
Suppose $k \leq n$ is the rank of the matrix $U$. By \cref{aa_lemma_rank1_to_read_once},
\[
    f(y_1, \dots, y_n) = \det \begin{pmatrix}
        I_n & Y & 0_{n,r} \\
        0_{n,n} & I_n & V^T \\
        U & 0_{r,n} & A_0
    \end{pmatrix}\text{.}
\]

Performing row transformations on $U$ and absorbing resulting scalars into the last row,
\[
    f(y_1, \dots, y_n) = \det \begin{pmatrix}
        I_n & Y & 0_{n,r} \\
        0_{n,n} & I_n & V^T \\
        \hat{U} & 0_{k,n} & C_1 \\
        0_{r-k, n} & 0_{r-k, n} & C_2
    \end{pmatrix}\text{,}
\]
where $\hat{U} \in \F^{k \times n}$, $C_1 \in \F^{k \times r}$, $C_2 \in \F^{(r-k) \times r}$ and $\hat{U}$ is full-rank. For $f$ to be non-zero, $C_2$ must be full rank. Performing row and column transformations on $C_2$ and absorbing scalars into a column of $\hat{A_0}$, we have
\begin{align*}
    f(y_1, \dots, y_n) &= \det \begin{pmatrix}
        I_n & Y & 0_{n,k} & 0_{n,r-k} \\
        0_{n,n} & I_n & \hat{V}^T & \hat{V}_1^T \\
        \hat{U} & 0_{k,n} & \hat{A_0} & C_3 \\
        0_{r-k, n} & 0_{r-k, n} & 0_{r-k, k} & I_{r-k}
    \end{pmatrix} \\
    &= \det \begin{pmatrix}
        I_n & Y & 0_{n,k} \\
        0_{n,n} & I_n & \hat{V}^T \\
        \hat{U} & 0_{k,n} & \hat{A_0}
    \end{pmatrix} \\
    &= \det(\hat{U}Y\hat{V}^T + \hat{A_0})\text{,}
\end{align*}
where $\hat{V} \in \F^{k \times n}$ and $\hat{A_0} \in \F^{k \times k}$.
\end{proof}

The following lemma shall be helpful in converting a read-once determinant back into a symbolic determinant under rank 1 restriction.

\begin{lemma}\label[lemma]{aa_lemma_detsplit}
    If $A \in \F^{(n-r) \times r}, B \in \F^{(n-r) \times (n-r)}$ and $C \in \F^{r \times (n-r)}$, then \[\det\begin{pmatrix}A & B \\ I_r & C\end{pmatrix} = 
        (-1)^{r(n-r)}\det\left(\begin{pmatrix} A & B \end{pmatrix} \begin{pmatrix} -C \\ I_{n-r} \end{pmatrix}\right)\text{.}\]
\end{lemma}
\begin{proof}
    Moving $\begin{pmatrix}I_r & C\end{pmatrix}$ to the top (each of the $r$ rows being swapped $n-r$ times, totalling $r(n-r)$ swaps) and using block Gaussian elimination,
    \begin{align*}
        \det\begin{pmatrix}A & B \\ I_r & C\end{pmatrix} &= (-1)^{r(n-r)}\det\begin{pmatrix}I_r & C \\ A & B\end{pmatrix} \\
            &= (-1)^{r(n-r)}\det\begin{pmatrix}I_r & 0_{r,n-r} \\ A & B - AC\end{pmatrix} \\
            &= (-1)^{r(n-r)}\det(B - AC) \\
            &= (-1)^{r(n-r)}\det\left(\begin{pmatrix} A & B \end{pmatrix} \begin{pmatrix} -C \\ I_{n-r} \end{pmatrix}\right)\text{.}
    \end{align*}
\end{proof}

\subsection{$\probA$ to $\probB$}

We first proceed similarly to \cite[Theorem 4.3]{GT20} in homogenising the polynomial $f$. Then we use the Isolation Lemma \cite{MVV87} to reduce learning $f$ to the black-box principal minor assignment problem.

\begin{theorem}\label[theorem]{aa_thm_redn_rank1_to_pma}
    Let $\card{\F} > n+1$. Given black-box access to $f(y_1,\dots,y_n) = \det(A_0 + \sum_{i \in [n]} A_i y_i)$ where $A_0, \dots, A_n \in \F^{r \times r}$ and $A_1, \dots, A_n$ are rank 1, we can, in randomised polynomial-time, construct black-box access to a polynomial $h(z_1,\dots,z_n) = \det(A + Z)$ where $A \in \F^{n \times n}$ and $Z = \diag(z_1, \dots, z_n)$. Furthermore, given a matrix $A' \in \F^{n \times n}$ such that $\det(A + Z) = \det(A' + Z)$, we can recover matrices $B_0, \dots, B_n \in \F^{n \times n}$ such that $f(y_1,\dots,y_n) = \det(B_0 + \sum_{i \in [n]} B_iy_i)$ and $B_1, \dots, B_n$ are rank 1.
\end{theorem}

\begin{proof}
Let $U$ and $V \in \F^{r \times n}$ be the matrices such that $f(y_1, \dots, y_n) = \det(UYV^T + A_0)$ where $Y = \diag(y_1, \dots, y_n)$. We first homogenise the multilinear polynomial $f$ as
\[ f'(y_1, y_2, \dots, y_n, t_1, t_2, \dots, t_n) = t_1t_2\dots t_nf(\frac{y_1}{t_1}, \frac{y_2}{t_2}, \dots, \frac{y_n}{t_n})\text{.}\]

where $t_1,\dots,t_n$ are a fresh set of variables. Due to \cref{obs:black-box-var-inversion}, we can obtain black-box access to $f'$. By \cref{aa_lemma_rank1_to_read_once}, we have
\begin{align*}
    f'(y_1, \dots, y_n, t_1, \dots, t_n) &= \det(T)\det\begin{pmatrix}
        I_n & T^{-1}Y & 0_{n,r} \\
        0_{n,n} & I_{n} & V^T \\
        U & 0_{r,n} & A_0
    \end{pmatrix} \\
    &= \det\begin{pmatrix}
        T & Y & 0_{n,r} \\
        0_{n,n} & I_{n} & V^T \\
        U & 0_{r,n} & A_0
    \end{pmatrix}\text{,}
\end{align*}
where $T = \diag(t_1, t_2, \dots, t_n)$. Observe that the matrix $\begin{pmatrix}U & A_0\end{pmatrix}$ must be full-rank for $f'$, and by extension $f$, to be non-zero. We now perform elementary row transformations on the last $n+r$ rows, and permute the columns to get
\[f'(y_1, \dots, y_n, t_1, \dots, t_n) = \alpha\det\begin{pmatrix}
        \begin{array}{ccc} \leftarrow & L'W' & \rightarrow \end{array} \\
        \begin{array}{cc} I_{n+r} & H \end{array} \\
    \end{pmatrix}\text{,}\]
where, letting $P$ be the $(2n+r) \times (2n+r)$ permutation matrix corresponding to the column permutations performed,
\begin{itemize}
    \item $L'$ is the $n \times (2n+r)$ matrix $\begin{pmatrix}I_n & I_n & 0_{n,r}\end{pmatrix} P$.
    \item $W'$ is the $(2n+r) \times (2n+r)$ diagonal matrix $P^T\begin{pmatrix}
        T & 0_{n,n} & 0_{n,r} \\
        0_{n,n} & Y & 0_{n,r} \\
        0_{r,n} & 0_{r,n} & 0_{r,r} \\
    \end{pmatrix}P$.
    \item $H$ is the $(n+r) \times n$ matrix resulting from the transformations.
    \item $\alpha$ is the scalar resulting from the transformations.
\end{itemize}
Let $R' = \begin{pmatrix}-H^T & I_n\end{pmatrix}$. Observe that $R'$ has size $n \times (2n+r)$. Then by \cref{aa_lemma_detsplit},
\[f'(y_1, \dots, y_n, t_1, \dots, t_n) = (-1)^{n(n+r)}\alpha\det(L'W'R'^T)\text{.}\]

We can drop the columns of $L'$ and $R'$ corresponding to the zeroes in the diagonal of $W'$ to get \[f'(y_1, \dots, y_n, t_1, \dots, t_n) = \det(LWR^T)\] with $L$ and $R$ being the resulting $n \times 2n$ matrices and $W$ being the $2n \times 2n$ diagonal matrix of the $y_i$ and $t_i$ variables. The constant $(-1)^{n(n+r)}\alpha$ is absorbed into the matrix $R$.

\begin{observation}\label[observation]{aa_obs_redn_rank1_L_structure}
    $L$ is a column permutation of the matrix $\begin{pmatrix}I_n & I_n\end{pmatrix}$. For all $i$, the columns of $L$ corresponding to $y_i$ and $t_i$ are equal.
\end{observation}

Observe that $f$ is a homogeneous multilinear polynomial of degree $n$. We can use the Isolation Lemma \cite[Section 3]{MVV87} to isolate a monomial \cite[Section 4.1]{NSV94}. Let this monomial comprise variables $z_1, z_2, \dots z_n$ where $z_i = y_i \text{ or } t_i$. Similarly, let $\bar{z}_i = t_i$ if $z_i = y_i$, and vice-versa. By \cref{aa_obs_redn_rank1_L_structure}, the columns in $L$ corresponding to $z_i$s must form an identity matrix. This lets us write $f'$ as

\[f'(y_1, \dots, y_n, t_1, \dots, t_n) = \det\left(\begin{pmatrix}I_n & I_n\end{pmatrix}\begin{pmatrix}Z & 0_{n,n} \\ 0_{n,n} & \bar{Z}\end{pmatrix} \begin{pmatrix}R_1 \\ R_2\end{pmatrix}\right)\text{,}\]
where $Z = \diag(z_1, \dots, z_n)$, $\bar{Z} = \diag(\bar{z}_1, \dots, \bar{z}_n)$, $R_1$ are the rows of $R^T$ corresponding to $Z$, and $R_2$ the rest of the rows of $R^T$. Since the coefficient of $z_1\dots z_n$ is $\det(R_1)$ and is non-zero, we can multiply $f'$ on the right by $\det({R_1}^{-1})$, letting $A = R_2{R_1}^{-1}$, to get
\begin{align*}
    g(y_1, \dots, y_n, t_1, \dots, t_n) &= f'(y_1, \dots, y_n, t_1, \dots, t_n) \det({R_1}^{-1}) \\
        &= \det\left(\begin{pmatrix}I_n & I_n\end{pmatrix}\begin{pmatrix}Z & 0_{n,n} \\ 0_{n,n} & \bar{Z}\end{pmatrix} \begin{pmatrix}I_n \\ A\end{pmatrix}\right) \\
        &= \det\left(\begin{pmatrix}Z & \bar{Z}\end{pmatrix}\begin{pmatrix}I_n \\ A\end{pmatrix}\right) \\
        &= \det\left(Z + \bar{Z}A\right)\text{.}
\end{align*}

Let $h(z_1, \dots, z_n, \bar{z}_1, \dots, \bar{z}_n) = \bar{z}_1\bar{z}_2\dots\bar{z}_{n}g(z_1, z_2, \dots, z_n, \frac{1}{\bar{z}_1}, \frac{1}{\bar{z}_2}, \dots, \frac{1}{\bar{z}_n})$. We see
\begin{align*}
    h(z_1, z_2, \dots, z_n, \bar{z}_1, \dots, \bar{z}_{n}) &= \det(\bar{Z})\det(Z+\bar{Z}^{-1}A) \\
        &= \det(\bar{Z}Z + A)\text{.}
\end{align*}

\begin{observation}\label[observation]{aa_obs_redn_rank1_homogenisation_correspondence}
    In the monomials of $h$, $z_i$ appears if and only if $\bar{z}_i$ appears. This is evident from the fact that the $y$ variables arose out of the homogenisation of $f$; the variable $t_i$ appears in a monomial of $g$ if and only if $y_i$ does not.
\end{observation}

Suppose we have a matrix $A' \in \F^{n\times n}$ such that $\det(Z + A') = h(z_1, \dots, z_n, 1, 1, \dots, 1) = g(z_1, \dots, z_n, 1, 1, \dots, 1) = \det(Z + A)$. By \cref{aa_obs_redn_rank1_homogenisation_correspondence}, $h(z_1,\dots,z_n,\bar{z}_1,\dots,\bar{z}_n) = \det(\bar{Z}Z + A')$. This means
\begin{align*}
    g(y_1, \dots, y_n, t_1, \dots, t_n) &= \det\left(Z + \bar{Z}A'\right) \\
        &= \det\left(\begin{pmatrix}Z & \bar{Z}\end{pmatrix}\begin{pmatrix}I_n \\ A'\end{pmatrix}\right) \\
        &= \det\left(\begin{pmatrix}I_n & I_n\end{pmatrix}\begin{pmatrix}Z & 0_{n,n} \\ 0_{n,n} & \bar{Z}\end{pmatrix} \begin{pmatrix}I_n \\ A'\end{pmatrix}\right)\text{.}
\end{align*}

Let $l_{y_i}$ be the column of $\begin{pmatrix}I_n & I_n\end{pmatrix}$ corresponding to $y_i$, and $r_{y_i}^T$ the row of $\begin{pmatrix}I_n \\ A'\end{pmatrix}$ corresponding to $y_i$. Similarly let $l_{t_i}$ and $r_{t_i}$ be the same for $t_i$. We see \[g(y_1, \dots, y_n, t_1, \dots, t_n) = \det \left( \left(\sum_{i \in [n]} l_{y_i}r_{y_i}^T y_i\right) + \left(\sum_{i \in [n]} l_{t_i}r_{t_i}^T t_{i}\right)\right)\text{.}\] Setting all $t_i$ variables to $1$ and $C = \sum_{i \in [n]} l_{t_i}r_{t_i}^T$, we have \[f(y_1, \dots, y_n) = \det(R_1) \det \left( \left(\sum_{i \in [n]} l_{y_i}r_{y_i}^T y_i\right) + C\right)\text{.}\]

Setting $B_i = l_{y_i}r_{y_i}^T$ with their first rows scaled by $\det(R_1)$, and $B_0 = C$ with its first row scaled by $\det(R_1)$, we have $f(y_1,\dots,y_n) = \det(B_0 + \sum_{i \in [n]} B_iy_i)$. \\
\end{proof}

To summarise, given black-box access to $f(y_1,\dots,y_n)$, the reduction algorithm is as follows:

\begin{algorithm}[H]
\caption{Reduction of \probA\, to \probB}
\label{algo:reducing probA to probB}
\textbf{Input:} Black-box access to $f = \det( A_0 + \sum_{i \in [n]} A_i y_i)$ for some matrices $A_i \in\F^{r\times r}$ and $A_1,\dots, A_n$ are rank one.\\
\textbf{Output:} Matrices $B_0, \dots, B_n \in \F^{n\times n}$ with $B_1,\dots,B_n$ being rank-one such that $f = \det( B_0 + \sum_{i \in [n]} B_i y_i)$. \\
\textbf{Assumption:} Oracle access to learning algorithm for $\det(A + Z)$ where $Z = \diag(z_1,\dots,z_n)$.
\begin{algorithmic}[1]
\State Homogenise $f$ to get $f'(y_1,\dots,y_n,t_1,\dots,t_n) = t_1t_2\dots t_nf(\frac{y_1}{t_1}, \frac{y_2}{t_2}, \dots, \frac{y_n}{t_n})$.
\State Use the Isolation Lemma to isolate a monomial of $f'$. Let this monomial be $z_1\dots z_n$ with coefficient $\gamma$.
\State $h(z_1,\dots,z_n) = \frac{1}{\gamma} f'(z_1,\dots,z_n, 1, \dots, 1)$ will be the input to the black-box principal minor assignment problem.
\State Suppose $A'$ is the output of the black-box principal minor assignment problem. Then $f' = \gamma \det\left(\begin{pmatrix}
    I_n & I_n\end{pmatrix}\begin{pmatrix}Z & 0_{n,n} \\ 0_{n,n} & \bar{Z}\end{pmatrix} \begin{pmatrix}I_n \\ A'
\end{pmatrix}\right)$.
\State Set the $t_i$ variables to $1$ and simplify to recover $f = \det(B_0 + \sum_{i \in [n]} B_i y_i)$. Output $B_0, \dots, B_n$.
\end{algorithmic}
\end{algorithm}


\subsection{$\probB$ to $\probA$}

\begin{theorem}\label[observation]{aa_thm_redn_pma_to_rank1}
    Given black-box access to $h(y_1,\dots,y_n) = \det(A + Y)$ where $Y = \diag(y_1, \dots, y_n)$ and $A \in \F^{n \times n}$, suppose we have matrices $B_0, \dots, B_n \in\F^{r \times r}$, $B_1, \dots, B_n$ being rank 1, such that $\det(A + Y) = \det(B_0 + \sum_{i \in [n]} B_iy_i)$. Then we can, in deterministic polynomial-time, recover a matrix $A' \in \F^{n \times n}$ such that $h(y_1,\dots,y_n) = \det(A' + Y)$.
\end{theorem}

\begin{proof}

Let $U, V \in \F^{r \times n}$ be the matrices such that $h(y_1,\dots,y_n) = \det(A + Y) = \det(UYV^T + B_0)$. Due to \cref{aa_lemma_rank1_size_redn}, we can assume that $r \leq n$. By \cref{aa_lemma_rank1_to_read_once},
\begin{align*}
h(y_1,\dots,y_n) &= \det\begin{pmatrix}
        I_n & Y & 0_{n,r} \\
        0_{n,n} & I_{n} & V^T \\
        U & 0_{r,n} & B_0
    \end{pmatrix}\text{.}
\end{align*}

Observe that the coefficient of $y_1 \dots y_n$ is $(-1)^n\det\begin{pmatrix}0_{n,n} & V^T \\ U & B_0\end{pmatrix} = 1$. This would be zero if $r<n$, thus $r=n$. Then, $(-1)^n\det\begin{pmatrix}0_{n,n} & V^T \\ U & B_0\end{pmatrix} = \det(U)\det(V^T) = 1$. This implies that both $U$ and $V^T$ are invertible, with $\det(U^{-1})\det({(V^T)}^{-1}) = 1$. Thus,
\begin{equation*}
h(y_1,\dots,y_n) = \det(U^{-1})\det(UYV^T + B_0)\det({(V^T)}^{-1})
    = \det(Y + U^{-1}B_0{(V^T)}^{-1})\text{.}
\end{equation*}
\end{proof}
\section{$\probB$ to PMAP for matrices with Property $\prop$} \label{sec:reduction to matrices with prop R}

In this section, we show how $\probB$ reduces to PMAP for matrices that satisfy property $\prop$. We first demonstrate that for any matrix $A$, the irreducible blocks of the matrix $(A+D)^{-1}$ where $D$ is a diagonal matrix of randomly chosen field elements satisfy property $\prop$ with high probability and how we get black-box access to $\det((A + D)^{-1} + Y)$ from black-box access to $\det(A + Y)$. Then, we find out the indices corresponding to the irreducible blocks of $(A+D)^{-1}$. Since the irreducible blocks of $(A+D)^{-1}$ satisfy property $\prop$, learning them reduces to PMAP for matrices satisfying property $\prop$. Finally, from the matrices $C_1,\dots,C_k$ that are principal minor equivalent to the irreducible blocks of $(A+D)^{-1}$, we find a matrix $C \PME A$.

\begin{algorithm}[H]
\caption{Reducing $\probB$ to PMAP for matrices with property $\prop$}
\label{algo:PMAP_redn_to_prop_R}
\textbf{Input:} Black box access to $\det(A+Y)$ where $Y = \diag(y_1,\dots,y_n)$.\\
\textbf{Output:} A matrix $B \in \F^{n\times n}$ such that $B \PME A$. \\
\textbf{Assumption:} Access to oracle \Call{PMAP-Prop-R}{ $\PMoracle{M}$ } which solves PMAP given access to the principal minor oracle $\PMoracle{M}$ (See \cref{def:prinicpal-minor-oracle}) of a matrix $M$ that satisfies property $\prop$.
\begin{algorithmic}[1]
\State Using \cref{lem:bba_to_adj}, obtain black box access to $\det((A+D)^{-1} + Y)$.
\State Using \cref{algo:PMAP_redn_finding_irred_blocks}, find the indices $T_1,\dots,T_s$ that correspond to the irreducible blocks of $(A+D)^{-1}$.
\For{$i \in [s]$}
\State Let $C_i \gets$ \Call{PMAP-Prop-R}{ $\PMoracle{(A+D)^{-1}[T_i]}$ }.
\EndFor
\State Find permutation matrix $P$ such that $P^T \diag(Y[T_1],\dots,Y[T_s])P = Y$.
\State Let $C \gets P^T\diag(C_1,\dots,C_s)P$.
\State \Return $C^{-1} - D$.
\end{algorithmic}
\end{algorithm}

\subsection{Irreducible blocks of \texorpdfstring{$(A+D)^{-1}$}{(A+D)\^(-1)} satisfy property $\prop$}

To find a diagonal matrix $D$ such that each irreducible block of $(A+D)^{-1}$ satisfies property $\prop$, it suffices to find a $D$ for which $\det(A+D) \neq 0$ and each irreducible block of $\adj{A+D}$ satisfies property $\prop.$
First, we will show that for any irreducible matrix $A$, the matrix $\adj{A+Y}\in \mathbb{F}(y_1,y_2,\dots,y_n)^{n \times n}$ satisfies property $\prop$ where $Y = \diag(y_1, \ldots, y_n)$. The following result of Hartfiel and Loewy~\cite[Theorem~2]{Hartfiell84} will be useful.
\begin{theorem}[\cite{Hartfiell84}, Theorem~2]
\label{lem:Qpropertyhelpertheorm}
    Let $n\geq 2$ and $M$ be $n\times n$ matrix over a field $\mathbb{F}$. Let $Z=\diag(0,0,z_3,z_4,\dots,z_n)$ where $z_3, z_4, \ldots, z_n$ are independent indeterminates. Let $\det(M+Z)=0$. Then, there exists a partition $(X_1, X_2)$ of $\{3,4,\dots,n\}$ such that $$\rank(M[\{1,2\}\cup X_1, \{1,2\}\cup X_2])\leq 1.$$ 
\end{theorem}

In addition to the above theorem, we also require the following result. The proof of the following lemma directly follows from \cite[Claim A.1]{Chatterjee25-full-version}.
\begin{lemma}
\label{lem:AandAdjArank1}
Let $n\geq 4$ be a positive integer. Let $M\in \mathbb{F}^{n\times n}$ and $X\subseteq [n]$ such that $\rank( M[X,\comp{X}])=1$. Then, $\rank M^{\ad}[X,\comp{X}]=1$.
\end{lemma}

The following claim shows that $\adj{A+Y}$ satisfies property $\prop$.
\begin{claim} \label{clm:AplusYsatisfy}
Define $Y = \diag(y_1, \ldots, y_n)$. For any irreducible matrix $A\in \mathbb{F}^{n\times n}$, the matrix $\adj{A+Y}$ satisfies property $\prop$.
\end{claim}
\begin{proof}
The first condition of property $\prop$ follows from \cite[Theorem 1]{Hartfiell84}. For the second condition, we need to show that for any distinct four elements $i,j,k,\ell\in [n],$ $$\rank \adj{A+Y}[\{i,j\},\{k,\ell\}]=1\Longrightarrow \rank \adj{A+Y}[X,\comp{X}]=1$$ for some $X\subset [n]$ with $\{i,j\}\subseteq X$ and $\{k,\ell\}\subseteq \comp{X}$. Without loss of generality, we can assume that $i=1,j=2,k=3,\ell=4$. Note that $$\rank(\adj{A+Y}[\{1,2\},\{3,4\}])=1\Longrightarrow \det(\adj{A+Y}[\{1,2\},\{3,4\}])=0.$$ Now, applying Jacobi's identity~\cite[Page~21]{Gantmacher59}, we obtain $\det((A+Y)[\comp{\{3,4\}},\comp{\{1,2\} }])=0$. Observe that for a matrix $A=(a_{i,j})_{i,j\in[n]}$,
\[
(A+Y)[\comp{\{3,4\}},\comp{\{1,2\} }]= \begin{pmatrix} a_{1,3} & a_{1,4} & a_{1,5} & \cdots & a_{1,n}\\
    a_{2,3} & a_{2,4} & a_{2,5} & \cdots & a_{2,n}\\
    a_{5,3} & a_{5,4} & a_{5,5}+y_5 & \cdots & a_{5,n}\\
    \vdots & & & \ddots &\\
    a_{n,3} & a_{n,4} & a_{n,5} & \cdots & a_{n,n}+y_n \end{pmatrix}.\]
From the above, we obtain $$(A+Y)[\comp{\{3,4\}},\comp{\{1,2\} }]=A'+Y' \text{ where } A'= A[\comp{\{3,4\}},\comp{\{1,2\} }] \text{ and } Y'= \diag(0,0, y_5,\dots,y_n).$$ Since $\det(A'+Y')=0$, from \cref{lem:Qpropertyhelpertheorm}, there exists a partition $X_1,X_2$ of $\{5, 6,\ldots,n\}$ such that $$\rank A[X_1\cup\{1, 2\},X_2\cup\{3, 4\}]= 1.$$ Let $X=X_1\cup \{1,2\}$. Then $\comp{X}=X_2\cup \{3, 4\}$. Thus, we obtain $\rank (A+Y)[X,\comp{X}]=1$. From \cref{lem:AandAdjArank1}, $\rank \adj{A+Y}[X,\comp{X}]=1$. This completes the proof of the lemma.
\end{proof}


Next, we show the necessary properties of a diagonal matrix $D$ such that $\adj{A+D}$ satisfies property $\prop$.
\begin{claim}
\label{clm:satisfying-prop-Q}
Let $D\in \mathbb{F}^{n\times n}$ be a diagonal matrix such that $\adj{A+D}$ 
satisfies the following properties:
\begin{enumerate}
\item $A+D$ is invertible, i.e.\ $\det(A + D) \neq 0$.
\item $\adj{A+D}$  has nonzero off-diagonal entries, i.e.\ for each $1\leq i,j \leq n$, $\det ((A + D)[\comp{i}, \comp{j}]) \neq 0$.
\item For any four distinct elements $i,j,k,\ell\in[n]$,  
\begin{equation}\label{eq:diag-shift-property}
  \det\left(\adj{A+D}[\{i,j\}, \{k,\ell\}]\right)=0 \Longrightarrow \det\left(\adj{A+Y}[\{i,j\},\{k,\ell\}]\right)=0.   
\end{equation}
\end{enumerate}
Then, $\adj{A+D}$ satisfies property $\prop$.
\end{claim}
\begin{proof}
Let $i,j,k,\ell\in[n]$ be four distinct elements. Then, the third property implies that $$\det\left(\adj{A+D}[\{i,j\}, \{k,\ell\}]\right)=0\Longrightarrow \det\left(\adj{A+Y}[\{i,j\}, \{k,\ell\}]\right)=0.$$ From~\cref{clm:AplusYsatisfy}, we know that $\adj{A+Y}$ satisfies property $\prop$. Thus, $$\rank \adj{A+D}[\{i,j\},\{k,\ell\}]=1\Longrightarrow \rank \adj{A+Y}[X,\comp{X}]=1 \Longrightarrow \rank \adj{A+D}[X,\comp{X}]=1$$ for some $X\subset [n]$ with $i,j\in X$ and $k,\ell\in \comp{X}$. 
Hence, $\adj{A+D}$ will satisfy property $\prop.$
\end{proof}

The following result shows that $\adj{A+D}$ satisfies property $\prop$ where $D$ is a `random' diagonal matrix.

\begin{lemma}\label{lemma:random-diag-shift-satisfies-R}
    Let $\card{\F} \geq n^6$. Consider a matrix $D\in \F^{n\times n}$ such that $D = \diag(d_1, d_2, \ldots, d_n)$ and each $d_i$ is chosen independently and uniformly at random from $\F$. Then, for any irreducible matrix $A\in \F^{n\times n}$, $\adj{A+D}$ satisfies property $\prop$ with high probability.
\end{lemma}

\begin{proof}
    It suffices to show that a randomly chosen diagonal matrix $D$ satisfies the properties stated in \cref{clm:satisfying-prop-Q}. Consider the following small collection of polynomials, where $Y = \diag(y_1, \ldots, y_n)$:
    \begin{enumerate}
        \item $\det(A+Y)$.
        \item For each $1\leq i,j \leq n$, $\det ((A + Y)[\comp{i}, \comp{j}])$.
        \item For each distinct $i,j,k,\ell\in[n]$, $\det\left(\adj{A+Y}[\{i,j\},\{k,\ell\}]\right)$. 
    \end{enumerate}
It follows from the Polynomial Identity Lemma \cite{Sch80,Zip79, DL78} that, upon substituting each $y_i$ with a randomly chosen $d_i \in \F$, each nonzero polynomial in the collection remains nonzero with high probability.  Consequently, for the diagonal matrix $D = \diag(d_1, d_2, \ldots, d_n)$, all the properties mentioned in \cref{clm:satisfying-prop-Q} hold. Hence, $\adj{A+D}$ satisfies property $\prop$.
\end{proof}


\begin{remark}\label{lemma:random-diag-shift-satisfies-R-derandomization}
    The proof of \cref{lemma:random-diag-shift-satisfies-R} continues to hold even if, for every distinct $i,j,k,\ell$, we replace $\det\left(\adj{A+Y}[\{i,j\},\{k,\ell\}]\right)$in the collection with the polynomial $\det((A+Y)[\comp{\{k,\ell\}},\comp{\{i,j\} }])$. This follows from Jacobi's identity~\cite[Page~21]{Gantmacher59}, which implies that $\det\left(\adj{A+Y}[\{i,j\},\{k,\ell\}]\right)\neq 0$ if and only if $\det((A+Y)[\comp{\{k,\ell\}},\comp{\{i,j\} }]) \neq 0$.  Furthermore, each polynomial in the collection is now computable by a read-once determinant (ROD). Therefore, we can use the quasi-polynomial size hitting set for RODs~\cite{GT20} to compute a quasi-polynomially large collection of $n\times n$ diagonal matrices such that there exists a diagonal matrix $D$ in the collection for which $\adj{A+D}$ satisfies property~$\prop$.
\end{remark}

\begin{corollary}\label{lemma:random-diag-shift-satisfies-R-for-irred-comp}
    Consider a matrix $D\in \F^{n\times n}$ such that $D = \diag(d_1, d_2, \ldots, d_n)$ and each $d_i$ is chosen independently and uniformly at random from $\F$. Then, for any matrix $A\in \F^{n\times n}$, each irreducible block of $\adj{A+D}$ (and as a result, $(A+D)^{-1}$) satisfies property $\prop$ with high probability.
\end{corollary}

\subsection{Black-box access to \texorpdfstring{$\det((A + D)^{-1} + Y)$}{det((A+D)\^(-1) + Y)} from that to \texorpdfstring{$\det(A + Y)$}{det(A + Y)}}

We are given black-box access to $\det(A+Y)$ for some matrix $A$ and $Y = \diag(y_1,\dots,y_n)$. We demonstrate that we can obtain black-box access to $\det((A+D)^{-1} + Y)$ as follows:

\begin{lemma}\label{lem:bba_to_adj} If $\card{\F} > n+1$, from black box access to $\det(A+Y)$ with $Y = \diag(y_1,\dots,y_n)$, we can obtain black box access to $\det((A+D)^{-1} + Y)$ where $D$ is a diagonal matrix such that $(A+D)$ is invertible. \end{lemma}

\begin{proof}
By translating $y_1,\dots,y_n$, we get black-box access to $\det((A+D) + Y)$. Since we know the value of $\det(A+D)$ and that it is non-zero, we see
\[\det(A+D)^{-1}\det((A+D) + Y) = \det(I_n + (A+D)^{-1}Y)\text{.}\]
By \cref{obs:black-box-var-inversion}, we can invert the variables and multiply by $y_1\cdots y_n$ to get black-box access to
\[\det((A+D)^{-1}Y^{-1} + I_n)\det(Y) = \det((A+D)^{-1} + Y)\text{.}\]
\end{proof}

Let $D$ be a diagonal matrix of randomly chosen field elements. Since $\det(A+Y)$ is invertible, $\det(A+D)$ is invertible with high probability. Thus, we have black-box access to $\det((A+D)^{-1} + Y)$, and by \cref{lemma:random-diag-shift-satisfies-R-for-irred-comp}, the irreducible blocks of $(A+D)^{-1}$ satisfy property $\prop$ with high probability.

\subsection{Learning the index sets of the irreducible blocks of \texorpdfstring{$\det((A + D)^{-1} + Y)$}{det(adj(A+D)^(-1) + Y)}} \label{sec:access to irreduble blocks}

Suppose we have an index set $T \subseteq [n]$ and we wish to get black-box access to $\det(A[T] + Y[T])$ as well as the principal minor $\det(A[T])$ from black-box access to $\det(A+Y)$ where $Y = \diag(y_1,\dots,y_n)$. We can do so using polynomial interpolation:

\begin{claim}\label{claim:bba-to-pm-oracle}
    Let $\card{\F} > n$. Suppose we have black-box access to $\det(A+Y)$ where $Y = \diag(y_1,\dots,y_n)$ and $A \in \F^{n \times n}$. Then, for any set $T \subseteq [n]$, we can simulate one query to $\det(A[T] + Y[T])$ (and thus obtain the principal minor $\det(A[T])$) using $n-\card{T}+1$ queries to the black-box $\det(A+Y)$.
\end{claim}

\begin{proof}
    Without loss of generality, we can assume $T = [k]$. Let $f(y_1,\dots,y_n) = \det(A+Y)$, and suppose we want to evaluate $\det(A[T] + Y[T])$ at point $\alpha_1,\dots,\alpha_k$. Let $g(z) = f(\alpha_1,\dots,\alpha_k,z,z,\dots,z)$ where $z$ is a fresh new variable. Since $f$ is multilinear, we can write $$g(z) = \sum_{i=0}^{n-k} \gamma_i z^i\text{.}$$ From the nature of the polynomial $\det(A+Y)$, it is easy to see that $\gamma_{n-k}$ is the value we need. We find $n-k+1$ distinct values $\beta_0,\dots \beta_{n-k}$ and use interpolation to compute $\gamma_{n-k}$.

    If $\alpha_1=\alpha_2=\dots=\alpha_k=0$, we see that $\gamma_{n-k} = \det(A[T])$.
\end{proof}


By the previous subsection, we have black-box access to $\det((A+D)^{-1} + Y)$ where $Y= \diag(y_1,\dots,y_n)$. Let $B = (A+D)^{-1}$. Suppose $T_1 \sqcup T_2 \sqcup \dots \sqcup T_s = [n]$ are the index sets such that the irreducible blocks of $B$ are $B[T_1],B[T_2]\dots,B[T_s]$. By \cref{lem:redToIr}, there exists a permutation matrix $P$ such that $B \PME P^{T}\diag(B[T_1],\dots,B[T_s]) P$.

To learn the irreducible blocks of $B$ separately, we first need to find the index sets $T_1 \sqcup T_2 \sqcup \dots \sqcup T_s$. We can do so with the help of the following claim:

\begin{claim} Let $B$ be an $n \times n$ matrix whose irreducible blocks have non-zero off-diagonal entries. Then $i,j \in [n]$ belong to the same irreducible block of $B$ if and only if $B[\{i,j\}]$ is irreducible.\end{claim}
\begin{proof}
    If $i$ and $j$ belong to the same irreducible block, then by the fact that the irreducible blocks of $B$ have non-zero off-diagonal entries, $B[i,j]$ and $B[j,i]$ must be non-zero. Thus, $B[\{i,j\}]$ is irreducible.

    If $i$ and $j$ belong to two different irreducible blocks $B[T_q]$ and $B[T_r]$, then by \cref{lem:redToIr} there exists a permutation matrix $P$ such that $P^TBP$ is a block-triangular matrix whose blocks include $B[T_q]$ and $B[T_r]$. Thus, at least one of $B[T_q,T_r]$ and $B[T_r,T_q]$ is a zero matrix, which means at least one of $B[i,j]$ and $B[j,i]$ is zero. Thus, $B[\{i,j\}]$ is reducible.
\end{proof}
\begin{corollary}\label{cor:off_diagonal_non_zero_irred_block_transitivity}
    For $i,j,k \in [n]$, if both $B[\{i,j\}]$ and $B[\{j,k\}]$ are irreducible, then $B[\{i,k\}]$ is irreducible.
\end{corollary}

By \cref{lemma:random-diag-shift-satisfies-R-for-irred-comp}, the irreducible blocks of $B = (A+D)^{-1}$ satisfy property $\prop$, and thus have non-zero off-diagonal entries. With black-box access to $\det(B[\{i,j\}] + Y[\{i,j\}])$ given by \cref{claim:bba-to-pm-oracle}, we can easily check if $B[\{i,j\}]$ is irreducible.

\begin{claim} \label{clm:checking 2x2 reducibility}
    Let $\det(B[\{i,j\}] + Y[\{i,j\}]) = y_iy_j + \alpha y_i + \beta y_j + \gamma$, then $B[\{i,j\}]$ is reducible if and only if $\gamma = \alpha\beta$.
\end{claim}
\begin{proof}
    \begin{align*}\det(B[\{i,j\}] + Y[\{i,j\}]) &= (y_i + B[i,i])(y_j + B[j,j]) - B[i,j]B[j,i] \\
        &= y_iy_j + y_iB[j,j] + B[i,i]y_j + B[i,i]B[j,j] - B[i,j]B[j,i]\text{.}\end{align*}
    $B[i,i]B[j,j] - B[i,j]B[j,i] = B[i,i]B[j,j]$ if and only if $B[i,j]B[j,i] = 0$, that is $B[\{i,j\}]$ is reducible.
\end{proof}

We now have the following algorithm to find the index sets for the irreducible blocks:

\begin{algorithm}[H]
\caption{Finding the indices of the irreducible blocks of $\det((A+D)^{-1} + Y)$}
\label{algo:PMAP_redn_finding_irred_blocks}
\textbf{Input:} Black box access to $\det(B+Y)$ where $Y = \diag(y_1,\dots,y_n)$ and the irreducible blocks of $B$ have non-zero off-diagonal entries.\\
\textbf{Output:} A set of index sets $\mathcal{T} = \{T_1,\dots,T_s\}$ such that $B[T_1],\dots,B[T_s]$ are the irreducible blocks of $B$.
\begin{algorithmic}[1]
\State Let $E \gets \emptyset$.
\State Let $\mathcal{T} \gets \emptyset$.
\For{$i \in [n]$}
    \If{$i \not\in E$}
        \State $E \gets E \cup \{i\}$.
        \State $T \gets \{i\}$.
        \For{$j \in [n]$}
            \label{alg:step-check-2-irreducible-1}\If{$j \not\in E$ and $B[\{i,j\}]$ is irreducible}
                \State $E \gets E \cup \{j\}$.
                \State $T \gets T \cup \{j\}$.
            \EndIf
        \EndFor
        \State $\mathcal{T} \gets \mathcal{T} \cup \{T\}$.
    \EndIf
\EndFor
\For{$T \in \mathcal{T}$}
    \For{$i,j \in T, i\neq j$}
        \label{alg:step-check-2-irreducible-2}\If{$B[\{i,j\}]$ is reducible}
            \label{algo:PMAP_redn_finding_irred_blocks_error}\State \Return \textsc{ERROR}.
        \EndIf
    \EndFor
\EndFor
\State \Return $\mathcal{T}$.
\end{algorithmic}
\end{algorithm}

In steps 8 and 18 of \cref{algo:PMAP_redn_finding_irred_blocks}, we check the irreducibility of the $2 \times 2$ matrices using \cref{clm:checking 2x2 reducibility}. Note that when the irreducible blocks of $B$ have non-zero off-diagonal entries, \cref{cor:off_diagonal_non_zero_irred_block_transitivity} ensures that the algorithm never throws an error in step 19.

Once we have the index sets $T_1,\dots,T_s$, using \cref{claim:bba-to-pm-oracle}, we have black-box access to $\det(B[T_1] + Y[T_1]), \det(B[T_2] + Y[T_2]), \dots, \det(B[T_s] + Y[T_s])$ as well as their principal minors. We can learn the matrices $B[T_1],B[T_2],\dots, B[T_s]$ separately. Since $B[T_1],B[T_2],\dots, B[T_s]$ satisfy property $\prop$ and we have access to their principal minors (which are just the principal minors of the matrix $B$), learning them reduces to PMAP for matrices satisfying property $\prop$.

\begin{remark}\label{remark:irred_blocks_finding_derandomisation}
    In the scenario where the diagonal matrix $D$ is `bad', that is, the irreducible blocks of $(A+D)^{-1}$ do not all have non-zero off-diagonal entries, observe that if $i,j$ are not indices to the same irreducible block, then \cref{algo:PMAP_redn_finding_irred_blocks} still correctly identifies $i$ and $j$ as indices to different blocks (as the matrix $B[\{i,j\}]$ remains reducible). It is in the case where $i,j$ are indices to the same irreducible block that the algorithm may fail: $B[\{i,j\}]$ might be reducible owing to the diagonal entries no longer being guaranteed to be non-zero. Thus, for a `bad' diagonal matrix $D$, \cref{algo:PMAP_redn_finding_irred_blocks} either finds a violation of \cref{cor:off_diagonal_non_zero_irred_block_transitivity} and throws an error or reports more irreducible blocks than what is actually correct.

    If we are given a quasi-polynomial number of diagonal matrices $D_1,\dots,D_k$ where at least one is guaranteed to be `good', as is the case in \cref{lemma:random-diag-shift-satisfies-R-derandomization}, then by running the algorithm for all matrices $(A+D_1)^{-1},(A+D_2)^{-1},\dots,(A+D_k)^{-1}$, we see that for the `good' diagonal matrix $D_i$, the algorithm will output the minimum number of index sets, and thus we can still find the index sets corresponding to the irreducible blocks of $(A+D_i)^{-1}$ (which is the same as that of $A$) in deterministic quasi-polynomial time.
\end{remark}


\subsection{Recovering a matrix that is principal minor equivalent to \texorpdfstring{$A$}{A}}

Let $T_1 \sqcup T_2 \sqcup \dots \sqcup T_s = [n]$ be the indices corresponding to the irreducible blocks of $(A+D)^{-1}$ where $D$ is the diagonal matrix with randomly chosen entries. Let $C_1, C_2, \dots, C_s$ be matrices such that $C_i \PME (A+D)^{-1}[T_i]$. Since \cref{algo:PMAP_redn_finding_irred_blocks} gave us the index sets $T_1,T_2,\dots,T_s$, we can find a permutation matrix $P$ such that $C = P^{T} \diag(C_1,\dots,C_s) P$ is principal minor equivalent to $(A+D)^{-1}$ (by \cref{lem:redToIr}). We now show how to recover a matrix $B \PME A$ from $C$.

\begin{lemma}\label{lem:pme_A_from_pme_adj_A}
    Let $C \PME (A+D)^{-1}$ where $D$ is a diagonal matrix such that $A+D$ is invertible. Then, $C^{-1} - D \PME A$.
\end{lemma}

\begin{proof}
Since $C \PME (A+D)^{-1}$, $\det(C + Y) = \det((A+D)^{-1} + Y)$ where $Y = \diag(y_1,\dots,y_n)$. As $\det(C) = \det((A+D)^{-1}) \neq 0$, 
\begin{align*}
    \det(C^{-1}) \det(C + Y) &= \det(A+D) \det((A+D)^{-1} + Y) \\
    \implies \det(I_n + C^{-1} Y)&= \det(I_n + (A+D)Y)\text{.}
\end{align*}
Inverting the variables and multiplying by $y_1\dots y_n$, we have
\begin{align*}
    \det(C^{-1} Y^{-1} + I_n)\det(Y) &= \det((A+D)Y^{-1} + I_n)\det(Y) \\
    \implies \det(C^{-1} + Y) &= \det(A+D + Y)\text{.}
\end{align*}
Translating $Y$ by $-D$, we see that $\det(C^{-1} - D + Y) = \det(A + Y)$. Thus, $C^{-1} - D \PME A$.
\end{proof}

\subsection{Derandomising the reduction}\label{sec:deranomisation}

The randomised reduction from $\probB$ to PMAP for matrices with property $\prop$, shown in the preceding sections, can be derandomised in quasi-polynomial time using the hitting-set construction from \cite{GT20}. We provide a sketch of this deterministic reduction below:  

\begin{enumerate}
    \item Using \cref{remark:irred_blocks_finding_derandomisation}, we can find the index sets corresponding to the irreducible blocks of $A$. 
    \item By \cref{lemma:random-diag-shift-satisfies-R-derandomization}, which uses the \cite{GT20} hitting-set, we have a list of quasi-polynomially many diagonal matrices $D_1,\dots,D_k$ that contains at least one matrix $D$ such that the irreducible blocks of $(A+D)^{-1}$ satisfy property~$\prop$. We say that such a $D$ is correct.
    \item How do we identify a correct $D$ in the above list? If we proceed with an incorrect $D$, then it is quite possible that the reconstructed matrix $B$ is not principal minor equivalent to $A$. This issue can be avoided if we succeed in checking at the end that the reconstructed $B$ is indeed principal minor equivalent to $A$. Next, we argue how this is done. 
    
    
    \item To verify if $B\PME A$, given $B$ \emph{explicitly} and black-box access to $\det(A+Y)$, we first check if their irreducible blocks have the same index sets (using \cref{remark:irred_blocks_finding_derandomisation}). Then, we verify PME of the irreducible blocks separately. For the rest of the argument, assume $A$ and $B$ are irreducible. By \cref{lemma:random-diag-shift-satisfies-R-derandomization}, we have quasi-polynomially many diagonal matrices $D_1,\dots,D_r$ that contain at least one matrix $D$ such that $(B+D)^{-1}$ satisfies property $\prop$. Since we know $B$ explicitly, we can find the correct $D$ by checking if it satisfies the conditions listed in \cref{clm:satisfying-prop-Q}; the third condition can be verified by checking if $\det\left(\adj{B+D}[\{i,j\}, \{k,l\}]\right)=0 \Longrightarrow \forall m\in[r], \ \det\left(\adj{B+D_m}[\{i,j\},\{k,l\}]\right)=0$. By \cref{lem:bba_to_adj}, we have black-box access to $\det((A+D)^{-1} + Y)$ and its principal minors. We can now verify if the corresponding principal minors of orders at most $4$ of $(A+D)^{-1}$ and $(B+D)^{-1}$ are equal. If they are, then by \cref{thm:main-characterization}, $(B+D)^{-1} \PME (A+D)^{-1}$, and by \cref{lem:pme_A_from_pme_adj_A}, $B \PME A$.
\end{enumerate}
\section{Sufficiency of PME up to Order $4$: Proof of Theorem \ref{thm:main-characterization}}
\label{sec:PME-upto-4-characterization}

A key technical contribution of this paper is the identification of a property, denoted by $\prop$, such that for any two square matrices $A$ and $B$ with nonzero off-diagonal entries, if $A$ satisfies $\prop$, then $A \PME B$ if and only if $\kPME{A}{B}{4}$. In this section, we first define property $\prop$ and then prove~\cref{thm:main-characterization}. Finally, we construct a counterexample showing that if $A$ does not satisfy $\prop$, it is possible for $\kPME{A}{B}{4}$ to hold while $A$ and $B$ are not principal minor equivalent.

\begin{definition}[Property $\prop$]
\label{def:matrix-property}
Let $A$ be an $n\times n$ over a field $\F$. We say that $A$ has the \emph{property $\prop$} if it satisfies the following:
\begin{enumerate}
\item all the off-diagonal entries are nonzero.
\item if $\rank A[\{i,j\}, \{k,\ell\}]=1$ for any  four distinct elements $i, j,k,\ell\in[n]$, then there exists a subset $X$ of $[n]$ such that $i,j\in X$, $k,\ell\in\comp{X}$, and $\rank A[X, \comp{X}]=1$.
\end{enumerate}
\end{definition}

The following lemma shows that if a matrix $A$ satisfies property $\prop$, then every principal submatrix of $A$ also satisfies property $\prop$.

\begin{lemma}
\label{prop:matrix-property-large-to-small}
Let $A$ be an $n\times n$ matrix over a field $\F$ with $n\geq 5$. Suppose that $A$ satisfies property $\prop$. Let  $S\subseteq [n]$ with $|S|\geq 5$. Then, $A[S]$ also satisfies property $\prop$.
\end{lemma}
\begin{proof}
Note that off-diagonal entries of $A[S]$ are also off-diagonal entries of $A$. Hence, they are nonzero. Let $i,j,k, \ell$ be four distinct elements from the set $S$ such that $\rank A[\{i,j\}, \{k,\ell\}]=1$. Since $A$ satisfies property $\prop$, there exists a subset $X$ of $[n]$ such that $i, j\in X$, $k,\ell\in\comp{X}$ and $\rank A[X, \comp{X}]=1$. Let $X_S=X\cap S$ and $\comp{X}_S=\comp{X}\cap S$. Then, the submatrix $A[X_S, \comp{X}_S]$ is a submatrix of $A[X, \comp{X}]$. Therefore, $\rank A[X_S, \comp{X}_S]=1$.
\end{proof}

Our next lemma is inspired by~\cite[Lemma~6]{Loewy86}. Suppose $A$ is an $n \times n$ matrix. The lemma provides a way to derive rank-one constraints on certain submatrices of $A$ from a cut in one of its $(n-1) \times (n-1)$ principal submatrices. Assume that $A$ has no cut. Then, combined with~\cref{lem:no-cut-property}, this result implies that $A$ must have at least three $(n-1) \times (n-1)$ principal submatrices with no cut. This will be useful for the inductive step in our proof.
\begin{lemma}
\label{lem:lift-of-cut}
Let $n\geq 5$ be a positive integer. Let $A$ be an $n\times n$ matrix over $\F$ satisfying property $\prop$. Let $X\subseteq [n]\setminus\{i\}$ be a cut of the matrix $A[\nsubset{i}]$ and $\comp{X}=[n]\setminus (X\cup\{i\})$. Then, $$\rank A[X+i,  \comp{X}]=1 \text{ or } \rank A[X,  \comp{X}+i]=1,$$ and $$\rank A[\comp{X}+i,  X]=1 \text{ or } \rank A[\comp{X},  X+i]=1.$$
\end{lemma}
\begin{proof}
From the definition of the cut of a matrix, we know that the cardinality of both $X$ and $\comp{X}$ is at least $2$. Let $s\in X$ and $t\in \comp{X}$. Let $X'=X-s$ and $\comp{X}'=\comp{X}-t$. Since $\rank A[X, \comp{X}]=1$, for all $u\in X'$ and $v\in\comp{X}'$, $$\rank A[\{s, u\}, \{t, v\}]=1.$$ Property $\prop$ of $A$ implies that, for all $u\in X'$ and $v\in\comp{X}'$, we have a subset $X_{uv}$ of $[n]$ such that $$s, u\in X_{uv},\ \ \ \  t, v\in \comp X_{uv}=[n]\setminus X_{uv}\  \text{ and }\ \rank A[X_{uv}, \comp X_{uv}]=1.$$ Now, for every $u\in X'$, consider the following subsets of $[n]$. $$X_u=\bigcap_{v\in \comp{X}'}X_{uv}\ \ \text{ and }\ \ \comp{X}_u=[n]\setminus X_u=\bigcup_{v\in \comp{X}'} \comp{X}_{uv}.$$ Then, \cref{prop:old-to-new-cut} implies that for all $u\in X'$, $\rank A[X_u, \comp{X}_u]=1$. Also, observe that for all $u\in X'$, $\comp{X}$ is a subset of $\comp{X}_u$ and $\{s,u\}\subseteq X_u$. We now consider the following subsets of $[n]$. $$T=\bigcup_{u\in X'}X_u \ \ \text{ and }\ \ \comp{T}=[n]\setminus T=\bigcap_{u\in X'}\comp{X}_u.$$ Note that $X\subseteq T$ and $\comp{X}\subseteq \comp{T}$. Furthermore, from~\cref{prop:old-to-new-cut}, we have $\rank A[T, \comp T]=1$. Since $X\subseteq T$, $\comp{X}\subseteq \comp{T}$, and $(X, \comp{X})$ is a partition of $\nsubset{i}$, $$\text{either } T=X+i \text{ and } \comp{T}=\comp{X},\ \text{ or }\ T=X \text{ and } \comp{T}=\comp{X}+i,$$ and  
\begin{equation}
\label{eqn:lift-of-cut-proof-eqn-1}
\rank A[X+i, \comp{X}]=1,\ \ \text{ or }\ \ \rank A[X, \comp{X}+i]=1
\end{equation}

Similarly, using $\rank A[\comp{X}, X]=1$, we can show that 
\begin{equation}
\label{eqn:lift-of-cut-proof-eqn-2}
\rank A[\comp{X}+i, X]=1,\ \ \text{ or }\ \ \rank A[\comp{X}, X+i]=1
\end{equation}
Thus,~\cref{eqn:lift-of-cut-proof-eqn-1} and~\cref{eqn:lift-of-cut-proof-eqn-2}, taken together, complete the proof of the lemma. 
\end{proof}

The combination of \cref{lem:no-cut-property} and \cref{lem:lift-of-cut} yields the following corollary.

\begin{corollary}
\label{cor:no-cut-from-big-to-small-matrix}
Let $n\geq 5$ be a positive integer. Let $A$ be an $n\times n$ matrix over $\F$ satisfying property $\prop$. Suppose that $A$ has no cut. Then, there exist at least three elements $i_1, i_2, i_3\in[n]$ such that the submatrices $A[\nsubset{i_j}]$ for $j=1,2,3$ have no cut. 
\end{corollary}
\begin{proof}
Note that $A$ has no cut. Therefore, for every $i\in [n]$ with $A[\nsubset{i}]$ has a cut,~\cref{lem:lift-of-cut}, when applied to $A$, yields a constraint of type~\cref{eqn:cut-of-matrix-two} or~\cref{eqn:cut-of-matrix-three}. Consequently, by applying~\cref{lem:no-cut-property}, the corollary follows.
\end{proof}

Suppose that $n\geq 5$ is a positive integer. Let $A$ and $B$ be two $n\times n$ matrices such that $A$ has no cut and $A$ satisfies property $\prop$. Additionally, assume that $\kPME{A}{B}{n-1}$. Then, from~\cref{cor:no-cut-from-big-to-small-matrix} and~\cref{lem:no-cut-to-PME},  we know that $$|\calD_1\left(A, B\right)\cup \calD_2\left(A, B\right)|\geq 3.$$ On the other hand, applying~\cref{lem:partial-DS-to-complete-DS}, if $$|\calD_1\left(A, B\right)\cup \calD_2\left(A, B\right)|\geq 5,$$ then $A\DE B$, and therefore using~\cref{obs:DE-to-PME}, $A\PME B$. We also show that assumption  
\begin{equation}
\label{eqn:union-of-D_1-and-D_2}
|\calD_1\left(A, B\right)\cup \calD_2\left(A, B\right)|=3 \text{ or } 4.
\end{equation}
leads to a contradiction. Therefore, $$|\calD_1\left(A, B\right)\cup \calD_2\left(A, B\right)|\neq 3 \text{ or } 4.$$
This, in turn, allows us to prove that $A\PME B$ if and only if $\kPME{A}{B}{4}$. First, we show that if~\cref{eqn:union-of-D_1-and-D_2} holds, then $\calD_1(A, B)$ and $\calD_2(A,B)$  cannot have a nonempty intersection. 
\begin{lemma}
\label{prop:empty-intersection-of-D_1-and-D_2}
Let $n\geq 5$ be a positive integer. Suppose that $A$ and $B$ are two $n\times n$ matrices with all off-diagonal entries of $A$ are nonzero, and $\kPME{A}{B}{3}$. Let $$|\calD_1\left(A, B\right)\cup \calD_2\left(A, B\right)|=3 \text{ or } 4.$$ Then, the intersection of $\calD_1(A,B)$ and $\calD_2(A,B)$ is empty.
\end{lemma}
\begin{remark}
If $|\calD_1\left(A, B\right)\cup \calD_2\left(A, B\right)|=4$, we do not require the assumption $\kPME{A}{B}{3}$ in our proof. Moreover, when $|\calD_1\left(A, B\right)\cup \calD_2\left(A, B\right)|=3$, it suffices to assume $A[S]\PME B[S]$ where $S$ is the union of $\calD_1(A, B)$ and $\calD_2(A, B)$.
\end{remark}
\begin{proof}
For the sake of contradiction, assume that $\calD_1(A, B)$ and $\calD_2(A, B)$ has nonempty intersection. We now divide our proof into the following two cases.
\paragraph*{Case I ($|\calD_1\left(A, B\right)\cup \calD_2\left(A, B\right)|=4$).} The nonempty intersection between $\calD_1(A,B)$ and $\calD_2(A,B)$ implies that either $|\calD_1(A, B)|=3$ or $|\calD_2(A, B)|=3$. Therefore, from~\cref{lem:partial-DS-to-complete-DS}, $\calD_1(A, B)\cup\calD_2(A, B)=[n]$, which is a contradiction. Hence, $\calD_1(A,B)$ and $\calD_2(A,B)$ have an empty intersection.

\paragraph*{Case II ($|\calD_1\left(A, B\right)\cup \calD_2\left(A, B\right)|=3$).} If the cardinality of the intersection of $\calD_1(A,B)$ and $\calD_2(A, B)$ is at least two, then either $|\calD_1(A, B)|=3$ or $|\calD_2(A, B)|=3$. Therefore, like \textbf{Case I},~\cref{lem:partial-DS-to-complete-DS} implies that $\calD_1(A, B)\cup\calD_2(A, B)=[n]$, which is a contradiction. Hence, $|\calD_1(A, B)\cap \calD_2(A, B)|=1$. Without loss of generality, assume that $$\calD_1(A, B)=\{1,2\} \text{ and } \calD_2(A, B)=\{2, 3\}.$$ For all $k\in[n]$, let $N_k=[k]$. Note that $$\{1,2\}\subseteq \calD_1(A[N_4], B[N_4])\ \ \ \text{ and }\ \ \ \{2,3\}\subseteq \calD_2(A[N_4], B[N_4]).$$ From the hypothesis of the lemma, we know that $A[N_3]\PME B[N_3]$. Therefore, applying~\cref{lem:no-cut-to-PME}, $4\in\calD_1(A[N_4], B[N_4])$ or $4\in\calD_2(A[N_4], B[N_4])$. Based on that, using~\cref{lem:partial-DS-to-complete-DS}, we obtain that $$A[N_4]\DS B[N_4]\ \ \text{ or }\ \ A[N_4]\DS B[N_4]^T.$$ Thus, by repeated application of~\cref{lem:partial-DS-to-complete-DS}, we obtain that $A[N_k]\DE B[N_k]$ for all $k\in\{5,6, \ldots, n\}$. Therefore, $\calD_1(A, B)\cup\calD_2(A, B)=[n]$, which is a contradiction. This completes the proof.
\end{proof}

The above lemma implies that in~\cref{eqn:union-of-D_1-and-D_2}, we can assume that the intersection of $\calD_1(A, B)$ and $\calD_2(A, B)$ is empty. Next we show that $|\calD_1\left(A, B\right)\cup \calD_2\left(A, B\right)|\neq 4$. To that end, the following lemma establishes a key property of $4 \times 4$ matrices having nonzero off diagonal entries.

\begin{lemma}
\label{lem:small-matrices-DE-to-cut-4-base-case}
Let $A$ and $B$ be two $4\times 4$ matrices such that all off-diagonal entries of $A$ are nonzero, and $A\PME B$. Let $X\subseteq \{1,2,3,4\}$ with $|X|=2$, $\calD_1(A,B)=X$ and $\calD_2(A,B)=\comp{X}$. Then, $X$ is a cut of $A$.  
\end{lemma}
The proof closely follows the proof of~\cite[Theorem~4]{Hartfiell84}. For details see~\cref{subsec:proof-of-lem:small-matrices-DE-to-cut-4-base-case}. The following lemma addresses the case when the union of $\calD_1(A, B)$ and $\calD_2(A, B)$ contains four elements (see~\cref{eqn:union-of-D_1-and-D_2}). A lemma of similar flavour appears in~\cite[Theorem~2]{Loewy86} with the following differences. In their hypothesis, they assume that $A\PME B$, which is stronger than our assumption $\kPME{A}{B}{4}$. On the other hand, we assume that $A$ satisfies property $\prop$, but they only require that the off-diagonal entries of $A$ are nonzero, that is, the first condition of property $\prop$. Nevertheless, the conclusion remains same in both the cases. Additionally, the proof of our lemma is significantly simpler compared to the proof of~\cite[Theorem~2]{Loewy86}. 

\begin{lemma}
\label{lem:small-matrices-DE-to-cut-4}
Let $n$ be a positive integer with $n\geq 4$. Let $A$ and $B$ be two $n\times n$ matrices over $\F$ such that $A$ satisfies property $\prop$, and $\kPME{A}{B}{4}$. Let $i_1, i_2, i_3, i_4\in[n]$ be four distinct positive integers such that $\calD_1(A,B)=\{i_1, i_2\}$ and $\calD_2(A,B)=\{i_3, i_4\}$. Then, there exists an $X\subseteq [n]$ such that $i_1, i_2\in X$, $i_3, i_4\in\comp{X}$, and $X$ is a cut of $A$.
\end{lemma}

\begin{remark}
In the lemma above, it is sufficient to use $A[\{i_1, i_2, i_3, i_4\}]\PME B[\{i_1, i_2, i_3, i_4\}]$ in place of $\kPME{A}{B}{4}$ to establish the same conclusion. 
\end{remark}

\begin{proof}
Without loss of generality, we can assume that $\{i_1, i_2\}=\{1,2\}$ and $\{i_3, i_4\}=\{3, 4\}$. 

Observe that there exists a matrix $A'$ such that it is diagonally similar to $A$ and $A'[\nsubset{j}]=B[\nsubset{j}]$ for $j=1,2$. Also, instead of $A$, we can work with any matrix diagonally similar to $A$. Thus, without loss of generality, we may assume that $A[\nsubset{j}]=B[\nsubset{j}]$ for $j=1,2$. This implies that 
\begin{equation}
\label{eqn:small-matrices-DE-to-cut-4-one}
A[i,j]=B[i,j] \text{ for all } (i,j)\notin\{(1,2), (2,1)\}.
\end{equation}

Let $C$ be an $(n-1)\times (n-1)$ diagonal matrix whose rows and columns are indexed by $[n]\setminus\{3\}$ and for $i\in [n]\setminus\{3\}$, $C[i,i]=c_i$ with $c_1=1$. Let $D$ be another $(n-1)\times (n-1)$ diagonal matrix whose rows and columns are indexed by $[n]\setminus\{4\}$ and for $i\in [n]\setminus\{4\}$, $D[i,i]=d_i$ with $d_1=1$. Furthermore, assume that 
\begin{equation}
\label{eqn:small-matrices-DE-to-cut-4-two}
B[\nsubset{3}]^T=C^{-1}A[\nsubset{3}]C \ \ \text{ and }\ \ B[\nsubset{4}]^T=D^{-1}A[\nsubset{4}]D.
\end{equation}
The above equation implies that
\begin{equation}
\label{eqn:small-matrices-DE-to-cut-4-three}
c_2=d_2 = \frac{B[2,1]}{A[1,2]}, \text{ and for all } 5\leq i\leq n, c_i=d_i =\frac{A[i,1]}{A[1,i]} 
\end{equation}
Additionally, from~\cref{eqn:small-matrices-DE-to-cut-4-one},~\cref{eqn:small-matrices-DE-to-cut-4-two} and~\cref{eqn:small-matrices-DE-to-cut-4-three}, we have the following:
\begin{enumerate}
\item $B[1,2]=A[2,1]c_1c_2^{-1}$ and $B[2,1]=A[1,2]c_1^{-1}c_2$
\item $A[j,4]=A[4,j]c_jc_4^{-1}$ for all $j\in\{1,2,5,6,\ldots, n\}$ 
\item $A[j,3]=A[3,j]c_jd_3^{-1}$ for all $j\in\{1,2,5,6,\ldots, n\}$
\item $A[j,i]=A[i,j]c_jc_i^{-1}$ for all $j\in\{1,2,5,6,\ldots, n\}$ and $i\in\{5,6,\ldots, n\}$
\end{enumerate}
Let $e_1=1$, $e_2=c_2$, $e_3=d_3$, $e_4=c_4$ and for all $5\leq k \leq n$, $e_k=c_k$. Then, from the above discussion, we obtain that 
\begin{equation}
\label{eqn:small-matrices-DE-to-cut-4-four}
A[j,i] =A[i,j]\cdot \frac{e_j}{e_i}\ \  \text{ for all } (i,j)\notin\{(1,2), (2,1), (3,4), (4,3)\}.
\end{equation}
Let $E$ be an $n\times n$ diagonal matrix such that for all $i\in[n]$, $E[i,i]=e_i$. Thus, applying~\cref{eqn:small-matrices-DE-to-cut-4-four}, for any $S\subseteq [n]$ with $\{1,2\}\subseteq S$ and $\{3,4\}\subseteq \comp{S}$, 
\begin{equation}
\label{eqn:small-matrices-DE-to-cut-4-five}
A[\comp{S}, S]=E[\comp{S}]\cdot A[S, \comp{S}]^T\cdot E[S]^{-1}.
\end{equation}

From the hypthesis of the lemma, we know that $$\{1,2\}\subseteq \calD_1\left(A[\{1,2,3,4\}], B[\{1,2,3,4\}]\right) \text{ and } \{3,4\}\subseteq\calD_2\left(A[\{1,2,3,4\}], B[\{1,2,3,4\}]\right).$$ Next, we argue that $$\calD_1\left(A[\{1,2,3,4\}], B[\{1,2,3,4\}]\right)=\{1,2\} \text{ and } \calD_2\left(A[\{1,2,3,4\}], B[\{1,2,3,4\}]\right)=\{3,4\}.$$ For the sake of contradiction, first, assume that either $3$ or $4$ is in $\calD_1\left(A[\{1,2,3,4\}], B[\{1,2,3,4\}]\right)$. Then, from~\cref{lem:partial-DS-to-complete-DS}, $A[\{1,2,3,4\}]\DS  B[\{1,2,3,4\}]$. Applying~\cref{eqn:small-matrices-DE-to-cut-4-one}, we obtain that $A=B$. This implies $\calD_1(A,B)=[n]$, which is a contradiction. Now, assume that either $1$ or $2$ is in $\calD_2\left(A[\{1,2,3,4\}], B[\{1,2,3,4\}]\right)$. Then, again from~\cref{lem:partial-DS-to-complete-DS}, $A[\{1,2,3,4\}]\DS  B[\{1,2,3,4\}]^T$. This combined with~\cref{eqn:small-matrices-DE-to-cut-4-two} implies that $B^T=E^{-1}AE$. Therefore, $\calD_2(A,B)=[n]$, which is a contradiction. Thus, 
\begin{equation}
\label{eqn:small-matrices-DE-to-cut-4-six}
\calD_1\left(A[\{1,2,3,4\}], B[\{1,2,3,4\}]\right)=\{1,2\} \text{ and } \calD_2\left(A[\{1,2,3,4\}], B[\{1,2,3,4\}]\right)=\{3,4\}.
\end{equation}

\cref{eqn:small-matrices-DE-to-cut-4-six} and~\cref{lem:small-matrices-DE-to-cut-4-base-case} combined imply that $\{1,2\}$ is a cut of the matrix $A[\{1,2,3,4\}]$. Since $A$ satisfies property $\prop$, there exists a subset $X\subseteq [n]$ with $1,2\in X$, $3,4 \in\comp{X}$ and $\rank A[X, \comp{X}]=1$. Now, applying~\cref{eqn:small-matrices-DE-to-cut-4-five}, $\rank A[\comp{X}, X]=1$. Therefore, $X$ is our desired cut for the matrix $A$.
\end{proof}

Our next lemma deals with the case when the union of $\calD_1(A, B)$ and $\calD_2(A, B)$ contains three elements (see~\cref{eqn:union-of-D_1-and-D_2}). As with the previous lemma, a result of similar flavour appears in~\cite[Theorem~4]{Loewy86}. Furthermore, the differences in the hypotheses between our lemma and~\cite[Theorem~4]{Loewy86} are the same as before: they assume $A \PME B$, whereas we assume only $\kPME{A}{B}{4}$; and we assume that $A$ satisfies property $\prop$, while they require only that the off-diagonal entries of $A$ are nonzero. Nevertheless, the conclusions in both cases remain the same.

\begin{lemma}
\label{lem:small-matrices-DE-to-cut-3}
Let $n$ be a positive integer with $n\geq 4$. Let $A$ and $B$ be two $n\times n$ matrices over $\F$ such that $A$ satisfies property $\prop$, and $\kPME{A}{B}{4}$. Let $i_1, i_2, i_3\in[n]$ be three distinct positive integers, and $\calD_1(A,B)=\{i_1, i_2\}$ and $\calD_2(A,B)=\{i_3\}$. Let $$T_1:=\left\{j\in[n]\setminus\{i_1, i_2, i_3\}\,\mid\, B[\{i_1, i_2, i_3, j\}]^T\DS A[\{i_1, i_2, i_3, j\}]\right\}$$ and $$T_2:=\left\{j\in[n]\setminus\{i_1, i_2, i_3\}\,\mid\, B[\{i_1, i_2, i_3, j\}]^T \nDS A[\{i_1, i_2, i_3, j\}]\right\}.$$ Then, there exists a subset $X$ of $[n]$ such that $\{i_1, i_2\}\subseteq X$, $T_2+i_3\subseteq \comp{X}$ and $X$ is a cut of the matrix $A$.  
\end{lemma}
\begin{proof}
Without loss of generality, we may assume that $\{i_1,i_2\}=\{1,2\}$ and $i_3=3$. We divide our proof into two cases.

\paragraph*{Case I ($n=4$).} We show that for $n=4$ the hypothesis of the lemma is false. Since $A[\{1,2,3\}]\PME B[\{1,2,3\}]$, from~\cref{lem:no-cut-to-PME}, $A[\{1,2,3\}]\DE A[\{1,2,3\}]$. Therefore, $4\in \calD_1(A, B)\cup \calD_2(A,B)$, which is a contradiction. Hence, for $n=4$, the lemma holds vacuously.

\paragraph*{Case II ($n\geq 5$).} We first show that $T_2$ is nonempty. For the sake of contradiction, assume that $T_2=\emptyset$. For all $i,j\in\{4,5,\ldots, n\}$, let $S_{i,j}=\{1,2,3,i,j\}$. Observe that $$T_2=\emptyset\Longrightarrow \{3,i,j\}\subseteq \calD_2(A[S_{i,j}], B[S_{i,j}]).$$ Therefore, from~\cref{lem:partial-DS-to-complete-DS}, $A[S_{i,j}]^T\DS B[S_{i,j}]$. Thus, repeatedly applying~\cref{lem:partial-DS-to-complete-DS}, we obtain that $\calD_2(A, B)=[n]$ which is a contradiction. Therefore, $T_2$ is a nonempty set.

Since $\kPME{A}{B}{4}$, $A[\{1,2,3\}]\PME B[\{1,2,3\}]$. Therefore, from~\cref{lem:no-cut-to-PME}, $A[\{1,2,3\}]\DS B[\{1,2,3\}]$, or $A[\{1,2,3\}]\DS B[\{1,2,3\}]^T$. Suppose that $A[\{1,2,3\}]\DS B[\{1,2,3\}]$. Then, for all $i\in\{4,5,6, \ldots, n\}$ $$\{1,2,i\}\subseteq \calD_1(A[\{1,2,3,i\}], B[\{1,2,3,i\}]).$$ Therefore, applying~\cref{lem:partial-DS-to-complete-DS}, $A[\{1,2,3,i\}]\DS B[\{1,2,3,i\}]$ for all $i\in\{4,5,6,\ldots, n\}$. Thus, repeatedly applying~\cref{lem:partial-DS-to-complete-DS}, we obtain $A\DS B$. Therefore, $\calD_1(A,B)=[n]$, which is a contradiction. Thus, 
\begin{equation}
\label{eqn:small-matrices-DE-to-cut-3-one}
A[\{1,2,3\}]\DS B[\{1,2,3\}]^T, \text{ but }A[\{1,2,3\}]\nDS B[\{1,2,3\}]
\end{equation}

Suppose that $i\in T_2$. Then, from the definition of $T_2$, $A[\{1,2,3,i\}]\nDS B[\{1,2,3,i\}]^T$.  From the hypothesis of the lemma and~\cref{eqn:small-matrices-DE-to-cut-3-one}, we have that $$\{1,2\}\subseteq \calD_1(A[\{1,2,3,i\}], B[\{1,2,3,i\}])\ \ \text{ and }\ \ \{3,i\}\subseteq\calD_2(A[\{1,2,3,i\}], B[\{1,2,3,i\}]).$$ Moreover, applying~\cref{eqn:small-matrices-DE-to-cut-3-one}, we obtain that $i\notin\calD_1(A[\{1,2,3,i\}], B[\{1,2,3,i\}])$. Since $\kPME{A}{B}{4}$, from~\cref{lem:no-cut-to-PME}, $A[\{1,2,i\}]\DS B[\{1,2,i\}]$ or $A[\{1,2,i\}]\DS B[\{1,2,i\}]^T$. If $A[\{1,2,i\}]\DS B[\{1,2,i\}]$, then applying~\cref{lem:partial-DS-to-complete-DS}, $A[\{1,2,3,i\}]\DS B[\{1,2,3,i\}]$. This implies $A[\{1,2,3\}]\DS B[\{1,2,3\}]$, which is a contradiction due to~\cref{eqn:small-matrices-DE-to-cut-3-one}. Therefore, $A[\{1,2,i\}]\DS B[\{1,2,i\}]^T$ but $A[\{1,2,i\}]\nDS B[\{1,2,i\}]$. Thus, $$\{1,2\}=\calD_1(A[\{1,2,3,i\}], B[\{1,2,3,i\}])\ \ \text{ and }\ \ \{3, i\}\subseteq \calD_2(A[\{1,2,3,i\}], B[\{1,2,3,i\}]).$$ Furthermore, if $1$ or $2$ is in $ \calD_2(A[\{1,2,3,i\}], B[\{1,2,3,i\}])$, then again by~\cref{eqn:small-matrices-DE-to-cut-3-one} and~\cref{lem:partial-DS-to-complete-DS}, $A[\{1,2,3,i\}]\DS B[\{1,2,3,i\}]^T$. This contradicts that $i\in T_2$. Therefore, for all $i\in T_2$,
\begin{equation}
\label{eqn:small-matrices-DE-to-cut-3-two}
\{1,2\}=\calD_1(A[\{1,2,3,i\}], B[\{1,2,3,i\}])\ \ \text{ and }\ \ \{3,i\}=\calD_2(A[\{1,2,3,i\}], B[\{1,2,3,i\}]).
\end{equation}

\cref{lem:small-matrices-DE-to-cut-4-base-case} and~\cref{eqn:small-matrices-DE-to-cut-3-two} imply that for all $i\in T_2$, $\{1,2\}$ is a cut of the matrix $A[\{1,2,3,i\}]$. Therefore, $\rank A[\{1,2\}, \{3, i\}]=1$. Since $A$ satisfies property $\prop$, for all $i\in T_2$, there exists $X_i\subset[n]$ such that $1,2\in X_i$, $3, i\in \comp{X}_i$ and $\rank A[X_i, \comp{X}_i]=1$. Consider the following sets $$X=\bigcap_{i\in T_2}X_i\:\: \text{ and }\:\: \comp{X}=\bigcup_{i\in T_2}\comp{X}_i.$$
From~\cref{prop:old-to-new-cut}, we obtain that 
\begin{equation}
\label{eqn:small-matrices-DE-to-cut-3-three}
\rank A[X, \comp{X}]=1 \:\:\text{ with }\:\: 1, 2\in X,\: \: \text{ and }\:\: T_2+3\subseteq \comp{X}.
\end{equation}

Observe that there exists a matrix $A'$ such that it is diagonally similar to $A$ and $A'[\nsubset{j}]=B[\nsubset{j}]$ for $j=1,2$. Additionally, instead of $A$, we can work with any matrix diagonally similar to $A$. Thus, without loss of generality, we may assume that $A[\nsubset{j}]=B[\nsubset{j}]$ for $j=1,2$. This implies that $$A[i,j]=B[i,j] \text{ for all } (i,j)\notin\{(1,2), (2,1)\}.$$

Let $C$ be an $n\times n$ diagonal matrix  such that for all $i\in[n]$, $C[i,i]=c_i$ with $c_1=1$, and the following hold:  
\begin{enumerate}
\item  $B[\nsubset{3}]^T=C[\nsubset{3}]^{-1}\cdot A[\nsubset{3}]\cdot C[\nsubset{3}]$
\item $B[\{1,2,3\}]^T=C[\{1,2,3\}]^{-1}\cdot A[\{1,2,3\}]\cdot C[\{1,2,3\}]$
\item for all $i\in T_1$, $B[\{1,2,3,i\}]^T=C[\{1,2,3,i\}]^{-1}\cdot A[\{1,2,3,i\}]\cdot C[\{1,2,3,i\}]$
\end{enumerate}
Observe that the existence of such an $C$ follows from $\calD_2(A,B)=\{3\}$, and~\cref{eqn:small-matrices-DE-to-cut-3-one}, and the definition of $T_1$. Thus, 
\begin{enumerate}
\item $B[1,2]=A[2,1]c_1c_2^{-1}$ and $B[2,1]=A[1,2]c_1^{-1}c_2$.
\item $A[j,i]=A[i,j]c_jc_i^{-1}$ for all $j\in\{4,5,\ldots, n\}$ and $i\in\{1,2,4,5,\ldots, n\}$.
\item $A[3,i]=A[i,3]c_3c_i^{-1}$ for all $i\in T_1\cup\{1,2\}$.
\end{enumerate}
This implies that for any $S\subseteq [n]$ with $1,2\in S$ and $T_2+3\subseteq \comp{S}$, $$A[\comp{S}, S]= C[\comp{S}]\cdot A[S, \comp{S}]^T\cdot C[S]^{-1}.$$
The above equation combined with~\cref{eqn:small-matrices-DE-to-cut-3-three} imply that $X$ is a cut of $A$ with $1,2\in X$ and $T_2+3\subseteq \comp{X}$. This completes the proof. 
 \end{proof}

Now, we prove~\cref{thm:main-characterization} for no-cut case. Specifically, we show that for any two square matrices $A$ and $B$ with $A$ has no cut  and satisfies property $\prop$, the corresponding principal minors of $A$ and $B$ are equal if and only if the corresponding principal minors of them of order up to $4$ are equal.

\begin{theorem}
\label{thm:PME-for-no-cut}
Let $A$ and $B$ be two $n\times n$ matrices over a field $\F$. Suppose that $A$ satisfies property $\prop$ and has no cut. Then $A\PME B$ if and only if $\kPME{A}{B}{4}$.
\end{theorem}
\begin{proof}
We prove it by induction on $n$. 
This is trivially true for $n = 4$. 

Assume that $n \ge 5$ and $\kPME{A}{B}{4}$. Following \cref{lem:partial-DS-to-complete-DS}, to prove $A\PME B$, it suffices to show 
\[
|\calD_1(A,B)\cup\calD_2(A,B)|\geq 5.
\]
From our assumption, $A$ satisfies property $\prop$ and $A$ does not have a cut. Applying~\cref{cor:no-cut-from-big-to-small-matrix}, there exist at least three distinct indices $i_1, i_2, i_3\in[n]$ such that the submatrices $A[\nsubset{i_j}]$ for each $j \in \{1,2,3\}$ do not have a cut. 
Furthermore, each $A[\nsubset{i_j}]$ continues to satisfy property $\prop$ following~\cref{prop:matrix-property-large-to-small}. Additionally, $\kPME{A[\nsubset{j}]}{B[\nsubset{j}]}{4}$. Therefore, from the induction hypothesis, $A[\nsubset{j}]\PME B[\nsubset{j}]$ for each $j\in \{1,2,3\}$. 
Since $A[\nsubset{j}]$ has no cut and all the off-diagonal entries are nonzero, from~\cref{lem:no-cut-to-PME}, $A[\nsubset{j}]\DE B[\nsubset{j}]$.
Equivalently,
\[
|\calD_1(A,B)\cup\calD_2(A,B)|\geq 3.
\]
However,~\cref{lem:small-matrices-DE-to-cut-3} and~\cref{lem:small-matrices-DE-to-cut-4} combined with~\cref{prop:empty-intersection-of-D_1-and-D_2} rule out the possibility of $|\calD_1(A,B)\cup\calD_2(A,B)|$ being equal to 3 and 4 respectively, as that would imply the matrix $A$ to have a cut which leads to a contradiction. 
\end{proof}

We now focus on proving the main technical result of this paper: for any two square matrices $A$ and $B$ with nonzero off-diagonal entries, if $A$ satisfies property~$\prop$, then $A$ and $B$ are principal minor equivalent if and only if their corresponding principal minors of order up to $4$ are equal. Our proof proceeds by induction on the size of the matrices. To facilitate this, the following lemma provides a decomposition of larger matrices into smaller ones, allowing us to apply the induction hypothesis to the smaller matrices. The core idea of this lemma already appears in~\cite{Chatterjee25}, although not in the precise form stated here.

\begin{lemma}[Decomposition lemma]
\label{lem:cut-decomposition}
Let $A$ and $B$ be two $n\times n$ matrices over a field $\F$ with nonzero off-diagonal entries. Let $S\subseteq [n]$ be a cut of both $A$ and $B$. Consider any two indices $s,t$ such that $s\in S$ and $t\in \comp{S}$. Then $A \PME B$ if and only if $A[S+t] \PME B[S+t]$ and $A[\comp{S} + s] \PME B[\comp{S} + s]$. Furthermore, if $S$ is a minimal cut of $A$, then $A[S+t]$ has no cut.
\end{lemma}

For a proof see~\cref{subsec:proof-of-lem:cut-decomposition}. We also require the following result, which establishes a relationship between the cuts of two matrices whose corresponding principal minors of order up to 4 are equal.

\begin{lemma}
\label{lem:common-cut}
Let $A$ and $B$ be two $n \times n$ matrices with nonzero off-diagonal entries and $\kPME{A}{B}{4}$. Let $S$ be a minimal cut of $A$ of size greater than $2$. Then, $S$ is also a cut of $B$.
\end{lemma}

A closely related result appears in~\cite[Lemma~3.6]{Chatterjee25}, with the key difference being the strength of the hypothesis: while~\cite{Chatterjee25} assumes $A \PME B$, the above lemma only requires the weaker assumption $\kPME{A}{B}{4}$. Nonetheless, the proof of the above lemma follows the same general outline as that of~\cite[Lemma~3.6]{Chatterjee25}, with an additional reliance on~\cref{thm:PME-for-no-cut}. For a complete proof, refer to~\cref{subsec:proof-of-lem:common-cut}. With all necessary tools now in place, we are ready to formally prove~\cref{thm:main-characterization}.


\begin{proof}[Proof of~\cref{thm:main-characterization}]
Property $\prop$ ensures that all off-diagonal entries of $A$ are nonzero. Next, we show that all off-diagonal entries of $B$ are also nonzero. For the sake of contradiction, assume that $B[i,j]= 0$ for some $i\neq j\in[n]$. Then, $\det(A[\{i,j\}])\neq \det(B[\{i,j\}])$, contradicting $\kPME{A}{B}{4}$. Therefore, all the off-diagonal entries of $B$ are nonzero.

We prove the theorem by induction on $n$. 
This trivially holds for $n = 4$. If $A$ does not have a cut, the statement follows from \cref{thm:PME-for-no-cut}.

For the inductive step, suppose that $\kPME{A}{B}{4}$ and $A$ has a cut and let $S\subseteq [n]$ be a minimal cut of $A$. Now we divide our proof into the following two cases.

\paragraph*{Case I ($S$ is a cut of $B$).} Following~\cref{lem:cut-decomposition}, it suffices to prove $A[S+t] \PME B[S+t]$ and $A[\comp{S} + s] \PME B[\comp{S} + s]$ for some $s\in S$ and $t\in \comp{S}$. Notice that both $A[S+t]$ and $A[\comp{S}+s]$ continue to satisfy property $\prop$ following~\cref{prop:matrix-property-large-to-small}. Additionally, $\kPME{A[S+t]}{B[S+t]}{4}$ and $\kPME{A[\comp{S}+s]}{B[\comp{S}+s]}{4}$. The proof now follows from our induction hypothesis.

\paragraph*{Case II ($S$ is not a cut of $B$).} From~\cref{lem:common-cut}, the cardinality of $S$ must be two. Then, applying~\cref{lem:Cut-and-PME-upto-4}, there exists a cut $X$ of $B$ such that $S$ is a cut of $\widetilde B=\tw(B, X)$. Since the principal minors are preserved under the cut-transpose operation (see~\cref{lem:PME-under-cut-transpose}), $B\PME \widetilde B$ and $\kPME{A}{\widetilde B}{4}$. Additionally, $S$ is now also a cut of $\widetilde B$. Therefore, invoking \textbf{Case I}, we obtain $A\PME \widetilde B$. Thus, $A\PME B$. 
\end{proof}

\paragraph{Can property  $\prop$ be dropped?}
Here, we address the following question. Suppose that $A$ and $B$ are two $n\times n$ matrices over a field $\F$ with nonzero off-diagonal entries. From~\cref{thm:main-characterization}, if $A$ has property $\prop$, then $A\PME B$ if and only if $\kPME{A}{B}{4}$. A natural question is whether we can drop property $\prop$ while still guaranteeing the same conclusion as in~\cref{thm:main-characterization}? We answer this question in the negative. More specifically, we construct a pair of matrices $A$ and $B$ such that $\kPME{A}{B}{4}$, but $A$ is \emph{not} principal minor equivalent to $B$. For details, see~\cref{appendix:necessity-of-Q}.

\section{Black-Box Cut Finding} \label{sec: cut discovery}
\paragraph{Notation.} We begin by introducing the notation that will be used in this and the following section.
 Suppose that $B\in\F^{4\times 4}$ is a matrix with nonzero off-diagonal entries. Then, $\fourPMEfamily{B}$ is the set of all $4\times 4$ matrices over $\F$ such that 
\begin{enumerate}
\item for any matrix $C\in \fourPMEfamily{B}$, $C[1, j]=1$ for $j=2,3,4$, and
\item for any $4\times 4$ matrix $D$ with $D\PME B$, there exists a matrix $C\in \fourPMEfamily{B}$ such that $D\DS C$.
\end{enumerate}
By~\cite[Lemma~3.1]{Chatterjee25} and~\cref{obs:uniqueness-of-DS}, the family $\fourPMEfamily{A}$ is finite. Later, in~\cref{subsec:computing-upto-four-order-matrices}, we present an algorithm for computing $\fourPMEfamily{A}$ given access to the principal minors of $A$. Suppose that $I$ is a finite set. Let $A$ be a matrix over $\F$ whose rows and columns are indexed by $I$, and off-diagonal entries of $A$ are nonzero. Then, $$\submatrixfamily{A}=\left\{(T,\fourPMEfamily{A[T]})\,\mid\,T\subseteq I \text{ with } |T|=4\right\}.$$

\subsection{Cut Finding Algorithm}
Suppose that $A$ is an $n\times n$ matrix over a field~$\F$ satisfying property~$\prop$. Our goal is to determine whether $A$ has a cut in a \emph{black-box manner}; that is, instead of direct access to the entries of~$A$, we are allowed to query its principal minors. Furthermore, if $A$ has a cut, the algorithm also outputs a subset $S \subseteq [n]$ with $2 \le |S| \le n - 2$ such that $S$ is a cut in some matrix $C \in \F^{n\times n}$ satisfying $C \PME A$.

We emphasize that, we later (in~\cref{subsec:computing-upto-four-order-matrices}) provide an algorithm which, for any $I \subseteq [n]$ with $|I| = 4$, given access to the principal minors of~$A[I]$, computes the family~$\fourPMEfamily{A[I]}$. Consequently, we can obtain the family~$\submatrixfamily{A}$. Access to this family is sufficient for our purpose. Moreover, this algorithm will be used as a subroutine in a later reconstruction algorithm, where $\submatrixfamily{A}$ is computed explicitly from the principal minors of~$A$. Therefore, in the description of the cut-finding algorithm, we assume that $\submatrixfamily{A}$ is given as input, rather than access to the principal minors. 

We now define a property for the $4\times 4$ principal submatrices of an $n \times n$ matrix $A$. It will be crucial in designing our algorithm.
\begin{definition}[Property $\calP$]
\label{def:property-P}
    Given a matrix $A$, a set of two pairs $t=\{\{i,j\},\{k,l\}\}$ with distinct $i,j,k,l$ is said to satisfy property $\calP$ if there exists a $4\times 4$ matrix $C_t\PME A[\{i,j,k,l\}]$ such that $\{i,j\}$ (equivalently $\{k,l\}$) is a cut of $C_t$. 
\end{definition}

The following observation establishes that verifying whether a pair satisfies property $\calP$ can be done efficiently. 
\begin{observation}
\label{obs:checkP} 
For any matrix $A\in\F^{n\times n}$, given the family $\submatrixfamily{A}$ and a pair $t=\{\{i, j\}, \{k, l\}\}$, we can check whether $t$ satisfies property $\calP$ in constant time.
\end{observation}

Intuitively, the above definition captures a certain kind of \emph{local} property that is applicable to only the $4\times 4$ principal submatrices of a large $n\times n$ matrix $A$. Next, using this local property, we define a global object for $A$. It will be helpful to characterize whether a matrix $A$ satisfying property $\prop$ has cut.
\begin{definition}[Plausible set]
\label{def:plausibleSet}
Given an $n \times n$ matrix $A$, a subset $S\subseteq [n]$ of size at least two and at most $n-2$ is called a \emph{plausible set}, if for each $\{i,j\}\subset S, \{k,l\}\subset, \comp{S}$, $\{\{i,j\},\{k,l\}\}$ satisfies property $\calP$.
\end{definition}

Next, we describe two lemmas that are backbone in designing our algorithm. We defer the proof of these lemmas to the end of the section. Suppose that $A$ is a matrix that satisfies property $\prop$. Then, the following lemma establishes a transitivity relation on the set of all pairs $\{\{i, j\}, \{k, l\}\}$ satisfying property $\calP$. Intuitively, given two such pairs that share common indices, one can derive a new pair that also satisfies~$\calP$.
\begin{lemma}[Transitivity of property $\calP$]
\label{lem:PropPtransitive}
Let $A$ be an $n\times n$ matrix that satisfies property $\prop$. Let $i,j,k,l_1,l_2\in [n]$ be distinct elements such that both $\{\{i,j\},\{k,l_1\}\}$ and $\{\{i,j\},\{k,l_2\}\}$ satisfy property $\calP$, then $\{\{i,j\},\{l_1,l_2\}\}$ also satisfies property $\calP$.
\end{lemma}
For proof of the above lemma, see~\cref{subsec:PropPtransitive}. The following observation presents an analogous result for cuts instead of plausible sets and is straightforward to verify.
\begin{observation}\label{obs:cutTransitive}
Let $A$ be an $n\times n$ matrix with non-zero off diagonal entries. Let $i,j,k,l_1,l_2\in [n]$ be distinct elements such that both $A[\{i,j,k,l_1\}]$ and $A[\{i,j,k,l_2\}]$ have the cut $\{i,j\}$. Then, $\{i,j\}$  is also a cut of $A[\{i,j,l_1,l_2\}]$. In other words, $\{i,j\}$ is a cut of $A[\{i,j,k,l_1,l_2\}]$.
\end{observation}

Suppose that $A$ is a matrix satisfying property $\prop$. Suppose that instead of the entries of $A$, we have access to the principal minors of $A$. Then, $A$ has a cut if and only if any matrix principal minor equivalent to $A$ has a cut. It follows from~\cite[Theorem~1.1]{Chatterjee25}. Thus, the existence of cut in $A$ is a global property over all the matrices principal minor equivalent to $A$. However, that is not true for the question of whether a particular subset $S$ is cut in $A$. It may happen that $S$ is not a cut in $A$, but is a cut in some other matrix $C$ such that $A\PME C$.  It follows from~\cite[Theorem~1.1 and Lemma~3.2]{Chatterjee25}. Thus, given access to only the principal minors of $A$, it is expected that we would not be able to output a cut. Fortunately, for our purpose, it suffices to find a subset $S$ such that $S$ is cut in some matrix $C$ with $C\PME A$, and plausible sets exactly characterize such sets. We first note the following immediate observation.
\begin{observation}\label{obs:PlausibleImpCut}
For an $n\times n$ matrix $A$ with $S\subset [n]$ as a cut, $S$ is also a plausible set for $A$.
\end{observation}

\begin{lemma}
\label{lem:CutChecking}
Let $A$ be an $n\times n$ matrix that satisfies property $\prop$. Then, a subset $S\subset [n]$ is plausible if and only if it is a cut of some matrix $C$ such that $C\PME A$.
\end{lemma}
For proof of the above lemma, see~\cref{subsec:PlausibleIsCut}. 

\subsubsection{A description of cut finding algorithm and its correctness} 
\label{subsec:DesCutFind}
For a pseudocode of the blackbox cut finding algorithm, see~\cref{algo:cut-detection-finding-algorithm}.


\begin{algorithm}[H]
\caption{Cut finding algorithm}\label{algo:cut-detection-finding-algorithm}
\textbf{Input:} A set $I$ with $|I|\geq 4$; the family $\submatrixfamily{A}$ for some matrix $A$ whose rows and columns are indexed by $I$, where $A$ satisfies property $\prop$.\\
\textbf{Output:} `NO', if $A$ does not have a cut. Otherwise, output a $S\subseteq I$ such that $S$ is a cut in some matrix $C$ with $A\PME C$. 
\begin{algorithmic}[1]
\State Let $|I|=n$, and without loss of generality, assume $I=[n]$.
\For{$t\leftarrow \pp{i,j}{k,l}$ with  distinct $i,j,k,l\in [n]$}
    \If{$t$ satisfies property $\calP$} \label{step:correctInit}
        \State For each $e\in [n]\setminus \{i,j,k,l\},$ let $x_e$ denote a boolean variable.
        \State Let $\phi_t$ be a 2-SAT instance initialized with zero clauses.
        \For{$e\xleftarrow{} [n]\setminus t$}
            \If{either $\pp{i,j}{k,e}$ or $\pp{i,j}{l,e}$ does not satisfy property $\calP$}
                \State $\phi_t\xleftarrow{}\phi_t\wedge (x_e=\ \mathrm{True})$\label{step:eInS}
            \EndIf
        \EndFor
        \For{$e\in [n]\setminus t$}
            \If{either $\pp{i,e}{k,l}$ or $\pp{j,e}{k,l}$ does not satisfy property $\calP$}
                \State $\phi_t\xleftarrow{}\phi_t\wedge (x_e=\ \mathrm{False})$ \label{step:eNotInS}
            \EndIf
        \EndFor
        \For{$(p,q)$ with $p\neq q$ and $p,q\in [n]\setminus\{i,j,k,l\}$, $a\in \{i,j\}, c\in \{k,l\}$ }
            \If{ $\pp{a,p}{c,q}$ does not satisfy property $\calP$}
                \State $\phi_t\xleftarrow{} \phi_t \wedge (x_q\vee \neg x_p)$. \label{step:pEqualq}
            \EndIf
        \EndFor
        \If{$\phi_t$ is satisfiable}
            \State $x^*\xleftarrow{}$ a satisfying assignment for  2-SAT $\phi_t$.
            \State $S\xleftarrow{}\{i,j\}\cup \left\{e\in [n]\setminus\{i,j, k, l\}\mid x^*_e=\ \mathrm{True}\right\}$.
            \State \Return $S$.
        \EndIf
    \EndIf
\EndFor
\State \Return `NO'
\end{algorithmic}
\end{algorithm}
        
\begin{remark} 
\label{rem:cut-finding}
We make the following two remarks about~\cref{algo:cut-detection-finding-algorithm}.
\begin{enumerate}
\item \label{rem:ExplicitCutFinding}
 Given an explicit irreducible matrix $A$, \cref{algo:cut-detection-finding-algorithm} can be adapted to find a cut of $A$ with a minor modification. Specifically, in every instance where the algorithm verifies if property $\calP$ holds for a pair $\pp{i,j}{k,l}$, we instead check whether $A[\{i,j,k,l\}]$ has $\{i,j\}$ as a cut. The proof of correctness is very similar to that of \cref{algo:cut-detection-finding-algorithm}, as discussed in \cref{lem:CorrectnessOfCutFinding}.
\item \label{rem:MinimalCutFinding}
A plausible set $S$ is \emph{minimal} if none of its proper subsets are plausible. If we can find a satisfying assignment for the 2-SAT instance with a minimal number of $1$s, then it would correspond to a minimal plausible set $S$, and thus by \cref{lem:CutChecking}, the cut any $S$ in $C \PME A$ would be minimal. If there existed a matrix $C' \PME A$ with cut $X \subset S$, then $X$ would be a plausible set, contradicting the fact that $S$ is minimal. Given a satisfying assignment for a 2-SAT instance, we can easily find a minimal satisfying assignment by flipping $1$s to $0$s, one by one, and checking satisfiability. This means \cref{algo:cut-detection-finding-algorithm} can be adapted to output a minimal plausible set, and therefore, a minimal cut in some $C \PME A$.
\end{enumerate}
 \end{remark} 

 The following lemma discusses the proof of correctness and running time of \cref{algo:cut-detection-finding-algorithm}.
\begin{lemma}\label{lem:CorrectnessOfCutFinding}
Let $A$ be an $n\times n$ matrix that satisfies property $\prop$. If $A$ has a cut, then \cref{algo:cut-detection-finding-algorithm} outputs a cut $S$ of some matrix $C$ such that $C\PME A$; otherwise it outputs `NO', in $\poly(n)$ time.
\end{lemma}
\begin{proof}
First, we show that if $A$ has a cut, then the algorithm does not output `NO'. Let $S$ be some cut of $A$. Without loss of generality, assume that $\{1,2\}\subset S,\{3,4\}\subset \comp{S}$. Then, consider the assignment  $x$ such that for each $e\in [n]\setminus [4]$, \[x_e=\begin{cases}\text{True} & \text{if } e\in S\\
\text{False} & \text{otherwise}\end{cases}.\] We will argue that  $x$ satisfies the 2-SAT $\phi_t$ for $t=\pp{1,2}{3,4}$. 

Suppose, for the sake of contradiction, that $x$ does not satisfy $\phi_t$. This implies that $x$ does not satisfy one of the clauses of $\phi_t$. 
\begin{description}
\item [Case 1:] Suppose that it does not satisfy the clause in line \ref{step:eInS}. This implies that there exists $e\notin S$ such that $\pp{1,2}{3,e}$ or $\pp{1,2}{4,e}$ does not satisfy  property $\calP$. However, this is a contradiction as $S$ is a cut of $A$ (and hence a plausible set from \cref{obs:PlausibleImpCut}) such that $\{1,2\}\subseteq S$ and $\{3,4,e\}\subseteq \overline{S}$. 

\item [Case 2:] Similarly to \emph{Case 1}, we can argue that $x$ satisfies the clauses of line \ref{step:eNotInS}. 

\item [Case 3:] We show that $x$ satisfies the clauses of line \ref{step:pEqualq}. Suppose that this is not true. This implies that there exist some $p,q\in [n]\setminus [4]$ such that $q\in \comp{S},p\in {S}$ and $a\in \{1,2\},b\in \{3,4\}$ such that $\pp{a,p}{q,l}$ does not satisfy property $\calP$. However, this is a contradiction as $S$ is a cut of $A$ and $\{1,2,p\}\subseteq S$ and $\{3,4,q\}\subseteq \comp{S}$. 
\end{description}
Hence, $x$ satisfies all clauses of~$\phi_t$, and therefore the algorithm does not output `NO' when~$A$ has a cut.

Now, we show that if for any set of two pairs $\pp{i,j}{k,l}$ the algorithm outputs some set $S$, then $S$ is a plausible set. By~\cref{lem:CutChecking}, this would imply that $S$ is a cut of some matrix $C $ such that $C\PME A$. To prove that $S$ is a plausible set, we must show that for each  $\{a,b\} \subset S, \{c,d\}\subset \comp{S}$, $\pp{a,b}{c,d}$ satisfies property $\calP$. We analyze this by considering cases based on the size of the intersection $T = \{i,j,k,l\}\cap \{a,b,c,d\}$. By construction, $\{i,j\} \subseteq S$ and $\{k,l\} \subseteq \comp{S}$. The case $|T|=4$ follows directly from line \ref{step:correctInit} of the algorithm.
\begin{enumerate}
\item \label{item:cut_algo_case1} $|T|=3$ : Without loss of generality, let $\{a,b\}=\{i,j\}$ and $c=k, d\neq l$. Since $x^*$ satisfies $\phi_t$ and $x^*_d$ is False, $\pp{i,j}{k,d}$ satisfies property $\calP$ otherwise clause in line  \ref{step:eInS} is violated.
\item  $|T|=2$ :  This can be divided further into the following subcases (up to reindexing of elements). 
\begin{enumerate}
\item  $\{a,b\}=\{i,j\}$ and $\{c,d\}\cap \{k,l\}=\emptyset$ : From case ~\ref{item:cut_algo_case1}, both $\pp{i,j}{k,c}$ and $\pp{i,j}{k,d}$ satisfy property $\calP$. Hence, from \cref{lem:PropPtransitive}, $\pp{i,j}{c,d}$ satisfies property $\calP$. 
    \item \label{item:cut_algo_case2b} $a=i, c=k$ : Since $x^*_b=$True and $x^*_d=$False, $\pp{i,b}{k,d}$ satisfies property $\calP$ otherwise, the clause in line \ref{step:pEqualq} is violated. 
\end{enumerate}
\item \label{item:cut_algo_case3} $|T|=1$ : Without loss of generality, let $a=i$ and $\{j,k,l\}\cap\{b,c,d\}=\emptyset$. From case \hyperref[item:cut_algo_case2b]{2(b)}, $\pp{i,b}{k,c}$ and $\pp{i,b}{k,d}$ satisfy property $\calP$. Hence, by \cref{lem:PropPtransitive}, $\pp{i,b}{c,d}$ satisfies property $\calP$. 

\item $|T|=0$ : From case \hyperref[item:cut_algo_case3]{3}, $\pp{i,a}{c,d}$ and $\pp{i,b}{c,d}$ satisfy property $\calP$. Hence, from \cref{lem:PropPtransitive}, $\pp{a,b}{c,d}$ satisfies property $\calP$.
\end{enumerate}

We now analyze the running time of the algorithm. For any set~$t$ of two pairs, the algorithm constructs the corresponding formula~$\phi_t$ in polynomial time. This is because verifying whether property~$\calP$ holds can be done in constant time (by~\cref{obs:checkP}), and only polynomially many such checks are required during the construction of~$\phi_t$. Since each~$\phi_t$ is a 2-SAT instance of polynomially many clauses and there are $O(n^4)$ possible choices of two pairs~$t$, the overall running time of the algorithm is polynomial in~$n$.
\end{proof}

\subsubsection{Proof of \cref{lem:PropPtransitive}}
\label{subsec:PropPtransitive}
We start by stating two lemmas from \cite{Chatterjee25} that we use later in the proof. 
\begin{lemma}{\cite[Lemma 3.1]{Chatterjee25} }\label{lem:4PME}
    Let $A$ be a $4\times 4$ matrix with non-zero off diagonal entries and $B\PME A$ be another matrix. Then, either $A\DE B$ or there exists a common cut $S$ such that $\ct(A,S)\DE B$. 
\end{lemma}
\begin{lemma}{\cite[Lemma 3.2]{Chatterjee25}}
    \label{thm:cutmap}
Let $A$ be an $n\times n$ matrix over $\F$ with nonzero off-diagonal entries. Let $S\subseteq [n]$ be a cut in $A$. Then, for any $T\subseteq [n]$ the following holds.
\begin{enumerate}
\item If $T$ or $\comp{T}$ is a subset of $S$ or $\comp{S}$, then $T$ is a cut in $A$ if and only if $T$ is a cut in $\tw(A, S).$
\item Otherwise, $T$ is a cut in $A$ if and only if $T\Delta S$ is a cut in $\tw(A,S)$.
\end{enumerate}
\end{lemma}
\begin{observation}\label[observation]{obs:FourMatP}
    Let $A$ be an $n\times n$ matrix with non-zero off-diagonal entries. If $\pp{i,j}{k,l}$ satisfies property $\calP$, then one of the following must be true for  $A[\{i,j,k,l\}]$:
    \begin{enumerate}
        \item $\{i,j\}$ is a cut. \label{eq:PropPtransitive-1}
        \item $\{i,k\}$ and $\{i,l\}$ both are cuts but $\{i,j\}$ is not a cut. \label{eq:PropPtransitive-2}
    \end{enumerate}
    \end{observation}
    \begin{proof}
        Let $A'=A[\{i,j,k,l\}]$. By definition, there exists a $4\times 4$ matrix $C'\PME A' $ with $\{i,j\}$ as a cut. From \cref{lem:4PME}, either $A'$ is diagonally equivalent to $C'$ or there exists a common cut, say $S$, such that, $A'$ is diagonally equivalent to $\ct(C',S)$. If it is the former case, then $A'$ has cut $\{i,j\}$. Suppose it is the latter case and $A'$ does not have cut $\{i,j\}$. Without loss of generality, let the common cut $S=\{i,k\}$. Since $\{i,j\}$ is a cut of $C'$, from \cref{thm:cutmap}, $\{i,k\}\Delta\{i,j\}=\{j,k\}$, equivalently $\{i,l\}$, is also a cut of $\ct(C',S)$ and hence of $A'$.  
    \end{proof}
Now, we come back to the proof of \cref{lem:PropPtransitive}. Without loss of generality, let $i=1,j=2,k=4,l_1=3,l_2=5$.  We prove this by considering following three cases based on which conditions $A_1=A[\{1,2,3,4\}]$ and $A_2=A[\{1,2,4,5\}]$ satisfy from \cref{obs:FourMatP}. Let $A_3=A[\{1,2,3,5\}]$.
\begin{enumerate}
    \item $A_1$ and $A_2$ both satisfy \cref{eq:PropPtransitive-1}. From \cref{obs:cutTransitive}, $\{1,2\}$ is a cut of $A[[5]]$. Hence, $A_3$ has cut $\{1,2\}$ implying $\pp{1,2}{3,5}$ satisfies $\calP$.
    \item $A_1$ and $A_2$ both satisfy \cref{eq:PropPtransitive-2}. This implies $\{1,4\}$ is a cut of both $A_1$ and $A_2$. Hence, $\{1,4\}$ is a cut of $A[[5]].$ Since $A_1$ and $A_2$ satisfy \cref{eq:PropPtransitive-2}, $\{1,3\}$ and $\{1,5\}$ are cuts of $A_1$ and $A_2$, respectively. This in turn implies that both of them have $\{2,4\}$ as a cut. Hence, from \cref{obs:cutTransitive}, $\{2,4\}$ is also a cut of $A[[5]]$. Let $C=\ct(A[[5]],\{1,4\})$. Then, from \cref{thm:cutmap}, $C\PME A[[5]]$ and has $\{1,4\}\Delta \{2,4\}=\{1,2\}$ as a cut. This implies $C[\{1,2,3,5\}]\PME A_3$ and has $\{1,2\}$ as a cut. Hence, $\pp{1,2}{3,5}$ also satisfies property $\calP$.
    \item $A_1$ satisfies \cref{eq:PropPtransitive-1} and $A_2$ satisfies \cref{eq:PropPtransitive-2}. Since $A$ satisfies property $\prop$, so does $A[[5]]$. Since $A[\{1,2\},\{3,4\}]$ has rank one, either $A[\{1,2,5\},\{3,4\}]$ or $A[\{1,2\},\{3,4,5\}]$ must have rank one. Similarly, $A[\{3,4,5\},\{1,2\}]$ or $A[\{3,4\},\{1,2,5\}]$ must have rank one as $A[\{3,4\},\{1,2\}]$ has rank one. This leads us to the following four subcases:
    \begin{enumerate}
        \item $A[\{1,2\},\{3,4,5\}]$ and $A[\{3,4,5\},\{1,2\}]$ have rank one. This implies $A[[5]]$ has $\{1,2\}$ as a cut. Hence, $\pp{1,2}{3,5}$ satisfies $\calP$.
        \item\label{item:cut-detec-lem1-case3b} $A[\{1,2,5\},\{3,4\}]$ and $A[\{3,4,5\},\{1,2\}]$ have rank one. Since $A$ has  non-zero off-diagonal entries, there exists a matrix diagonally similar to $A$ that has all ones in the first row except the diagonal entry. Hence,  without loss of generality, we can assume that $A$ has all ones in the first row except the diagonal entry. Let us assume that $A[2,3]=a,$ and $A[(3,4,5),1]=(b,c,d)^T$. Now, we show that $A[[5]]$ looks like the following matrix using known rank one submatrices where $\alpha$ is some non-zero scalar. This in turn implies that $\{1,2\}$ is a cut of $A[[5]]$ implying $\pp{1,2}{3,5}$ satisfies property $\calP$.
        \[A[[5]]= \begin{bmatrix}
            * & 1 & 1& 1 & 1\\
            a/\alpha & * & a & a & a\\
            b & \alpha b & * & * & *\\
            c & \alpha c & * & * & \alpha c\\
            d & \alpha d & \alpha d & \alpha d & *
        \end{bmatrix}\]
        Since $A[\{3,4,5\},\{1,2\}]$ has rank one, there exists a non-zero  scalar $\alpha$ such that $A[(3,4,5),2]=\alpha(b,c,d)^T$. Since $A_2$ satisfies \cref{eq:PropPtransitive-2}, $A[\{1,5\},\{2,4\}]$ and $A[\{1,4\},\{2,5\}]$ both have rank one. Hence, $A[5,4]=\alpha d$ and $A[4,5]=\alpha c$.  Since $A[\{1,2,5\},\{3,4\}]$ has rank one, $A[2,4]=a, A[5,3]=\alpha d$. Since $A_2$ satisfies \cref{eq:PropPtransitive-2}, $A[\{2,5\},\{1,4\}]$ has rank one. This implies $A[2,1]=\frac{A[5,1]A[2,4]}{A[5,4]}= a/\alpha$. Finally, since $A[\{2,4\},\{1,5\}]$ has rank one, $A[2,5]=\frac{A[2,1]A[4,5]}{A[4,1]}=\frac{a/\alpha\cdot \alpha c}{c}=a.$ This completes the argument why $A[[5]]$ appears as above.
        \item $A[\{1,2\},\{3,4,5\}]$ and $A[\{3,4\},\{1,2,5\}]$ have rank one. This reduces to case \hyperref[item:cut-detec-lem1-case3b]{3(b)} if we substitute $A$ with its transpose.
        \item $A[\{1,2,5\},\{3,4\}]$ and $A[\{3,4\},\{1,2,5\}]$ have rank one. This implies $A[\{1,3,4,5\}]$ has $\{1,5\}$ as a cut. Also, it is given that $A_2=A[\{1,2,4,5\}]$ has $\{1,5\}$ as a cut. This implies $A[[5]]$ has $\{1,5\}$ as a cut. Similarly, since $A[\{2,3,4,5\}]$ and $A_2$ have $\{2,5\}$ as a cut, $A[[5]]$ also has $\{2,5\}$ as a cut. Let $C=\ct(A,\{2,5\})$. From \cref{thm:cutmap}, $\{1,5\}\Delta \{2,5\}=\{1,2\}$ is a cut of $C\PME A$. This implies $C[\{1,2,3,5\}]\PME A_3$ and has $\{1,2\}$ as a cut. Hence, $\pp{1,2}{3,5}$ also satisfies property $\calP$. 
    \end{enumerate}
\end{enumerate}

\subsubsection{Proof of \cref{lem:CutChecking}}\label{subsec:PlausibleIsCut}
One direction is trivial. If $S$ is some cut of a matrix $C\PME A$, then $S$ is also a plausible set. Now, we show the other direction. Suppose $S$ is some plausible set. We call a plausible set $S$ \emph{minimal} if there does not exist another plausible set $S'$ which is a proper subset of $S$. We prove \cref{lem:CutChecking} in three steps. First, we show it for plausible sets of size two, then we show it for minimal plausible sets and finally for any plausible set.

\paragraph{$S$ has size two.}
Suppose $S$ is a plausible set for $A$ of size $2$. Define a binary relation $\sim$ on $[n] \setminus S$ as follows:
\[ e \sim f \iff e=f \text{ or } A[S \cup \{e,f\}] \text{ has cut } S\text{.}\]
We see that $\sim$ is an equivalence relation on $[n] \setminus S$:
\begin{itemize}
    \item $e \sim e$ as $e = e$.
    \item $e \sim f \iff f \sim e$. 
    \item If $S$ is a cut of both $A[S \cup \{e,f\}]$ and $A[S \cup \{f,g\}]$, then by \cref{obs:cutTransitive}, $S$ must also be a cut of $A[S \cup \{e,g\}]$, thus $e \sim f\;\wedge\;f \sim g \implies e \sim g$.
\end{itemize}
For all $e \in [n] \setminus S$, let $T_e$ be the equivalence class of $e$, and $\comp{T}_e = ([n] \setminus S) \setminus T_e$.

Before showing the following claim, we discuss how it implies the existence of a matrix $C\PME A$ with $S=\{s_1,s_2\}$ as a cut. The claim states that either $A$ has cut $S$ or for all $e\in [n]-S$, $\{s_1\}\cup T_e$ and $\{s_2\} \cup T_e$ are cuts. In the former setting, $C=A$. In the latter setting, $C=\ct(A,\{s_1\}\cup T_e)$. From \cref{thm:cutmap}, $(\{s_1\}\cup T_e)\Delta (\{s_2\}\cup T_e)=S$ is a cut of $C$.

\begin{claim}\label[claim]{cl:PSsizetwo}
    Given an $n \times n$ matrix $A$ with non-zero off-diagonal entries, suppose $S = \{s_1, s_2\}$ is a plausible set for $A$. Then either $S$ is a cut of $A$ or for all $e \in [n] \setminus S$, both $\{s_1\} \cup T_e$ and $\{s_2\} \cup T_e$ are cuts of $A$.
\end{claim}

\begin{proof}
We prove by induction on $n$. Assume without loss of generality $S = \{1,2\}$. If $n=4$, we have two cases:
\begin{enumerate}
    \item If $3 \sim 4$, then $\{1,2\}$ is a cut of $A$.
    \item If $3 \not\sim 4$, then by \cref{obs:FourMatP}, $\{1,3\}$,$\{2,3\}$, $\{1,4\}$ and $\{2,4\}$ are all cuts of $A$.
\end{enumerate}
Now we handle the induction case for $n\geq5$.

\begin{claim}\label[claim]{cl:e-with-comp-geq-2} If for some $e \in [n] \setminus \{1,2\}$, $\card{\comp{T}_e} \geq 2$, then both $\{1\} \cup T_e$ and $\{2\} \cup T_e$ are cuts of $A$.
\end{claim}

\begin{proof}
For each $f \in \comp{T}_e$, note that $\{1,2\}$ remains a plausible set for $A[[n] - f]$ and since $f\in \comp{T}_e$, $T_e$ remains the same for $A[[n] - f]$ as well. By induction, one of the two is true:
\begin{enumerate}
    \item\label{item:e-with-comp-geq-2-item1} $\{1,2\}$ is a cut of $A[[n] - f]$.
    \item Both $\{1\} \cup T_e$ and $\{2\} \cup T_e$ are cuts of $A[[n] - f]$.
\end{enumerate}

Let $\{f,g\} \subseteq \comp{T}_e$. We now show that \cref{item:e-with-comp-geq-2-item1} cannot be true for any $h \in \comp{T}_e$; suppose $\{1,2\}$ were a cut of $A[[n]-h]$. Without loss of generality, let $h\neq g$.  Then $\{1,2\}$ would be a cut in $A[\{1,2,e,g\}]$, contradicting the fact that $g \in \comp{T}_e$. Thus, both $\{1\} \cup T_e$ and $\{2\} \cup T_e$ are cuts of $A[[n] - f]$ and $A[[n] - g]$, making both $\{1\} \cup T_e$ and $\{2\} \cup T_e$ cuts of $A$.

\end{proof}

If there exists $e \in [n] \setminus \{1,2\}$ such that $\card{\comp{T}_e} = 0$, then for all $f \in [n] \setminus \{1,2,e\}$, $A[\{1,2,e,f\}]$ has cut $\{1,2\}$. Thus, $\{1,2\}$ is a cut in $A$.
Now suppose $\card{\comp{T}_e} \geq 1$ for all $e \in [n] \setminus \{1,2\}$. Then for each $e \in [n] \setminus \{1,2\}$,
\begin{enumerate}
    \item If $\card{\comp{T}_e} \geq 2$, then by \cref{cl:e-with-comp-geq-2}, both $\{1\} \cup T_e$ and $\{2\} \cup T_e$ cuts of $A$.
    \item If $\card{\comp{T}_e} = 1$, let $\comp{T}_e = \{f\}$. Then $\comp{T}_f = T_e = [n] \setminus \{1,2,f\}$ and $\card{\comp{T}_f} = n -3 \geq 2$. Since $T_f = \{f\}$, by \cref{cl:e-with-comp-geq-2}, both $\{1,f\}$ and $\{2,f\}$ are cuts of $A$, which means both $\{1\} \cup T_e$ and $\{2\} \cup T_e$ are cuts of $A$.
\end{enumerate}

\end{proof}

\paragraph{S is minimal.} 
Now, we prove \cref{lem:CutChecking} for minimal plausible sets. by induction on $n$. Since any plausible set for a $4 \times 4$ matrix is of size $2$, the base case of $n=4$ is already proven. Suppose the lemma is true for minimal plausible sets for matrices of size $< n$. Let $A$ be a matrix with non-zero off-diagonal entries of size $n$ with a minimal plausible set $S$. Since we have already proved the lemma for $|S|=2$, without loss of generality, we can assume that $|S|,|\comp{S}|> 2$. In the following claim, we show that if $|\comp{S}|>2$, then for each $e\in \comp{S}$, $S$ is again a minimal plausible set for $A[[n]-e]$. By induction hypothesis, for each $e\in \comp{S}$, $S$ is a cut of some matrix $C_e\PME A[[n]-e]$. Since $S$ is a minimal plausible set for $A[[n]-e]$, from \cref{cl:minimalityPreservance}, it is also a minimal plausible set for $C_e$. By induction hypothesis, $S$ is a cut of $C_e$. Observe that minimal plausible sets must correspond to minimal cuts, implying $S$ is a minimal cut of $C_e$. Since $C_e\PME A[[n]-e]$ and $S$ has size $>2$, from  \cref{lem:common-cut}, $S$ is a cut of $A[[n]-e]$ as well. Since $|\comp{S}|>2$ and for each $e\in \comp{S}$, $A[[n]-e]$ has cut $S$, $S$ is also a cut of $A$. In fact, $S$ is a minimal cut of $A$.\\
\begin{claim}\label[claim]{cl:minimalityPreservance}
    Let $A$ be an $n\times n$ matrix with property $\prop$ and $S$ be a minimal plausible set of $A$ with $|\comp{S}|>2$. Let $e\in \comp{S}$. Then, $S$ is also a minimal plausible set of $A[[n]-e]$. 
\end{claim}
\begin{proof}
Clearly, $S$ is a plausible set for $A[[n]-e]$. Suppose $S$ not minimal for $A[[n]-e]$ and there exist some $X\subset S$ which is a plausible set for $A[[n]-e]$. Since $X$ is not a plausible set for $A$ and is a subset of $S$, there must exist some $\{i,j\}\subseteq X, l\in S\setminus X$ such that $\pp{i,j}{l,e}$ does not satisfy property $\calP$. Let $k\in \comp{S}-e$. Since $S$ is a plausible set for $A$, $\{i,j\}\subseteq X\subset S$ and $k,e\in \comp{S}$, $\pp{i,j}{k,e}$ satisfies property $\calP$. Similarly, since $X$ is a plausible set for $A[[n]-e]$, and $\{k,l\}\subset \comp{X}-e$, $\pp{i,j}{k,l}$ satisfies property $\mathcal{P}$. Since $A$ satisfies property $\prop$, by \cref{lem:PropPtransitive}, $\pp{i,j}{l,e}$ must also satisfy property $\calP$ which is a contradiction. Hence, $S$ is a minimal plausible set for $A[[n]-e]$.
\end{proof}

\paragraph{General S.} First, we state a result from \cite[Lemma 3.4]{Chatterjee25} and then show \cref{cl:reductionStep} which we use later in the proof. 
\begin{lemma} \label[lemma]{lem:Blowerhalf}
Let $A$ be an $n\times n$ matrix over $\mathbb{F}$ with nonzero off-diagonal entries. Let $X\subseteq [n]$ be a cut in the matrix $A$ and $s\in X$, and suppose $Y\subseteq \comp{X}$ is a cut in $A[\comp{X}+s]$. Then, $Y$ is also a cut in the matrix $A$.
\end{lemma} 
\begin{lemma}
\label{cl:reductionStep}
    Let $A$ be an $n\times n$ matrix with non-zero off diagonal entries and $S$ be a cut of $A$. Let $s\in S,t\in \comp{S}$ and $M$ and $N$ be two matrices such that $M\PME A[S+t]$ and $A[s+\comp{S}]\PME N$. Then, the matrix $A'$ defined as follows is PME to $A$.\[
    A'=\begin{blockarray}{ccc}
                    &  S             & \comp{S}\\
\begin{block}{c(cc)}
S                   & M[S]                                           & M[S, t]\cdot \frac{N[s,\comp{S}]}{N[s,t]}\\
&&\\
\comp{S}            & \frac{N[\comp{S}, s]}{N[t,s]}\cdot M[t, S] & N[\comp{S}]  \\
\end{block}
\end{blockarray}.
   \]
\end{lemma}
\begin{proof}
    Note that $A'[S+t]=M$, therefore $A[S+t]\PME A'[S+t]$. Now, we show $A'[s+\comp{S}]$ is diagonally similar to $N$. Precisely, we claim that $A'[s+\comp{S}]=DND^{-1}$ where \[ D[e,e]=\begin{cases}\frac{M[s,t]}{N[s,t]} & \text{ if } e=s\\
    1 & \text{ if } e\in\comp{S} \end{cases}.\] From the construction, it is easy to verify that $DND^{-1}[p,q]=A'[p,q]$ for each $p\in s+\comp{S}, q \in \comp{S}$. All that is left is to verify the the column corresponding to index $s$. From the principal minor equivalence of $A[s,t], M[s,t]$ and $N[s,t]$, we have $A[s] = N[s] = M[s]$ and $A[t] = N[t] = M[t]$, from which we infer that \[A[s,t]A[t,s]=M[s,t]M[t,s]=N[s,t]N[t,s].\] Observe that since $A[s,t]A[t,s]$ is nonzero, $N[s,t]$, $N[t,s]$, $M[s,t]$, and $M[t,s]$ must all be as well. Now, for $e\in \comp{S}$
    \[A'[e,s]=\frac{N[e,s]M[t,s]}{N[t,s]}=\frac{N[e,s]N[s,t]}{M[s,t]}=N[e,s]\frac{D[e,e]}{D[s,s]}.\]
    Since $S$ is a common cut of $A$ and $A'$, $A[S+t]\PME A'[S+t]$ and $A[s+\comp{S}]\PME A'[s+\comp{S}]$, by \cref{lem:cut-decomposition}, $A\PME A'$.
   
\end{proof}

Now, we come back to the proof for a general plausible set $S$. We show this by induction on $n$. The base case of $n=4$ is trivial as $S$ can only be a plausible set of size 2. Suppose the result is true for any matrices of size $<n$ for $n\geq 5$. Let $A$
be a matrix of size $n$ with property $\prop$ and $S$ be a plausible set for $A$. If $S$ is a minimal plausible set, then from above discussion, $S$ is indeed a cut of some matrix $C\PME A$. Suppose $S$ is not minimal and there exist a minimal plausible set $X$ which is a subset of $S$. This implies there exists a matrix $C\PME A$ for which $X$ is a cut. Let $s\in X$. Note that $\comp{S}$ is a plausible set for $A[s+\comp{X}]$ and hence even for $C[s+\comp{X}]$. By induction hypothesis, there exists some $N\PME C[s+\comp{X}]$ with $\comp{S}$ as a cut. Let $t\in S \setminus X$ and $M=C[X+t]$.  Let $C'$ be the matrix from \cref{cl:reductionStep} which is PME to $C$ and hence $A$. From the construction of $C'$, $X$ is a cut of $C'$ and $C'[s+\comp{X}]$ is diagonally similar to $N$. Since $N$ has cut $\comp{S}$, $C'[s+\comp{X}]$ also has cut $\comp{S}$ which is subset of $\comp{X}$. Hence, from \cref{lem:Blowerhalf}, $\comp{S}$ (and hence $S$) is a cut of $C'$ which PME to $A$. This completes the proof.

\section{PMAP for matrices with Property $\prop$: Proof of~\cref{thm:main-reconstruction-PM-with-Q-property}} \label{sec: PMAP for matrices with prop R}

In this section, we prove~\cref{thm:main-reconstruction-PM-with-Q-property}. More specifically, we design an algorithm such that given access to the principal minors of an $n\times n$ matrix $A$ with property $\prop$, in time $\poly(n)$, it outputs a matrix $B\in\F^{n\times n}$ such that $A\PME B$. We start by formally defining the principal minor oracle. 

\begin{definition}[Principal minor oracle]
\label{def:prinicpal-minor-oracle}
Let $A$ be an $n\times n$ matrix over a field $\F$. Then, $\PMoracle{A}$ denotes the principal minor oracle for the matrix $A$, that is, for any nonempty subset $S\subseteq [n]$ given as input to $\PMoracle{A}$, it outputs $\det(A[S])$.
\end{definition}
The following lemma describes that given access to the oracle $\PMoracle{A}$, some tasks that can be performed very efficiently.

\begin{lemma}
\label{lem:computation-using-PM-oracle}
Let $A$ be an $n\times n$ matrix over $\F$. Then, the following holds.
\begin{enumerate}
\item For any $i\in[n]$, $A[i,i]$ can be computed using one query to the oracle $\PMoracle{A}$.
\item For any $\{i, j\}\subseteq [n]$, $A[i,j]A[j,i]$ can be computed using three queries to the oracle $\PMoracle{A}$.
\item For any $\{i, j, k\}\subseteq [n]$, $A[i,j]A[j, k]A[k, i]+A[i, k]A[k, j]A[j, i]$ can be computed using seven queries to $\PMoracle{A}$.
\end{enumerate}
\end{lemma}
\begin{proof}
For all $i\in[n]$, we obtain $A[i,i]$ by querying $\PMoracle{A}$ on input $\{i\}$. For every $\{i,j\}\subseteq [n]$, $$A[i,j]A[j,i]=A[i,i]A[j,j]-\det(A[\{i,j\}]).$$ Therefore, $A[i,j]A[j,i]$ can be computed by querying $\PMoracle{A}$ on inputs $\{i\}$, $\{j\}$ and $\{i,j\}$. For every $\{i,j,k\}\subseteq [n]$ 
\begin{align*}
A[i,j]A[j,k]A[k,i]+A[i,k]A[k,j]A[j,i] &= \det(A[\{i,j,k\}])-\det(A[\{i,j\}])A[k,k]  \\
&\ \ \ -\det(A[\{j,k\}])A[i,i]-\det(A[\{i,k\}])A[j,j] \\
&\ \ \ +2A[i,i]A[j,j]A[k,k]
\end{align*}
Hence, we can compute $A[i,j]A[j,k]A[k,i]+A[i,k]A[k,j]A[j,i]$ using seven queries to $\PMoracle{A}$.
\end{proof}

We reuse the notation defined at the beginning of ~\cref{sec: cut discovery}. This section is organized as follows: In~\cref{subsec:computing-upto-four-order-matrices}, given access to principal minors of a $4\times 4$ matrix $A$ with nonzero off-diagonal entries, we describe how to compute the family $\fourPMEfamily{A}$. It will be used in other algorithms in this section. We reduce our learning problem for matrices with property $\prop$ to learning problem for matrices with property $\prop$ and have no cut. The corresponding algorithm is presented in~\cref{subsec:reconstruction-with-no-cut-&-Q}. Finally, in~\cref{subsec:reconstruction-with-Q}, we describe our complete learning algorithm for matrices with property $\prop$.

\subsection{Reconstruction of  $4\times 4$ matrices}
\label{subsec:computing-upto-four-order-matrices}
We begin with the following observation that provides a procedure to compute principal minor equivalent $2\times 2$ matrices.
\begin{observation}
\label{prop:reconstruction-size-two-matrix}
Let $A$ be a $2\times 2$ matrix over $\F$ with nonzero off-diagonal entries. Then, for any $2\times 2$ matrix $B$, $A\PME B$ if and only if $A[i,i]=B[i,i]$ for $i=1,2$ and $$B[1,2]B[2,1]=A[1,1]A[2,2]-\det(A).$$ 
\end{observation}

We now focus on learning of principal minor equivalent $3\times 3$ matrices. Suppose that $A\in\F^{3\times 3}$ with nonzero off-diagonal entries. By~\cref{obs:uniqueness-of-DS} and~\cref{obs:DE-to-PME}, there exists a $3\times 3$ matrix $B$ such that $B\PME A$ and $B[1,2]=B[1,3]=1$. The next lemma provides a way of computing such principal minor equivalent matrices. 
\begin{lemma}
\label{lem:reconstruction-size-three-matrix}
Let $A$ be a $3\times 3$ matrix over $\F$ with nonzero off-diagonal entries. Then, for any $3\times 3$ matrix $B$ with $B[1,2]=B[1,3]=1$, $A\PME B$ if and only if the following holds: 
\begin{enumerate}
\item $B[i,i]=A[i,i]$ for $i=1,2,3$, 
\item $B[j, 1]=A[1,j]A[j,1]$ for $j=2,3$, 
\item $B[2,3]$ is a root of the quadratic polynomial $az^2-bz+c$ where
 $$a=A[1,3]A[3,1],\ \ \ b=A[1,2]A[2,3]A[3,1]+A[1,3]A[3,2]A[2,1], \text{ and }$$  $$c=A[1,2]A[2,1]A[2,3]A[3,2],$$ 
\item $B[3,2]=\frac{A[2,3]A[3,2]}{B[2,3]}.$
\end{enumerate}
\end{lemma}
\begin{proof}
We first assume that $B\PME A$. Then, 
\begin{align}
A[i,i] &= B[i,i] \ \text{ for } i=1,2,3\label{eqn:reconstruction-size-three-matrix-one}\\
\det(A[\{i,j\}]) &= \det(B[\{i,j\}]) \text{ for } \{i, j\}\subseteq \{1,2,3\} \label{eqn:reconstruction-size-three-matrix-two}\\
\det(A) &=\det(B)\label{eqn:reconstruction-size-three-matrix-three}
\end{align}
From~\cref{eqn:reconstruction-size-three-matrix-one} and \cref{eqn:reconstruction-size-three-matrix-two}, we obtain that for all $\{i, j\}\subset\{1,2,3\}$, $$A[i,j]A[j,i]=B[i,j]B[j,i].$$
Furthermore, since $B[1,2]=B[1,3] = 1$, $B[2,1] = A[1,2]A[2,1]$ and $B[3,1]=A[1,3]A[3,1]$. This combined with~\cref{eqn:reconstruction-size-three-matrix-three}, we have 
\begin{align*}
0 &= A[1,2]A[2,3]A[3,1]+A[1,3]A[3,2]A[3,1]-B[1,2]B[2,3]B[3,1]-B[1,3]B[3,2]B[2,1]\\
&=A[1,2]A[2,3]A[3,1]+A[1,3]A[3,2]A[3,1]-B[2,3]A[1,3]A[3,1]-\frac{A[2,3]A[3,2]A[1,2]A[2,1]}{B[2,3]}
\end{align*}
This implies that $B[2,3]$ is a root of the equation $az^2-bz+c$, and by~\cref{eqn:reconstruction-size-three-matrix-two}, $B[3,2]=\frac{A[2,3]A[3,2]}{B[2,3]}$.

For the converse direction, $B$ is a matrix with $B[1,2]=B[1,3]=1$. Moreover, $$B[i,i]=A[i,i] \ \forall\,i\in\{1,2,3\} \text{ and } B[i,j]B[j,i]=A[i,j]A[j,i] \text{ for }\{i,j\}\subseteq \{1,2,3\}.$$ Therefore, $\kPME{A}{B}{2}$. Since off-diagonal entries of $A$ are nonzero, $c\neq 0$. Furthermore, $B[2,3]$ is a root of $az^2-bz+c$. Hence, $B[2,3]$ is nonzero. Additionally, $$aB[2,3]^2-bB[2,3]+c = 0.$$ Therefore,
\begin{align*}
b&=\left(A[1,3]A[3,1]\right)B[2,3]+\frac{A[1,2]A[2,1]A[2,3]A[2,3]}{B[2,3]}\\
&=B[2,3]B[3,1]+B[3,2]B[2,1] \qquad (\because\  A[i,j]A[j,i]=B[i,j]B[j,i])\\
&=B[1,2]B[2,3]B[3,1]+B[1,3]B[3,2]B[2,1] \qquad (\because\ B[1,2]=B[1,3]=1)
\end{align*}
This completes the proof.
\end{proof}

We next consider the problem of computing all matrices that are principal minor equivalent to a given $4\times4$ matrix.
Assume that $A\in\F^{4\times4}$ has nonzero off-diagonal entries and that we have access to all its principal minors. Our goal is to construct a family of matrices such that, for every matrix $C$ satisfying $C\PME A$, the family contains a matrix that is diagonally similar to~$C$.~\cref{algo:reconstruction-of-size-four-matrices} presents an algorithm for this task, and~\cref{lem:reconstruction-of-size-four-matrices} establishes its correctness. 
\begin{algorithm}[H]
\caption{Reconstruction of $4\times 4$ principal minor equivalent matrices}
\label{algo:reconstruction-of-size-four-matrices}
\textbf{Input:} An access to oracle $\PMoracle{A}$ for some matrix $A\in\F^{4\times 4}$ with nonzero off-diagonal entries\\
\textbf{Output:} The family $\fourPMEfamily{A}$ (for definition, see the first paragraph of \cref{sec: cut discovery})\\
\textbf{Assumption:} A root finding algorithm for quadratic univariate polynomials over $\F$ is available
\begin{algorithmic}[1]
\State For every $\{i,j\}\subseteq \{2,3,4\}$ with $i<j$, compute the univariate polynomial
\begin{equation}
\label{eqn:reconstruction-of-size-four-matrices-equation}
E_{i,j}(z)= a_{i,j}z^2-b_{i,j}z+c_{i,j}
\end{equation}
such that \begin{itemize}
\item $a_{i,j}=A[1,j]A[j,1]$
\item $b_{i,j}=A[1,i]A[i,j]A[j,1]+A[1, j]A[j, i]A[i,1]$
\item $c_{i,j}=A[1,i]A[i,1]A[i,j]A[j,i]$
\end{itemize}
\State For every $\{i,j\}\subseteq \{2,3,4\}$, compute the set of roots $R_{i,j}$ for the polynomial $E_{i, j}(z)$.
\State $\mathcal T\leftarrow \emptyset$
\State For every $(r_1, r_2, r_3)\in R_{2,3}\times R_{2, 4}\times R_{3,4}$ do the following:
\begin{enumerate}[label=\roman*)]
\item construct a $4\times 4$ matrix $B$ such that
\begin{itemize}
 \item $B[i,i]\leftarrow A[i,i]$ for all $i\in\{1,2,3,4\}$  
 \item $B[1, j]\leftarrow 1$ for $j=2,3,4$,
 \item $B[2,3]=r_1$, $B[2,4]=r_2$ and $B[3,4]=r_3$
 \item for $\{i, j\}\subseteq \{1,2,3,4\}$ with $i<j$, $B[j,i]\leftarrow\frac{A[i,j]A[j,i]}{B[i,j]}$
 \end{itemize}
\item If $A\PME B$, $\mathcal T\leftarrow \mathcal T\cup\{B\}$
\end{enumerate}
\State Output $\mathcal T$
\end{algorithmic}
\end{algorithm}
\begin{lemma}
\label{lem:reconstruction-of-size-four-matrices}
Let $A$ be a $4\times 4$ matrix with nonzero off-diagonal entries. Then, given the oracle access to $\PMoracle{A}$, \cref{algo:reconstruction-of-size-four-matrices} outputs the set $\fourPMEfamily{A}$ in constantly many operations over $\F$.
\end{lemma}
\begin{proof}
Let $B\in\F^{4\times 4}$ be a matrix such that $A\PME B$. Then, the off-diagonal entries of $B$ must be nonzero. Otherwise, for any $B[i,j]=0$, one can show that one of the following must be true:
\begin{enumerate}
\item $A[i,i]\neq B[i,i]$
\item $A[j, j]\neq B[j, j]$
\item $\det(A[\{i,j\}])\neq \det(B[\{i,j\}])$
\end{enumerate}
 Any of this would contradict the assumption $A\PME B$. 
 
 By~\cref{obs:uniqueness-of-DS}, we may assume without loss of generality that $B[1,j]=1$ for $j=2,3,4$. We now show that $B\in \fourPMEfamily{A}$. 

From~\cref{lem:reconstruction-size-three-matrix}, for every $\{i,j\}\subseteq \{2,3,4\}$ with $i<j$, $B[i,j]$ is a root of~\cref{eqn:reconstruction-of-size-four-matrices-equation} and $B[j, i]=\frac{A[i,j]A[j,i]}{B[i,j]}$. Thus, $B\in \fourPMEfamily{A}$. Moreover, by~\cref{lem:computation-using-PM-oracle}, we can compute~\cref{eqn:reconstruction-of-size-four-matrices-equation} using only a constant number of queries to the oracle $\PMoracle{A}$. 

The time complexity follows directly from the description of the algorithm. 
\end{proof}

\subsection{PMAP for matrices with property $\prop$ and no cut}
\label{subsec:reconstruction-with-no-cut-&-Q}

In this section, we describe our learning algorithm for matrices with $\prop$ property and have no cut. A detailed description of the algorithm is available in~\cref{algo:reconstruction-with-no-cut-&-Q}. The input consists of the index set $I$, access to the oracle $\PMoracle{A}$, and the family $\submatrixfamily{A}$ (see the beginning of \cref{sec: cut discovery} for the definition). The goal is to output a matrix $B$ such that $B\PME A$. We note that if we are given only access to $\PMoracle{A}$, then by~\cref{lem:reconstruction-of-size-four-matrices}, the family $\submatrixfamily{A}$ can be computed. However, since this algorithm is used as a subroutine in~\cref{algo:reconstruction-with-Q}, where the family is already computed, we include it as part of the input to avoid recomputation. 

\begin{algorithm}
\caption{Reconstruction of principal minor equivalent matrix with property $\prop$ and no cut}
\label{algo:reconstruction-with-no-cut-&-Q}
\textbf{Input:} A set $I$ with $|I|\geq 4$; the family $\submatrixfamily{A}$ and access to the oracle $\PMoracle{A}$  for some matrix $A\in\F^{n\times n}$ whose rows and columns are indexed by $I$, where $A$ has no cut and satisfies property $\prop$\\
\textbf{Output:} A matrix $B\in \F^{|I|\times |I|}$ such that $A\PME B$\\
\textbf{Assumption:} A root finding algorithm for quadratic univariate polynomials over $\F$ is available
\begin{algorithmic}[1]

\State Let$|I|=n$. Fix any two elements $i_1, i_2\in I$. \label{algo:reconstruction-with-no-cut-&-Q-fixing-indices}
\State Applying~\cref{algo:finding-no-cut-sequence} with input $(I, \submatrixfamily{A}, i_1, i_2)$, compute a sequence $(i_n,i_{n-1}, \ldots, i_{3})$. \label{algo:reconstruction-with-no-cut-&-Q-first-no-cut-sequence}

\State Construct the output matrix $B\in\F^{n\times n}$, whose rows and columns are indexed by $I$, as follows: 
\State $B[i_1,i_1] \leftarrow A[i_1,i_1]$, and $B[i_2,i_2] \leftarrow A[i_2,i_2]$
\State $B[i_1, i_2]\leftarrow 1$ and $B[i_2, i_1]\leftarrow A[i_1, i_2]A[i_2, i_1]$. 

\For{$j\leftarrow 3$ to $n$} \label{algo:reconstruction-with-no-cut-&-Q-first-loop}

\State $B[i_j,i_j]\leftarrow A[i_j, i_j]$

\State $B[i_1, i_j]\leftarrow 1$ and $B[i_j, i_1]\leftarrow A[i_1, i_j]A[i_j, i_1]$\label{algo:reconstruction-with-no-cut-&-Q-first-loop-first-row-entries}

\State Let $I_j=\{i_1, i_2, i_3,\ldots, i_j\}$. 

\State If $j=3$, $k_3\leftarrow i_2$, otherwise applying~\cref{algo:finding-no-cut-sequence} with input ($I_j$, $\submatrixfamily{A_{I_j}}$, $i_1$, $i_j$), \label{algo:reconstruction-with-no-cut-&-Q-second-no-cut-sequence}
\State \hspace{0.5 cm} compute a sequence $(k_j, k_{j-1}, \ldots, k_3)$.

\For{$\ell\leftarrow 3$ to $j$}\label{algo:reconstruction-with-no-cut-&-Q-second-loop}
\State Let $L_{\ell}=\{i_1, i_j, k_3, k_4, \ldots, k_\ell\}$
\State Let $R_{j, \ell}$ be the roots of the polynomial 
\vspace{-1em}
\begin{equation}
\label{eqn:reconstruction-with-no-cut-&-Q-one}
az^2-bz+c, \text{ where }   
\end{equation}\label{algo:reconstruction-with-no-cut-&-Q-root-set}
\vspace{-2em}
\begin{itemize}[leftmargin=3cm]
\item $a=A[i_1, k_\ell]A[k_\ell, i_1]$, 
\item $b=A[i_1, i_j]A[i_j, k_\ell]A[k_\ell, i_1]+A[i_1, k_\ell]A[k_\ell, i_j]A[i_j, i_1]$, and  
\item $c=A[i_1, i_j]A[i_j, i_1]A[i_j, k_\ell]A[k_\ell, i_j].$
\end{itemize}

\State Let $B_1\in \F^{|L_\ell|\times |L_\ell|}$ be a matrix such that 
\begin{itemize}[leftmargin=3cm]
\item  $B_1[k_\ell, i_j]$ and $B_1[i_j, k_\ell]$ are unknown, and
\item other entries are same as $B[L_\ell]$
\end{itemize} \label{algo:reconstruction-with-no-cut-&-Q-B1-matrix} 
\State Let $D_\ell\in\F^{|L_{\ell-1}|\times |L_{\ell-1}|}$ be diagonal matrix s.t.
\vspace{0em}
\[D_\ell[i_1, i_1]=1 \text{ and }D[s,s]=B[s,i_1]\ \ \forall\, s\in L_{\ell-1}-i_1\] \label{algo:reconstruction-with-no-cut-&-Q-diagonal-matrix}
\vspace{-1em}
\State  Let $B_2\in \F^{|L_\ell|\times |L_\ell|}$ be another matrix such that
\begin{itemize}[leftmargin=3cm]
\item  $B_2[L_{\ell-1}]=D_\ell\cdot B[L_{\ell-1}]^T\cdot D_\ell^{-1}$
\item  $B_2[k_\ell, i_j]$ and $B_2[i_j, k_\ell]$ are unknown, and
\item  rest of the entries are same as $B[L_\ell]$
\end{itemize} \label{algo:reconstruction-with-no-cut-&-Q-B2-matrix}
\For{$s\in\{1,2\}$ and $\gamma\in R_{j, \ell}$}\label{algo:reconstruction-with-no-cut-&-Q-third-loop}
\State $B_s[i_j, k_\ell]\leftarrow \gamma \text{ and } B_s[k_\ell, i_j]\leftarrow \frac{A[i_j, k_\ell]A[k_\ell, i_j]}{\gamma}$ 
\If{$\kPME{A[L_\ell]}{B_s[L_\ell]}{4}$ holds} \label{algo:reconstruction-with-no-cut-&-Q-if-condition}
\State $B[L_\ell]\leftarrow B_s$\label{algo:reconstruction-with-no-cut-&-Q-submatrix-constr}
\State \textbf{exit} the `for'-loop at Step~\ref{algo:reconstruction-with-no-cut-&-Q-third-loop}
\EndIf
\EndFor
\EndFor
\EndFor

\State Output $B$
\end{algorithmic}
\end{algorithm}

\begin{algorithm}
\caption{Finding a sequence of indices satisfying no cut property}
\label{algo:finding-no-cut-sequence}
\textbf{Input:} A set $I$ with $|I|\geq 4$; the family $\submatrixfamily{A}$ for some matrix $A$ whose rows and columns are indexed by $I$, where $A$ satisfies  property $\prop$ and has no cut; two elements $r, s\in I$.\\
\textbf{Output:} A sequence $(i_{|I|}, i_{|I|-1}, \ldots, i_{3})$ such that all $i_j$ are distinct elements of $I\setminus \{r,s\}$ and $j\in \{4,5, \ldots, |I|\}$, the submatrix $A[I_j]$ has no cut, where $I_j=\{r,s, i_3, \ldots, i_j\}$. 
\begin{algorithmic}[1]
\State \Return \Call{Finding-no-cut-sequence}{$I$, $\submatrixfamily{A}$, $r$, $s$}
\bigskip
\Function{Finding-no-cut-sequence}{$I$, $\submatrixfamily{A}$, $r$, $s$}
\If{$|I|=4$}
\State Let $\{i, j\}=I\setminus \{r,s\}$. Then, \Return $(i,j)$
\Else
\State For all $ i\in I\setminus\{r, s\}$, applying~\cref{algo:cut-detection-finding-algorithm} with input $(I-i, \submatrixfamily{A[I-i]})$, find an $e\in I\setminus\{r, s\}$ such that $A[I-e]$ has no cut.\label{algo:finding-no-cut-sequence-index-find}
\State $\mathcal X\leftarrow$ \Call{Finding-no-cut-sequence}{$I-e$, $\submatrixfamily{A[I-e]}$, $r$, $s$}\label{algo:finding-no-cut-sequence-recursive-call}
\State \Return $(e, \mathcal X)$.
\EndIf
\EndFunction
\end{algorithmic}
\end{algorithm}

\subsubsection{Correctness of~\cref{algo:reconstruction-with-no-cut-&-Q}}
\label{sec:reconstruction-with-no-cut-&-Q-correctness}

\cref{algo:reconstruction-with-no-cut-&-Q} invokes~\cref{algo:finding-no-cut-sequence} as a key subroutine. Therefore, we begin by establishing the correctness of~\cref{algo:finding-no-cut-sequence}. 
\paragraph{Correctness of~\cref{algo:finding-no-cut-sequence}.}  The input of~\cref{algo:finding-no-cut-sequence} has the following components:
\begin{enumerate}
\item a finite set $I$ with $|I|\geq 4$
\item the family $\submatrixfamily{A}$ for some matrix $A$ whose rows and columns are indexed by $I$, where $A$ satisfies property $\prop$ and has no cut
\item two elements $r,s\in I$.
\end{enumerate}
Let $|I|=n$.~\cref{algo:finding-no-cut-sequence} calls a recursive procedure \textsc{Finding-no-cut-sequence} with the input $(I, \submatrixfamily{A}, r,s)$, and it outputs a sequence $\mathcal X= (i_{n}, i_{n-1}, i_{n-2}, \ldots, i_{3})$. For all $j\in\{3,4,5, \ldots, n\}$ , let $I_j=\{r, s, i_3, \ldots, i_j\}.$ Then, the output sequence $\mathcal X$ satisfies the following property.
\begin{claim}
\label{clm:correctness-of-finding-no-cut-sequence}
Given $(I, \submatrixfamily{A}, r,s)$ as input to \textsc{Finding-no-cut-sequence}, it returns a sequence $$\mathcal X=(i_{n}, i_{n-1}, i_{n-2}, \ldots, i_3)$$ such that for all $j\in\{4,5,6, \ldots, n-1\}$, $A[I_j]$ has no cut.
\end{claim}
\begin{proof}
We prove it by induction on $|I|$. Since the matrix $A$ has no cut, the base case, $|I|=4$, trivially follows. Next, we focus on the induction step. Let $I$ be  the input set with $|I|\geq 5$. Since the input $A$ has no cut and satisfies property $\prop$, by~\cref{cor:no-cut-from-big-to-small-matrix}, there exists an $e\in I\setminus\{r,s\}$ such that $A[I-e]$ has no cut. Step~\ref{algo:finding-no-cut-sequence} finds such an index $e$. Since $A$ has property $\prop$, applying~\cref{prop:matrix-property-large-to-small}, $A[I-e]$ also satisfies property $\prop$. Let $(i_{n-1}, i_{n-2}, \ldots, i_3)$ be the sequence returned in Step~\ref{algo:finding-no-cut-sequence-recursive-call}. Then, from the induction hypothesis, for all $j\in\{4,5, 6, \ldots, n-1\}$, $A[I_j]$ has no cut. Let $i_{n}=e$. Then, the output sequence $(i_n, i_{n-1}, i_{n-2}, \ldots, i_3)$ satisfies property $\prop$. 
\end{proof}
The correctness of~\cref{algo:finding-no-cut-sequence} immediately follows from~\cref{clm:correctness-of-finding-no-cut-sequence}.

\paragraph{Correctness of~\cref{algo:reconstruction-with-no-cut-&-Q}.}
As input,~\cref{algo:reconstruction-with-no-cut-&-Q} has access to the following things:
\begin{enumerate}
\item A set $I$ of size at least $4$
\item the family $\submatrixfamily{A}$ and the oracle $\PMoracle{A}$ for some matrix $A$ whose rows and columns are indexed by $A$, where $A$ does not have cut and satisfies property $\prop$.
\end{enumerate}
By~\cref{obs:uniqueness-of-DS} and~\cref{obs:DE-to-PME}, without loss of generality, we may assume that $A[i_1, s]=1$ for all $s\in I\setminus\{i_1\}$, for some arbitrary element $i_1\in I$ fixed at the starting of the algorithm. We want to prove that~\cref{algo:reconstruction-with-no-cut-&-Q} produces a matrix $B$ whose rows and columns are indexed by $I$ and $A\PME B$.  Moreover, $B[i_1, s]=1$ for all $s\in I\setminus\{i_1\}$.

Let $(i_n, i_{n-1}, i_{n-2},\ldots, i_3)$  be the sequence produced at Step~\ref{algo:reconstruction-with-no-cut-&-Q-first-no-cut-sequence}. For all $j\in[n]$, let $I_j$ denote the set $\{i_1, i_2, \ldots, i_j\}$. From~\cref{clm:correctness-of-finding-no-cut-sequence}, for all $j\in\{4, 5, 6, \ldots,n\}$, the submatrix $A[I_j]$ has no cut. Observe that for each $j\in\{3,4,5,\ldots, n\}$,  at the starting $j$-th iteration of the `for'-loop at Step~\ref{algo:reconstruction-with-no-cut-&-Q-first-loop}, all the entries of the submatrix $B[I_{j-1}]$ are known. Our goal is to show that the following loop-invariant is satisfied at Step~\ref{algo:reconstruction-with-no-cut-&-Q-first-loop}.  
\begin{equation}
\label{eqn:reconstruction-with-no-cut-&-Q-loop-invariant}
\forall\, j\in\{2,3,4,\ldots,n\},\ B[I_j]\PME A[I_j] \text{ and } B[i_1, i_j]=1.
\end{equation}
This will immediately imply that after the $n$-th iteration, we obtain the desired output matrix $B$, that is, $B\PME A$ and $B[i_1, s]=1$ for all $s\in I\setminus\{i_1\}$. 

We prove the invariant in~\cref{eqn:reconstruction-with-no-cut-&-Q-loop-invariant} by induction on $j$. The base case $j=2$ follows from~\cref{prop:reconstruction-size-two-matrix}. Moreover, by applying~\cref{lem:computation-using-PM-oracle}, we can compute $B[i_2, i_1]$ using the oracle $\PMoracle{A}$.

We now show the induction step. From the induction hypothesis, for every $j>2$, we obtain that $B[I_{j-1}]\PME A[I_{j-1}]$ and $B[i_1, s]=1$ for all $s\in I_{j-1}\setminus \{i_1\}$. Note that initially, all the entries of $B[I_j]$ except those in row and column indexed by $i_j$ are known. From~\cref{lem:computation-using-PM-oracle}, $A[i_j, i_j]$ can be computed using the oracle $\PMoracle{A}$. We make $B[i_j,i_j]=A[i_j, i_j]$. Subsequently, we compute the remaining unknown entries of $B[I_j]$ in iterative manner.

Since $A$ satisfies property $\prop$, from~\cref{prop:matrix-property-large-to-small}, $A[I_j]$ also satisfies property $\prop$. Additionally, $A[I_j]$ has no cut. Therefore, from~\cref{clm:correctness-of-finding-no-cut-sequence}, the sequence $(k_j, k_{j-1}, \ldots, k_3)$, satisfies the following property: For all $\ell\in\{3,4, 5, \ldots, j\}$, $A[L_\ell]$ has no cut where $L_\ell=\{i_1, i_j, k_3, k_4, \ldots, k_\ell\}$. Let $L_2=\{i_1, i_j\}$.  Next, we prove the following loop-invariant satisfied by the `for'-loop at Step~\ref{algo:reconstruction-with-no-cut-&-Q-second-loop}.
\begin{claim}
\label{clm:reconstruction-with-no-cut-&-Q-loop-invariant}
For all $\ell\in\{2,3,4, \ldots, j\}$, $A[L_\ell]\PME B[L_\ell]$
\end{claim}

We defer the proof of the above claim, and complete the induction step, assuming the claim holds. Specifically, we aim to show that $A[I_j]$ has no cut and $B[i_1, i_j]=1$,  Observe that $L_j=I_j$. Therefore, by~\cref{clm:reconstruction-with-no-cut-&-Q-loop-invariant}, we immediately obtain $A[I_j]\PME B[I_j]$. Furthermore, we will later show that Step~\ref{algo:reconstruction-with-no-cut-&-Q-submatrix-constr} must be executed for some $s\in\{1,2\}$ and $\gamma\in R_{j,\ell}$. Consequently, by the definitions of $B_1$ and $B_2$, we have $B[i_1, i_j]=1$. 

\begin{proof}[Proof of~\cref{clm:reconstruction-with-no-cut-&-Q-loop-invariant}]
We prove by induction on $\ell$. The base case $\ell=2$ follows from~\cref{prop:reconstruction-size-two-matrix}.  Assume that for some $\ell>2$, $A[L_{\ell-1}]\PME B[L_{\ell-1}]$. Next, we want to show that $A[L_{\ell}]\PME B[L_\ell]$.

Observe that all entries of $B[L_\ell]$, except $B[i_j, k_\ell]$ and $B[k_\ell, i_j]$, are known.~\cref{lem:computation-using-PM-oracle} ensures that we can compute ~\cref{eqn:reconstruction-with-no-cut-&-Q-one} using constantly many queries to the oracle $\PMoracle{A}$. Additionally, using the root finding algorithm available to us, the set $R_{j, \ell}$ in Step~\ref{algo:reconstruction-with-no-cut-&-Q-root-set} can be computed.

Since $A[I_{j-1}]$ has no cut and $B[I_{j-1}]\PME A[I_{j-1}]$, from~\cref{lem:no-cut-to-PME}, $$B[I_{j-1}]\DS A[I_{j-1}], \text{ or } B[I_{j-1}]\DS A^T[I_{j-1}].$$ As $A\PME A^T$, the oracle $\PMoracle{A}$ is identical to $\PMoracle{A^T}$. Thus, without loss of generality, we may assume that $B[I_{j-1}]\DS A[I_{j-1}]$. Then, by~\cref{obs:uniqueness-of-DS}, it follows that $B[I_{j-1}]=A[I_{j-1}]$. Note that $L_\ell-i_j$ is a subset of $I_{j-1}$. Therefore, we obtain 
\begin{equation}
\label{eqn:reconstruction-with-no-cut-&-Q-2nd-loop-invariant-one}
B[L_\ell-i_j]=A[L_\ell-i_j].
\end{equation}

Again, since $A[L_{\ell-1}]$ has no cut and $A[L_{\ell-1}]\PME B[L_{\ell-1}]$, by~\cref{lem:no-cut-to-PME}, $$B[L_{\ell-1}]\DS A[L_{\ell-1}], \text{ or } B[L_{\ell-1}]\DS A^T[L_{\ell-1}].$$ Thus, by~\cref{obs:uniqueness-of-DS}, we obtain
\begin{equation}
\label{eqn:reconstruction-with-no-cut-&-Q-2nd-loop-invariant-two}
B[L_{\ell-1}]=A[L_{\ell-1}], \text{ or } B[L_{\ell-1}]=D_\ell\cdot A^T[L_{\ell-1}]\cdot D_{\ell}^{-1},
\end{equation}
where $D_\ell$ is defined at Step~\ref{algo:reconstruction-with-no-cut-&-Q-diagonal-matrix}. Next, we divide our proof into the following two cases.

\paragraph*{Case I ($B[L_{\ell-1}]=A[L_{\ell-1}]$).} From~\cref{eqn:reconstruction-with-no-cut-&-Q-2nd-loop-invariant-one} and~\cref{eqn:reconstruction-with-no-cut-&-Q-2nd-loop-invariant-two}, we know that all the entries, except $B[i_j, k_\ell]$ and $B[k_\ell, i_j]$,  of $B[L_\ell]$ are known; more specifically, they equal to the corresponding entry of $A[L_\ell]$. Therefore, from the definition in Step~\ref{algo:reconstruction-with-no-cut-&-Q-B1-matrix}, $$B_1[L_{\ell-1}]=A[L_{\ell-1}] \text{ and } B_1[L_{\ell}-i_j]=A[L_{\ell}-i_j].$$ By~\cref{lem:reconstruction-size-three-matrix}, $A[i_j, k_j]$ is a root of~\cref{eqn:reconstruction-with-no-cut-&-Q-one}. Hence, for $s=1$ and $\gamma=A[i_j, k_j]$ at Step~\ref{algo:reconstruction-with-no-cut-&-Q-third-loop}, we obtain a $B[L_\ell]$ such that $\kPME{A[L_\ell]}{B[L_\ell]}{4}$. Since $A[L_\ell]$ satisfies property $\prop$ by \cref{prop:matrix-property-large-to-small}, applying~\cref{thm:PME-for-no-cut} yields $A[L_\ell]\PME B[L_\ell]$.

\paragraph*{Case II ($B[L_{\ell-1}]=D_\ell\cdot A^T[L_{\ell-1}]\cdot D_{\ell}^{-1}$).} This implies that $$A[L_{\ell-1}]= D_\ell\cdot B^T[L_{\ell-1}]D_{\ell}^{-1}.$$ Therefore, from the definition in Step~\ref{algo:reconstruction-with-no-cut-&-Q-B2-matrix}, $B_2[L_{\ell-1}]=A[L_{\ell-1}]$. Combining this with~\cref{eqn:reconstruction-with-no-cut-&-Q-2nd-loop-invariant-one}, we observe that all the entries, except $B_2[i_j, k_\ell]$ and $B_2[k_\ell, i_j]$, are equal to the corresponding entries of $A[L_\ell]$. The remaining proof is similar to \textbf{Case I}. 
\end{proof}

\subsubsection{Time complexity of~\cref{algo:reconstruction-with-no-cut-&-Q}} In Steps~\ref{algo:reconstruction-with-no-cut-&-Q-first-no-cut-sequence} and~\ref{algo:reconstruction-with-no-cut-&-Q-second-no-cut-sequence},  \cref{algo:reconstruction-with-no-cut-&-Q} invokes~\cref{algo:finding-no-cut-sequence}. We therefore begin by analyzing the time complexity of~\cref{algo:finding-no-cut-sequence}. 

\cref{algo:finding-no-cut-sequence} invokes the recursive function \textsc{Finding-no-cut-sequence}. The function makes at most one recursive call to itself, where the size of the input set $I$ decreases by one. Moreover, in Step~\ref{algo:finding-no-cut-sequence-index-find} of each recursive call, it invokes~\cref{algo:cut-detection-finding-algorithm} at most $|I|$ times. By~\cref{lem:CorrectnessOfCutFinding}, the running time of each invocation is polynomial. Hence, the overall running time of~\cref{algo:finding-no-cut-sequence} is $\poly(|I|)$.

We now analyze the time complexity of~\cref{algo:reconstruction-with-no-cut-&-Q}. The algorithm is iterative in nature. Step~\ref{algo:reconstruction-with-no-cut-&-Q-first-no-cut-sequence} invokes~\cref{algo:finding-no-cut-sequence}, which runs in polynomial time.  In Step~\ref{algo:reconstruction-with-no-cut-&-Q-first-loop}, it executes a `for'-loop at most $n$ times, where $n$ is the cardinality of the input set $I$. For the $j$th iteration of this loop,~\cref{algo:reconstruction-with-no-cut-&-Q}:
\begin{enumerate}
\item In Step~\ref{algo:reconstruction-with-no-cut-&-Q-second-no-cut-sequence}, invokes~\cref{algo:finding-no-cut-sequence} with an input set of size at most $j$, for $j\geq 4$.
\item In Step~\ref{algo:reconstruction-with-no-cut-&-Q-second-loop}, executes another `for'-loop  at most $j$ times. In each iteration of this loop, it identifies a correct pair $(s, \gamma)$ from the set $\{1,2\}\times R_{j,\ell}$ in Step~\ref{algo:reconstruction-with-no-cut-&-Q-third-loop}. Given access to the oracle $\PMoracle{A}$, we can verify whether $\kPME{A[L_\ell]}{B[L_\ell]}{4}$ in polynomial time in Step~\ref{algo:reconstruction-with-no-cut-&-Q-if-condition}. Therefore, the correct pair $(s, \gamma)$ can be found in polynomial time.
\end{enumerate}
All the remaining steps of the algorithm can also be performed in polynomial time. Thus, the overall running time of~\cref{algo:reconstruction-with-no-cut-&-Q} is $\poly(n)$.

\subsection{PMAP for matrices with property $\prop$: Reduction to the no-cut case}
\label{subsec:reconstruction-with-Q}

We now present an algorithm for computing a principal minor equivalent matrix to a given  matrix $A\in\F^{n\times }$ satisfying property $\prop$. As input, we have access to the oracle $\PMoracle{A}$ (see~\cref{def:prinicpal-minor-oracle}). It is a recursive algorithm that reduces our learning problem of a matrix with property $\prop$ to the learning problem of a matrix with property $\prop$ and has \emph{no} cut. For a detailed description, see~\cref{algo:reconstruction-with-Q}. 



\begin{algorithm}[H]
\caption{Reconstruction of matrices with  property $\prop$}
\label{algo:reconstruction-with-Q}
\textbf{Input:} An access to oracle $\PMoracle{A}$ for some matrix $A\in\F^{n\times n}$ with $\prop$ property.\\
\textbf{Output:} A matrix $B$ such that $B\PME A.$\\
\textbf{Assumption:} A root finding algorithm for quadratic univariate polynomials over $\F$ is available
\begin{algorithmic}[1]
\item For $n=2,3$, reconstruct $B$ using~\cref{prop:reconstruction-size-two-matrix} and~\cref{lem:reconstruction-size-three-matrix}, respectively. Otherwise, do the following.

\State For every $I\subseteq [n]$ with $|I|=4$, applying~\cref{algo:reconstruction-of-size-four-matrices}, compute $\fourPMEfamily{A_I}$. Thus, we compute $\submatrixfamily{A}$\label{algo:reconstruction-with-Q-family-construction}

\State Output \Call{Reconstruction}{$[n]$}\label{algo:reconstruction-with-Q-function-call}

\bigskip
\Function{Reconstruction}{$I$}
\State For $|I|\geq 4$, applying~\cref{algo:cut-detection-finding-algorithm} with input $(I, \submatrixfamily{A[I]})$, decide whether $A[I]$ has a cut. If cut exists, let $S\subseteq I$ be the output of~\cref{algo:cut-detection-finding-algorithm}.\label{algo:reconstruction-with-Q-cut-finding}
\If{$A$ has no cut OR $|I|=3$}
\State If $|I|=3$, using~\cref{lem:reconstruction-size-three-matrix}, reconstruct $B\in \F^{|I|\times |I|}$ with $A[I]\PME B$
\State Otherwise, using~\cref{algo:reconstruction-with-no-cut-&-Q} with input $(I, \submatrixfamily{A[I]}, \PMoracle{A})$, reconstruct $B\in\F^{|I|\times |I|}$ with $A\PME B$\label{algo:reconstruction-with-Q-no-cut-reconstr-step}
\State \Return $B$
\Else
\State Let $s\in S$ and $t\in \comp{S}=I\setminus S$.
\State $B_1\leftarrow$\Call{Reconstruct}{$S+t$} \label{algo:reconstruction-with-Q-B1-matrix}
\State $B_2\leftarrow$\Call{Reconstruct}{$\comp S+t$}\label{algo:reconstruction-with-Q-B2-matrix}
\State Using~\cref{cl:reductionStep}, combine $B_1, B_2$, and obtain a matrix $B\in\F^{|I|\times |I|}$ with $B\PME A$.\label{algo:reconstruction-with-Q-combine-step}
\State \Return $B$
\EndIf
\EndFunction
\end{algorithmic}
\end{algorithm}
\begin{remark} We make the following two remarks about~\cref{algo:reconstruction-with-no-cut-&-Q} and~\cref{algo:reconstruction-with-Q}.
\begin{enumerate}
\item In both~\cref{algo:reconstruction-with-no-cut-&-Q} and~\cref{algo:reconstruction-with-Q}, every query to~$\PMoracle{A}$ involves an input subset $S$ of size at most $4$. In other words, it is sufficient to query only principal minors of  order up to $4$.

\item According to~\cref{rem:MinimalCutFinding} of~\cref{rem:cut-finding}, we can assume that the cut $S$ obtained in Step~\ref{algo:reconstruction-with-Q-cut-finding} of~\cref{algo:reconstruction-with-Q} is a minimal cut for some matrix $C$ such that $C\PME A$. Then, by~\cite[Lemma~3.3]{Chatterjee25} and~\cref{lem:common-cut}, $A[S+t]$ has no cut whenever $|S|>2$. Consequently, we can skip the recursive call at Step~\ref{algo:reconstruction-with-Q-B1-matrix} and directly invoke~\cref{algo:reconstruction-with-no-cut-&-Q} to obtain the matrix $B_1$. Thus, instead of two recursive calls in~\cref{algo:reconstruction-with-Q}, a single call suffices.
\end{enumerate}
\end{remark}

\subsubsection{Correctness of~\cref{algo:reconstruction-with-Q}}
The algorithm has an access to the oracle $\PMoracle{A}$ (see~\cref{def:prinicpal-minor-oracle}) for some matrix $A\in\F^{n\times n}$ satisfying property $\prop$. The algorithm outputs an $n\times n$ matrix $B$. We want to show that $B\PME A$. For $n=2$ and $n=3$, the correctness of the algorithm follows from~\cref{prop:reconstruction-size-two-matrix} and~\cref{lem:reconstruction-size-three-matrix}, respectively. By repeatedly applying~\cref{algo:reconstruction-of-size-four-matrices}, we compute the family $\submatrixfamily{A}$. In Step~\ref{algo:reconstruction-with-Q-function-call}, the algorithm invokes the recursive procedure \textsc{Reconstruction} on the input $[n]$. More generally, for any $I\subseteq [n]$ with $|I|\geq 3$, \textsc{Reconstruction}($I$) satisfies the following.

\begin{claim}
\label{clm:reconstruction-with-Q-recursion}
Let $n\geq 3$ be a positive integer, and $I\subseteq [n]$ of size at least $3$. Then, given $I$ as input, the function \textsc{Reconstruction} returns a matrix $B\in \F^{|I|\times |I|}$ such that $B\PME A[I]$.
\end{claim}
\begin{proof}
We prove this by induction. 

\paragraph*{Base case ($|I|=3$ or $A[I]$ has no cut).} For $|I|=3$,  by~\cref{lem:reconstruction-size-three-matrix}, the matrix $B\PME A[I]$. If $A$ has no cut, then by~\cref{algo:reconstruction-with-no-cut-&-Q}, the algorithm ensures that $B\PME A$. 
\paragraph*{Induction step ($|I|\geq4$ and $A[I]$ has a cut).} Step~\ref{algo:reconstruction-with-Q-cut-finding} of~\cref{algo:reconstruction-with-Q} finds a cut $S\subseteq I$ of $A[I]$. Let $s\in S$ and $t\in \comp{S}=I\setminus S$. Since $2\leq |S|\leq |I|-2$, both $|S+t|$ and $|\comp{S}+s|$ are at least $3$. Then, from the induction hypothesis, the matrices $B_1\in \F^{|S+t|\times |S+t|}$ and $B_2\in\F^{|\comp{S}+s|\times |\comp{S}+s|}$ constructed in Step~\ref{algo:reconstruction-with-Q-B1-matrix} and Step~\ref{algo:reconstruction-with-Q-B2-matrix}, respectively, satisfy $B_1\PME A[S+t]$ and $B_2\PME A[S+t]$. Then, by~\cref{cl:reductionStep}, the algorithm combines $B_1$ and $B_2$, and returns the matrix such that $B\PME A[I]$.  
\end{proof}
The above claim immediately implies that the output matrix $B$ of~\cref{algo:reconstruction-with-Q} satisfies $B\PME A$.

\subsubsection{Time complexity of~\cref{algo:reconstruction-with-Q}} By~\cref{lem:reconstruction-of-size-four-matrices}, in Step~\ref{algo:reconstruction-with-Q-family-construction}, the family $\submatrixfamily{A}$ can be computed in $\poly(n)$ time. Thereafter, ~\cref{algo:reconstruction-with-Q} invokes the recursive function \textsc{Reconstruction}. 

In a particular invocation of \textsc{Reconstruction} with $|I|\geq 4$, by~\cref{lem:CorrectnessOfCutFinding},~\cref{algo:cut-detection-finding-algorithm} decides whether $A[I]$ has a cut in time $\poly(|I|)$. Moreover, if $A[I]$ has a cut, it also returns a cut $S\subseteq I$ within the same time complexity. If $A[I]$ has no cut, in Step~\ref{algo:reconstruction-with-Q-no-cut-reconstr-step},~\cref{algo:reconstruction-with-no-cut-&-Q} computes the matrix $B$ in $\poly(|I|)$ time. 

In Steps~\ref{algo:reconstruction-with-Q-B1-matrix} and~\ref{algo:reconstruction-with-Q-B2-matrix}, \textsc{Reconstruction} recursively calls itself on inputs $S+t$ and $\comp{S}+s$ where $\comp{S}=I\setminus S$, $s\in S$ and $t\in \comp{S}$. Finally, by~\cref{cl:reductionStep}, Step~\ref{algo:reconstruction-with-Q-combine-step} combines $B_1$ and $B_2$ in time $\poly(|I|)$. Thus, for an input $I$ with $|I|=n$, the time complexity of \textsc{Reconstruction} satisfies $$T(n)=T(n_1)+T(n_2)+\poly(n), \text{ where } n_1+n_2=n+2.$$
Solving this recursion, we obtain that~\cref{algo:reconstruction-with-Q} runs in time polynomial in $n$.

\section{An NC Algorithm for PME: Proof of Theorem \ref{thm:main-NC-algorithm}}\label{sec:NC-algo-for-PME}
In this section, we present an NC algorithm to test principal minor equivalence of two matrices and prove \cref{thm:main-NC-algorithm}. The high level idea is as follows:  given two input matrices $A$ and $B$ of size $n \times n$, we design an NC reduction that transforms the problem of testing whether $A \PME B$ into testing principal minor equivalence for polynomially many pairs $(A', B')$ of matrices with $ A'$ satisfies property $\prop$. 

From~\cite{Csanky76, Borodin83}, we know that the determinant of a matrix can be computed in NC. Combined with \cref{thm:main-characterization}, this implies that principal minor equivalence of matrices $A'$ and $B'$, where $A'$ satisfies $\prop$, can also be tested in NC. Thus, our reduction yields an NC algorithm for testing principal minor equivalence of general matrices $A$ and $B$. 

The NC reduction proceeds via the following steps. 
\begin{description}
\item [\textbf{Step I:}] We perform an NC reduction that reduces the problem of testing principal minor equivalence of two arbitrary matrices to testing principal minor equivalence for at most $n$ instances where the input pairs are irreducible matrices. By \cref{lem:redToIr}, matrices $A$ and $B$ are principal minor equivalent if and only if their maximal irreducible submatrices define the same partition of the index set, and the corresponding principal submatrices are themselves principal minor equivalent.

In a directed graph, we can find strongly connected components in NC~\cite{Gazit88, Cole89}. Hence, using \cref{ob:irreducible-matrix-and-directed-graph}, we can find the maximal irreducible submatrices of the input matrices in NC. If the index sets of the irreducible components of two matrices do not match, we output that $A\nPME B$. Otherwise, we test principal minor equivalence for each corresponding pair of maximal irreducible submatrices, in parallel.

\item [\textbf{Step II:}] In this step, we transform, in NC, an input pair $(A, B)$ of \emph{irreducible} matrices into another pair $(A', B')$ such that all off-diagonal entries of $A'$ are \emph{nonzero}, and $$A\PME B\Longleftrightarrow A'\PME B'.$$ For details, see~\cref{subsec:NC-algo-performing-step-II}.

\item [\textbf{Step III:}] Given a pair $(A, B)$ with all off-diagonal entries of $A$ are \emph{nonzero}, we further transform it, in NC, to another pair $(A', B')$ where $A'$ satisfies property $\prop$, such that $$A\PME B\Longleftrightarrow A'\PME B'.$$ For details, see~\cref{subsec:NC-algo-performing-step-III} 

Currently, we do not know of an NC procedure to construct such a pair $(A', B')$ directly from irreducible matrices $(A, B)$. If such a procedure were available, we could bypass \textbf{Step II} entirely. 
\end{description}
Next, we provide a detailed description of \textbf{Step II} and \textbf{Step III}.
\subsection{Ensuring nonzero off-diagonal entries} 
\label{subsec:NC-algo-performing-step-II}
The following lemma ensures that the \textbf{Step II} of the above-mentioned procedure can be done.
\begin{lemma}
\label{lem:ShiftAdjugatePreservesPME}
Let $M$ and $N$ be two $n\times n$ irreducible matrices over $\mathbb{F}$ and $Y=\diag(y_1,y_2,\dots,y_n)$. Then, 
\begin{enumerate}
\item $M\PME N$ if and only if $\adj{M+Y}\PME \adj{N+Y}$.
\item Moreover, if $|\F|\geq 10n^5$, we can find, in NC, a diagonal matrix $D\in \mathbb{F}^{n\times n}$ (equivalently, a substitution of $Y$) such that $M+D$ and $N+D$ are invertible and $\adj{M+D}$ and $\adj{N+D}$ have nonzero off-diagonal entries.
\end{enumerate}
\end{lemma}
 \begin{remark}
 When $|\F|\leq 10n^5$, we can construct, in NC, an extension $\mathbb K$ of $\F$ such that $|\mathbb K|>10n^5$ and work over that field. 
 \end{remark}

\begin{proof}[Proof sketch]
The first part of the lemma follows from \cite[Lemma~4]{Hartfiell84}. The second part follows from the observation that the all the steps of construction procedure described in~\cite[Claim~4.1]{Chatterjee25} can be implemented in NC. There are NC algorithms available for finding the shortest path between two vertices in a directed graph~\cite{Han97}, determinant computation~\cite{Csanky76, Borodin83}, polynomial interpolation and evaluation. Thus, the construction in~\cite[Claim~4.1]{Chatterjee25} can be performed in NC. 
\end{proof}

The above lemma combined with~\cite[Lemma~2.9]{Chatterjee25} ensures that given a pair $(A, B)$ of $n\times n$ irreducible matrices, we can compute, in NC, a diagonal matrix $D\in\F^{n\times n}$ such that $$A\PME B \Longleftrightarrow \adj{A+D}\PME \adj{B+D}.$$ After \textbf{Step II}, $\adj{A+D}$ and $\adj{B+D}$ will be our output matrices. Since the determinant of a matrix can be computed in NC~\cite{Csanky76, Borodin83}, given $D$, we can compute $\adj{A+D}$ and $\adj{B+D}$ in NC.

\subsection{Ensuring property $\prop$} 
\label{subsec:NC-algo-performing-step-III}
In \textbf{Step III}, we assume that $A$ and $B$ are two input matrices with all the off-diagonal entries of $A$ are nonzero. We now discuss an NC algorithm to compute a diagonal matrix $D$ such that $\adj{A+D}$ and $\adj{B+D}$ both satisfy property $\prop$.

\begin{lemma}
\label{lem:findDwithQ}
Let $n\geq 5$ and $\mathbb{F}$ be a field of sufficiently large size. Let $A,B\in \mathbb{F}^{n \times n}$ be invertible matrices with non-zero off-diagonal entries and $Y=\diag(y_1,y_2,\dots,y_n)$. Then, in NC, we can find a diagonal matrix $D\in \mathbb{F}^{n\times n}$  with the following properties.
\begin{enumerate}
\item $A+D$ and $B+D$ are invertible.
\item $\adj{A+D}$ and $\adj{B+D}$ have nonzero off-diagonal entries.
\item For any four distinct elements $i,j,k,\ell\in [n]$, $\adj{A+D}$ and $\adj{B+D}$ both satisfy \cref{eq:diag-shift-property} in \cref{clm:satisfying-prop-Q}.
\end{enumerate}
\end{lemma}
\begin{remark}
More specifically, we require the size of the underlying to be greater than $10n^5$. If $|\F|\leq 10n^5$, we construct, in NC, a field extension $\mathbb K$ of $\F$ of size greater than $10n^5$, and perform our computations over $\mathbb K$.
\end{remark}
\begin{proof}
From~\cref{lem:ShiftAdjugatePreservesPME}, we can compute, in NC, a diagonal matrix $\widetilde D\in \F^{n\times n}$ such that the first two properties of the lemma are satisfied. That is, both $A+\widetilde D$ and $B+\widetilde D$ are invertible, and $\adj{A+\widetilde D}$ has nonzero off-diagonal entries. Next, we focus on satisfying the third property. Let $$\mathcal{T}_A  =\left\{(S,T)\mid S,T\subset [n], |S|=|T|=2,\ S\cap T=\emptyset,\  \rank \adj{A+Y}[S,T] =2\right\}.$$
We consider the following claim.
\begin{claim}
\label{clm:3typeMatrix}
For any $(S, T)\in \mathcal T_A$, we can compute, in NC, a diagonal matrix $A_{S,T}\in \mathbb{F}^{n\times n} $  such that $$\rank \adj{A+A_{S,T}}[S,T]=2 \Longleftrightarrow \rank \adj{A+Y}[S,T]=2.$$
\end{claim}
We defer the proof of the claim to the end of the section. Once we have constructed the diagonal matrices $\widetilde D$ and $A_{S,T}$ for $(S, T)\in \mathcal T_A$,  we use a standard technique based on Lagrange interpolation to compute, in NC, a single matrix $D$ which satisfies all three properties of the lemma. For details of this step, see the proof of~\cite[Claim~4.1]{Chatterjee25}\,\footnote{Although the method in \cite[Claim~4.1]{Chatterjee25} is presented in the context of a polynomial-time algorithm, it can also be implemented in NC.}. 
\end{proof}
We now describe the proof of~\cref{clm:3typeMatrix}.
\begin{proof}[Proof of~\cref{clm:3typeMatrix}]
Note that $A+Y$ is nonsingular. Therefore, applying Jacobi's identity~\cite[Page~21]{Gantmacher59}, $$\det\left(\adj{A+Y}[S,T]\right)=0 \Longleftrightarrow \det\left((A+Y)[\comp{T},\comp{S}]\right)=0.$$ Without loss of generality, let $S=\{1,2\}$ and $T=\{3,4\}$. Let $A=(a_{i,j})_{i,j\in[n]}$. Then,
\[
A'=(A+Y)[\comp{T},\comp{S}]= 
\begin{pmatrix} a_{1,3} & a_{1,4} & a_{1,5} & \cdots & a_{1,n}\\
a_{2,3} & a_{2,4} & a_{2,5} & \cdots & a_{2,n}\\
a_{5,3} & a_{5,4} & a_{5,5}+y_5 & \cdots & a_{5,n}\\
    \vdots & & & \ddots &\\
a_{n,3} & a_{n,4} & a_{n,5} & \cdots & a_{n,n}+y_n 
\end{pmatrix}
\]
We now perform some column operations on $A'$. Consider that the rows and the columns of $A'$ are indexed by $[n-2]$. Note that $a_{1, 3}=A'[1, 1]\neq 0$. We subtract $A'[1,i]/A'[1,1]$ times the first column from $i$-th column of $A'$ for each $i\in \{2,3,\dots n-2\}$. This operations zero out all the entries in the first row of $A'$, except for the first entry, and can be  performed in NC. 

Let $\widetilde Y=\diag (0, y_5, y_6,\ldots, y_n)$. After this column operations, let $\tilde A+\widetilde Y$ be the resulting principal submatrix whose rows and columns are indexed by $\{2,3,\dots,n-2\}$. Then, we have $$\det(A')=a_{1,3}\cdot \det(\tilde A+\widetilde Y),$$ and so the determinant of $A'$ is nonzero if and only if $\det(\tilde A+\widetilde Y)$ is nonzero. 

Therefore, it suffices to find a substitution $(a_1, a_2, \ldots, a_{n-4})\in\F^{n-4}$ of $(y_5, y_6, \ldots, y_n)$ such that $$\det\left(\tilde A+\widetilde Y\right)\neq 0\Longrightarrow \det\left(\tilde A+\widetilde D\right)\neq 0$$ where  $\widetilde D=\diag(0,a_1, a_2, \ldots, a_{n-4})$. This problem reduces to the following problem: Given an $M\in\F^{(n-1)\times (n-1)}$, find a diagonal matrix $C\in\F^{(n-1)\times (n-1)}$ such that $$\adj{M+Z}[1,2 ]\neq 0\Longrightarrow \adj{M+C}[1,2]\neq 0,$$ where $Z=\diag(z_1, z_2, \ldots, z_{n-1})$. This problem is solved in NC by~\cref{lem:ShiftAdjugatePreservesPME}. 

Thus, for all $(S, T)\in\mathcal T_A$, we can compute the desired diagonal matrix $A_{S, T}$ in NC.  
\end{proof}
Since both $(A+D)$ and $(B+D)$ are invertible,~\cite[Lemma~2.9]{Chatterjee25} implies that $$A\PME B\Longleftrightarrow \adj{A+D}\PME\adj{B+D}.$$ Hence, from~\cref{thm:main-characterization}, 
$$A\PME B \Longleftrightarrow \kPME{\adj{A+D}}{\adj{B+D}}{4}.$$
We can compute $\adj{A+D},\adj{B+D}\in \mathbb{F}^{n\times n}$~\cite{Csanky76, Borodin83} and check whether $\kPME{\adj{A+D}}{\adj{B+D}}{4}$. Therefore, we can test whether $A\PME B$ in NC. It completes the proof of~\cref{thm:main-NC-algorithm}.

\paragraph{Acknowledgements} We would like to thank anonymous reviewers for their valuable comments.

\newcommand{\etalchar}[1]{$^{#1}$}

\appendix
\section{Missing Proofs from~\cref{sec:PME-upto-4-characterization}}
\label{appendix:missing-proofs-from-prprelim}

\subsection{Proof of~\cref{lem:small-matrices-DE-to-cut-4-base-case}}
\label{subsec:proof-of-lem:small-matrices-DE-to-cut-4-base-case}
\repstatement{lem:small-matrices-DE-to-cut-4-base-case}{
Let $A$ and $B$ be two $4\times 4$ matrices such that all off-diagonal entries of $A$ are nonzero, and $A\PME B$. Let $X\subseteq \{1,2,3,4\}$ with $|X|=2$, $\calD_1(A,B)=X$ and $\calD_2(A,B)=\comp{X}$. Then, $X$ is a cut of $A$.} 
\begin{proof}
Without loss of generality, we can assume that $\calD_1(A, B)=\{1,2\}$ and $\calD_2(A, B)=\{3,4\}$. Since $\calD_1(A,B)=\{1,2\}$, there exists a matrix $A'$ that is diagonally similar to $A$ and $A'[\{j, 3,4\}]=B[\{j,3,4\}]$ for $j=1,2$. Furthermore, without loss of generality, we may work with any matrix that is diagonally similar to $A$. Thus,  we assume that $$A[\{j,3,4\}]=B[\{j,3,4\}] \text{ for } j=1,2.$$ This implies that $A$ and $B$ of the following form:
\begin{equation}
\label{eqn:small-matrices-DE-to-cut-4-base-case-one}
A=
\begin{pmatrix}
a_{11}  &  a_{12}   & a_{13}    & a_{14}\\
a_{21}  &  a_{22}   & a_{23}    & a_{24}\\
a_{31}  &  a_{32}   & a_{33}    & a_{34}\\
a_{41}  &  a_{42}   & a_{43}    & a_{44}
\end{pmatrix}
\ \ \text{ and }\ \ 
B=
\begin{pmatrix}
a_{11}  &  b_{12}   & a_{13}    & a_{14}\\
b_{21}  &  a_{22}   & a_{23}    & a_{24}\\
a_{31}  &  a_{32}   & a_{33}    & a_{34}\\
a_{41}  &  a_{42}   & a_{43}    & a_{44}
\end{pmatrix}
\end{equation}

Let $C$ be a $3\times 3$ diagonal matrix such that for all $i=1,2,3$, $C[i,i]=c_i$ with $c_1=1$. Similarly, let $D$ be a $3\times 3$ diagonal matrix such that for all $i=1,2,3$, $D[i,i]=d_i$ with $d_1=1$. Moreover, assume that 
\begin{equation}
\label{eqn:small-matrices-DE-to-cut-4-base-case-two}
B[\{1,2,3\}]^T=C^{-1}\cdot A[\{1,2,3\}]\cdot C\ \ \text{ and }\ \ B[\{1,2,4\}]^T=D^{-1}\cdot A[\{1,2,4\}]\cdot D
\end{equation}
As a consequence, $c_2=d_2$. Let $e_1=1, e_2=c_2=d_2, e_3=c_3, e_4=d_4$. Thus, from~\cref{eqn:small-matrices-DE-to-cut-4-base-case-one} and~\cref{eqn:small-matrices-DE-to-cut-4-base-case-two}, $A$ and $B$ are of the following form:
\[
A=
\begin{pmatrix}
a_{11}     &  a_{12}                  & a_{13}    & a_{14}\\
&&&\\
a_{21}     &  a_{22}                  & a_{23}    & a_{24}\\
&&&\\
a_{13}e_3  &  \frac{a_{23}e_3}{e_2}   & a_{33}    & a_{34}\\
&&&\\
a_{14}e_4  &  \frac{a_{24}e_4}{e_2}   & a_{43}    & a_{44}
\end{pmatrix}
\ \ \text{ and }\ \ 
B=
\begin{pmatrix}
a_{11}     &  \frac{a_{21}}{e_2}      & a_{13}    & a_{14}\\
&&&\\
a_{12}e_2  &  a_{22}                  & a_{23}    & a_{24}\\
&&&\\
a_{13}e_3  &  \frac{a_{23}e_3}{e_2}   & a_{33}    & a_{34}\\
&&&\\
a_{14}e_4  &  \frac{a_{24}e_4}{e_2}   & a_{43}    & a_{44}
\end{pmatrix}
\]
Since $A\PME B$,
\begin{align*}
0 &= \det(A)-\det(B)\\
  &= (a_{12}a_{23}a_{34}a_{41}-a_{14}a_{43}a_{32}b_{21}) + (a_{14}a_{43}a_{32}a_{21}-b_{12}a_{23}a_{34}a_{41})\\
  &\ \ \ \ +(a_{12}a_{24}a_{43}a_{31}-a_{13}a_{34}a_{42}b_{21})+(a_{13}a_{34}a_{42}a_{21}-b_{12}a_{24}a_{43}a_{31})\\
  &= a_{12}a_{23}a_{14}(a_{34}e_4-a_{43}e_3)-\frac{a_{14}a_{23}a_{21}}{e_2}(a_{34}e_4-e_3a_{43})\\
  &\ \ \ \ -a_{13}a_{24}a_{12}(a_{34}e_4-a_{43}e_3)+\frac{a_{13}a_{24}a_{21}}{e_2}(a_{34}e_4-a_{43}e_3)\\
  &=(a_{34}e_4-a_{43}e_3)\left(a_{12}a_{23}a_{14}-\frac{a_{14}a_{23}a_{21}}{e_2}-a_{13}a_{24}a_{12}+\frac{a_{13}a_{24}a_{21}}{e_2}\right)\\
  &= (a_{34}e_4-a_{43}e_3)(a_{13}a_{24}-a_{23}a_{14})\left(\frac{a_{21}}{e_2}-a_{12}\right)
  \end{align*}
We now divide the proof into following three cases.
\begin{enumerate}
\item \textbf{Case I ($a_{34}e_4-a_{43}e_3=0$):} In this case, it follows that $B^T=E^{-1}AE$ where $E$ is a diagonal matrix such that $E[i,i]=e_i$ for all $i=1,2,3,4$. Therefore, $\calD_2(A,B)=\{1,2,3,4\}$, which contradicts the assumptions of the lemma.

\item \textbf{Case II ($\frac{a_{21}}{e_2}-a_{12}=0$):} This implies that $A=B$, which contradicts the assumption that $\calD_1(A,B)=\{1,2\}$.

\item \textbf{Case III ($a_{13}a_{24}-a_{23}a_{14}=0$):} In this case, it follows that $\{1,2\}$ is a cut of $A$.
\end{enumerate}
This concludes the proof of the lemma.
\end{proof}

\subsection{Proof of~\cref{lem:cut-decomposition}}
\label{subsec:proof-of-lem:cut-decomposition}
\repstatement{lem:cut-decomposition}{
Let $A$ and $B$ be two $n\times n$ matrices over a field $\F$ with nonzero off-diagonal entries. Let $S\subseteq [n]$ be a cut of both $A$ and $B$. Consider any two indices $s,t$ such that $s\in S$ and $t\in \comp{S}$. Then $A \PME B$ if and only if $A[S+t] \PME B[S+t]$ and $A[\comp{S} + s] \PME B[\comp{S} + s]$. Furthermore, if $S$ is a minimal cut of $A$, then $A[S+t]$ has no cut.
}

\begin{proof}
The “only if” direction is straightforward. We now prove the converse. The spirit of our argument follows that of~\cite[Lemma~2.12]{Chatterjee25}. Without loss of generality, assume that $S=[|S|]$, $s=|S|-1$ and $t=|S|+1$. Consider the following definitions.
\begin{align*}
p=A[t, S], \ \ &\tilde p=B[t, S],\ \ q^T=\frac{A[s, \comp{S}]}{A[s, t]},\ \ \tilde q^T=\frac{B[s, \comp{S}]}{B[s, t]},\ \text{ and }\\
u=A[s, \comp{S}], \ \ & \tilde u=B[s, \comp{S}],\ \ v^T=\frac{A[t, S]}{A[t, s]},\ \text{ and }\ \tilde v^T=\frac{B[t, S]}{B[t, s]}.
\end{align*}
Therefore, $A$ and $B$ can be written as follows: 
\[
A=
\begin{blockarray}{ccc}
         & S          & \comp{S}\\
\begin{block}{c(cc)}
S        & M          & p\cdot q^T\\
&&\\
\comp{S} & u\cdot v^T & N\\
\end{block}
\end{blockarray} 
\:\:\text{ and }\:\:
B=
\begin{blockarray}{ccc}
         & S          & \comp{S}\\
\begin{block}{c(cc)}
S        & \widetilde M          & \tilde p\cdot \tilde q^T\\
&&\\
\comp{S} & \tilde u\cdot \tilde v^T & \widetilde N\\
\end{block}
\end{blockarray}
\]

Let $X\subseteq [n]$. Observe that if $X$ is a subset of $S+t$ or $\comp{S}+s$, then $\det(A[X])= \det(B[X])$. Now consider that  $X=X_1\sqcup X_2$ such that $X_1$ and $X_2$ are nonempty subsets of  $S$ and $\comp{S}$, respectively.  Next, we prove that $\det(A[X])=\det(B[X])$.

Assume that the coordinates $p, \tilde p, v, \tilde v$ are indexed by $S$ and the coordinates of $q, \tilde q, u, \tilde u$ are indexed by $\comp{S}$. By $p_{X_1}, q_{X_2}, u_{X_2}$ and $v_{X_1}$, we denote the projection of the respective vectors on the respective coordinates. Similarly, we can define the vectors $\tilde p_{X_1}, \tilde q_{X_2}, \tilde u_{X_2}$ and $\tilde v_{X_1}$. Let $A'$ and $B'$ denote the submatrices $A[X]$ and $B[X]$, respectively. Then, 
\[
A'=
\begin{pmatrix}
M[X_1]                 & p_{X_1}\cdot q_{X_2}^T\\
&\\
u_{X_2}\cdot v_{X_1}^T & N[X_2]
\end{pmatrix}
\:\:\text{ and }\:\: 
B'=
\begin{pmatrix}
\tilde M[X_1]                 & \tilde p_{X_1}\cdot \tilde u_{X_2}^T\\
&\\
\tilde q_{X_2}\cdot \tilde v_{X_1}^T & \tilde N[X_2]^T
\end{pmatrix}
\]


Let $\ell=|X|$, $k=|X_1|$, $K=[k]$ and $\comp{K}=\{k+1, k+2, \ldots, \ell\}$. Suppose that the rows and columns of $A'$ and $B'$ are indexed by $[\ell]$, and the rows and columns of $M[X_1]$ and $\widetilde M[X_1]$ are indexed by $K$.  For each $i\in K$, let $M_i$ denote the $k\times k$ matrix obtained by removing $i$-th column of $M[X_1]$ and appending $p_{X_1}$ as the $k$-th column. For $j\in \comp{K}$, let $N_j$ denote the $(l-k)\times (l-k)$ matrix obtained by removing $j$-th column of $N[X_2]$ and adding $u_{X_2}$ as the first column. Similarly, for all $i\in K$ and $j\in \comp{K}$, we can define the matrices $\widetilde M_i$ and $\widetilde N_j$ from $\widetilde M[X_1]$ and $\widetilde N[X_2]$, respectively. 

Using the Generalized Laplace Theorem (see~\cite[Theorem~3.1]{Ahmadieh23}),  $\det(A')$ can be written as follows.
$$\det(A') = \sum_{T\subseteq [\ell],\, |T|=k} (-1)^{\sum K+\sum T}\det(A'[K, T])\det(A'[\comp{K},\comp{T}]).$$
Note that for all $T\subset [\ell]$ with $|T\cap \comp{K}|\geq 2$, the submatrix $A'[K, T]$ is not full rank since $\rank(A'[K,\comp{K}])$ $\leq 1$. Therefore, for all such $T\subseteq [\ell]$ with $|T|=k$,  $\det(A'[K, T])=0$. This implies that 
$$\det(A') = \det(M[X_1])\det(N[X_2]) + \sum_{i\in K,j\in \comp{K}}(-1)^{j-i}\det(A'[K, K-i+j])\det(A'[\comp{K},\comp{K}+i-j]).$$
Observe that for all $i\in K$ and $j\in \comp{K}$, $$\det(A'[K, K-i+j])=q_{X_2}[j-k]\det(M_i)\ \text{ and }\ \det(A'[\comp{K}, \comp{K}+i-j])=v_{X_1}[i]\det(N_j).$$ Therefore, from the above two equations,

\begin{align*}
\det(A') &= \det(M[X_1])\det(N[X_2]) + \sum_{i\in K,j\in \comp{K}}(-1)^{j-i}v_{X_1}            [i]q_{X_2}[j-k]\det(M_i)\det(N_j)\\
        &= \det(M[X_1])\det(N[X_2])+\left(\sum_{i\in K}(-1)^{i} v_{X_1}[i]\det(M_i)\right)\left(\sum_{j\in \comp{K}}(-1)^{j} q_{X_2}[j-k]\det(N_j)\right)
\end{align*}
Similarly, using the Generalized Laplace Theorem (see~\cite[Theorem~3.1]{Ahmadieh23}), we obtain that 
$$\det(B')=\det(\widetilde M[X_1])\det(\widetilde N[X_2])+\left(\sum_{i\in K}(-1)^{i} \tilde v_{X_1}[i]\det(\widetilde M_i)\right)\left(\sum_{j\in \comp{K}}(-1)^{j} \tilde q_{X_2}[j-k]\det(\widetilde N_j)\right)$$ 
Since $A[S+t]\PME B[S+t]$ and $A[\comp{S}+s]\PME B[\comp{S}+s]$, $\det(M[X_1])\cdot \det(N[X_2])$ is equal to $\det(\widetilde M[X_1])\cdot \det(\widetilde N[X_2])$. Therefore, $\det(A')=\det(B')$ if and only if the second terms of the above expressions for $\det(A')$ and $\det(B')$ are equal. 


Let $P$ and $\widetilde P$ be two $(|X_2|+1)\times (|X_2|+1)$ matrices defined as follows:
\[
P=
\begin{pmatrix}
0       & q_{X_2}^T\\
u_{X_2} & N[X_2]
\end{pmatrix}
\:\:\text{ and }\:\:
\widetilde P=
\begin{pmatrix}
0       & \tilde q_{X_2}^T\\
\tilde u_{X_2} & \widetilde N[X_2]
\end{pmatrix}
\]
Then, 
$$(-1)^k\det(P)      =  \sum_{j\in\comp{K}}(-1)^j q_{X_2}[j-k]\det(N_j) \:\:\text{ and }\:\: (-1)^k\det(\widetilde P) =  \sum_{j\in\comp{K}}(-1)^j \tilde q_{X_2}[j-k]\det(\widetilde N_j).$$
Similarly, let $R$ and $\widetilde R$ be two $(|X_1|+1)\times (|X_1|+1)$ matrices defined as follows:
\[
R=
\begin{pmatrix}
M[X_1]      & p_{X_1}\\
v_{X_1}^T & 0
\end{pmatrix}
\:\:\text{ and }\:\:
\widetilde R=
\begin{pmatrix}
\widetilde M[X_1]      & \tilde p_{X_1}\\
\tilde v_{X_1}^T & 0
\end{pmatrix}
\]
Then, 
$$(-1)^k \det(R)            =  \sum_{i\in K}(-1)^i v_{X_1}[i]\det(M_i)\:\: \text{ and }\:\: (-1)^k \det(\widetilde R) =  \sum_{i\in K}(-1)^i \tilde v_{X_1}[i]\det(\widetilde M_i)$$
Thus, considering the above expressions for $\det(A')$ and $\det(B')$, we get that $$\det(A')=\det(B')\Longleftrightarrow \det(P)\det(R)=\det(\widetilde P)\det(\widetilde R).$$  

From the definitions of $q_{X_2}, u_{X_2}, \tilde q_{X_2}, \tilde u_{X_2}$, we can write that 
\[
A[X_2+s]=
\begin{pmatrix}
A[s,s]      & A[s, t]\cdot q_{X_2}^T\\
u_{X_2} & N[X_2]
\end{pmatrix}
\:\:\text{ and }\:\:
B[X_2+s]=
\begin{pmatrix}
B[s,s]      & B[s, t]\cdot \tilde q_{X_2}^T\\
\tilde u_{X_2} & \widetilde N[X_2]
\end{pmatrix}
\]
Since $A[\comp{S}+s]\PME B[\comp{S}+s]$, $\det(A[X_2+s])=\det(B[X_2+s])$. Thus, 
\begin{equation}
\label{eqn:decomp_bottom-part}
A[s, t]\det(P)=B[s, t]\det(\widetilde P)  
\end{equation}

From the definitions of $p_{X_1}, v_{X_1}, \tilde p_{X_1}, \tilde v_{X_1}$, we obtain the following: 
\[
A[X_1+t]=
\begin{pmatrix}
M[X_1]      & p_{X_1}\\
A[t, s]\cdot v_{X_1}^T & A[t,t]
\end{pmatrix}
\:\:\text{ and }\:\:
B[X_1+t]=
\begin{pmatrix}
\widetilde M[X_1]      & \tilde p_{X_1}\\
A[t, s]\cdot \tilde v_{X_1}^T & B[t,t]
\end{pmatrix}
\]
Since $A[S+t]\PME B[S+t]$, $\det(A[X_1+t])=\det(B[X_1+t])$. Thus, 
\begin{equation}
\label{eqn:decomp_top-part}
A[t, s]\det(R)=B[t, s]\det(\widetilde R).
\end{equation}
Since $A[S+t]\PME B[S+t]$, $A[\{s, t\}]\PME B[\{s, t\}]$. This combined with~\cref{eqn:decomp_bottom-part} and~\cref{eqn:decomp_top-part} imply that 
$\det(P)\det (R)=\det(\widetilde P)\det(\widetilde R)$. Thus, $\det(A')=\det(B')$. This completes the proof of the lemma.

The second part of the lemma directly follows from~\cite[Lemma~3.3]{Chatterjee25}.
\end{proof}

\subsection{Proof of~\cref{lem:common-cut}}
\label{subsec:proof-of-lem:common-cut}
\repstatement{lem:common-cut}{
Let $A$ and $B$ be two $n \times n$ matrices with nonzero off-diagonal entries and $\kPME{A}{B}{4}$. Let $S$ be a minimal cut of $A$ of size greater than $2$. Then, $S$ is also a cut of $B$.
}
\begin{proof}
Fix an element $s \in \comp{S}$. We will show that for all $t \in \comp{S+s}$, the set $T_t := \{s, t\}$ is a cut in $B[S \cup T_t]$. This will imply that $S$ is a cut in $B$.

Since $S$ is a minimal cut in $A$ of size greater than two, by~\cite[Lemma~3.3]{Chatterjee25}, both $A[S + {s}]$ and $A[S + {t}]$ have no cuts. By assumption, $\kPME{A}{B}{4}$ holds. Then, applying \cref{thm:PME-for-no-cut}, we get $A \PME B$. Therefore, by~\cref{lem:no-cut-to-PME}, we conclude that $A[S + {s}] \DE B[S + {s}]$ and $A[S + {t}] \DE B[S + {t}]$. Consequently, both $B[S + {s}]$ and $B[S + {t}]$ have no cuts.

Now, suppose for contradiction that $T_t$ is not a cut in $B[S \cup T_t]$. Note that $T_t$ is a cut in $A[S \cup T_t]$ of size two. Then, by~\cref{lem:Cut-and-PME-upto-4}, there exists a cut $X \subseteq S \cup T_t$ in the matrix $B[S \cup T_t]$ such that $s \in X$ and $t \notin X$.

Since $|S \cup T_t| \geq 5$, either $|X| > 2$ or the size of the complement $\widetilde{X} := (S \cup T_t) \setminus X$ is greater than 2.
\begin{itemize}
\item If $|X| > 2$, then $X - {s}$ is a cut in $B[S + {t}]$, since $X$ is a cut in $B[S \cup T_t]$ and $s \in X$.

\item Otherwise, if $|\widetilde{X}| > 2$, then $\widetilde{X} \setminus {t}$ is a cut in $B[S + {s}]$.
\end{itemize}
In both cases, we obtain a contradiction, as we have already shown that $B[S \cup {s}]$ and $B[S + {t}]$ have no cuts.

Thus, $T_t$ is a cut in $B[S \cup T_t]$ for all $t \in \comp{S+s}$. This completes the proof.
\end{proof}

\section{Can property~$\prop$ be dropped?}
\label{appendix:necessity-of-Q}
{\em We here construct a pair of matrices $A,B\in\F^{n\times n}$ such that 
\begin{itemize}
\item $A$ \emph{does not} satisfy property $\prop$, 
\item $\kPME{A}{B}{4}$ is satisfied, but $A$ is \emph{not} principal minor equivalent to $B$.
\end{itemize}

We first  relax that the first condition of property $\prop$, that is, some off-diagonal entries of $A$ can be zero. Suppose that the matrices $A$ and $B$ are defined as follows: $$A[i, i+1]=B[i, i+1]=1 \text{ for all } i\in[n-1],\ A[n,1]=1,\ B[n,1]=2, $$ and the other entries of $A$ and $B$ are zero. Then, it is not difficult to verify that $\kPME{A}{B}{n-1}$ but $A\nPME B$. Thus, if we allow zero off-diagonal entries in $A$ and $B$, then $\kPME{A}{B}{4}$ does not imply $A\PME B$.   

 We now assume that off-diagonal entries of $A$ and $B$ are nonzero, but $A$ does not satisfy property $\prop$, that is, $A$ violates the second condition of property $\prop$. Suppose that $A$ is an $n\times n$ matrix of the following form:
\begin{equation}
\label{eqn:counter-example-one}
A=
\begin{pmatrix}
1                &  1                & 1                & 1          & a_{15}      & a_{16} &\cdots & a_{1n}\\
2                &  1                & 1                & 1          & a_{25}      & a_{26} & \cdots  & a_{2n}\\
1                &  1                & 1                & 2          & a_{35}      & a_{36} & \cdots  & a_{3n}\\
1                &  1                & 1                & 1          & a_{45}      &  a_{46} & \cdots  & a_{4n}\\
a_{15}           &  a_{25}           & a_{35}           & a_{45}     & 1           & a_{56} & \cdots  & a_{5n}\\
a_{16}           &  a_{26}           & a_{36}           & a_{46}     & a_{56}      & 1 & \cdots  & a_{6n}\\
\vdots           &  \vdots           & \vdots           & \vdots     &\vdots    &\vdots   &\ddots  & \vdots\\
a_{1n}           &  a_{2n}           & a_{3n}           & a_{4n}     & a_{5n}      & a_{6n} & \cdots & 1\\
\end{pmatrix}
\end{equation}
Then, we can show the following claim.
\begin{claim}
\label{clm:counter-example-tool}
Let $A$ be an $n\times n$ matrix over a field $\F$ with nonzero off-diagonal entries, and assume that$A$ is of the form given in~\cref{eqn:counter-example-one}. Suppose that $A$ has no cut $X$ such that $\{1,2\}\subseteq X$ and $\{3,4\}\subseteq \comp{X}$. Let $B$ be another $n\times n$ matrix such that 
\begin{equation}
\label{eqn:counter-example-two}
B[i,j]=A[i,j]\ \forall\, (i,j)\notin\{(1,2), (2,1)\},\ \ \   B[1,2]=A[2,1]=2, \text{ and } B[2,1]=A[1,2]=1
\end{equation}
Then, $\kPME{A}{B}{4}$, but $A\nPME B$.
\end{claim}
\begin{proof}
Observe that $A[\nsubset{j}]=B[\nsubset{j}]$ for $j=1,2$, and $A[\nsubset{j}]=B[\nsubset{j}]^T$ for $j=3,4$. Moreover, $B[\{1,2,3,4\}]=\tw\left(A, \{3,4\}\right)$, therefore $A[\{1,2,3,4\}]\PME B[\{1,2,3,4\}]$ (see~\cref{lem:PME-under-cut-transpose}). Thus, $\kPME{A}{B}{4}$, and $$\{1,2\}\subseteq \calD_1(A,B)\ \ \text{ and }\ \ \{3,4\}\subseteq \calD_2(A,B).$$ Since $A[\{1,2,3,4\}]\nDS B[\{1,2,3,4\}]$ and $A[\{1,2,3,4\}]\nDS B[\{1,2,3,4\}]^T$, $$ \calD_1(A,B)=\{1,2\}\ \ \text{ and }\ \ \calD_2(A,B)=\{3,4\}.$$ If $A\PME B$, then the above equation combined with~\cite[Theorem~2]{Loewy86} implies that $A$ has a cut $X$ such that $\{1,2\}\subseteq X$ and $\{3,4\}\subseteq \comp{X}$. This is a contradiction. Therefore, $A\nPME B$.
\end{proof}

\paragraph*{Explicit construction.} Using the above claim, we provide an explicit construction of matrices $A$ and $B$ such that $\kPME{A}{B}{4}$ but $A\nPME B$. We construct an $n\times n$ matrix $A$ that satisfies the following properties:
\begin{enumerate}
\item The matrix $A$ of the form~\cref{eqn:counter-example-one}, and the matrix $B$ is defined as per~\cref{eqn:counter-example-two} 
\item  The matrix $A$ does not satisfy the second condition of property~$\prop$. Note that rank of submatrix $ A[\{1,2\}, \{3,4\}]$ is one. However, the construction will imply that there does not exist a subset $X \subseteq [n]$ such that $1,2 \in X$, $3,4 \in \comp{X}$, and $\rank A[X, \comp{X}] = 1$. 
\end{enumerate}

In this construction, we define the matrix $A$ as follows:
\[
A=
\begin{pmatrix}
1       &  1      & 1       & 1       & 1       & 1      & 1      & \cdots & 1      & \textcolor{red}{2}\\
2      &  1      & 1       & 1       & 1       & 1      & 1      & \cdots  & 1      & 1\\
1       &  1      & 1       & 2     & 1       & 1      & 1      & \cdots  & 1      & 1\\
1       &  1      & 1       & 1       & \textcolor{red}{2}       & 1      & 1      & \cdots  & 1      & 1\\
1       &  1      &  1      & \textcolor{red}{2}      & 1       & \textcolor{red}{2}      & 1      & \cdots  & 1      & 1\\
1       &  1      &  1      & 1       & \textcolor{red}{2}      & 1      & \textcolor{red}{2}     & \cdots  & 1      & 1\\
1       &  1      &  1      & 1       & 1       & \textcolor{red}{2}     & 1      & \cdots  & 1      & 1\\
\vdots  &  \vdots & \vdots  & \vdots  &\vdots   &\vdots  & \vdots & \cdots  & \vdots & \vdots\\
1     &  1      & 1       & 1       & 1       & 1      & 1      & \cdots  & 1      & \textcolor{red}{2}\\
\textcolor{red}{2}       &  1      & 1       & 1       & 1       & 1      & 1      & \cdots  & \textcolor{red}{2}     & 1\\
\end{pmatrix},
\]
that is,
\begin{enumerate}
\item $A[\{1,2,3,4\}]$ is defined as in mentioned in~\cref{eqn:counter-example-one}
\item For all $i\in\{5,6,7,\ldots, n\}$, $A[i-1, i]=A[i, i-1]=2$, and additionally  $A[1,n]=A[n,1]=2$
\item All the remaining remaining entries are 1.
\end{enumerate}
 
Now, assume that $X\subseteq [n]$ such that $\{1,2\}\subseteq X$ and $\{3,4\}\subseteq \comp{X}$. Since $\rank A[\{1,2\}, \{3, n\}]=2$, in order for $X$ to be cut of $A$, it must include $n$. Continuing this reasoning inductively, we can show that $\{4,5,6,\ldots, n\}\subseteq X$, which contradicts the assumption $4\in \comp{X}$. Thus, no such subset $X$ is a cut of $A$. Therefore, from~\cref{clm:counter-example-tool}, $A\nPME B$.

The above discussion demonstrates the necessity of property $\prop$ to conclude that $\kPME{A}{B}{4}$ implies $A\PME B$. Moreover, the same construction implies something stronger. Consider the following question: can we relax the second condition of property $\prop$ if we strengthen our assumption from $\kPME{A}{B}{4}$ to $\kPME{A}{B}{\ell}$ for some $4<\ell<n$? The construction described above also provides a negative answer to this question. 

Let $k\in\{5,6,7, \ldots, n\}$, and  define $$X_k=\{1,2,k+1, k+2, \ldots, n\}\ \ \text{ and }\ \ \comp{X}_k=\nsubset{k}\setminus X_k,$$ that is, $\comp{X}_k=\{3,4,5, 6, \ldots k-1\}$. Observe that $X_k$ is a cut of the matrix $A[\nsubset{k}]$, and $$B[\nsubset{k}]=\tw\left(A[\nsubset{k}], \comp{X}_k\right).$$ Therefore, from~\cref{lem:PME-under-cut-transpose}, $A[\nsubset{k}]\PME B[\nsubset{k}]$ for all $k\in\{5,6, \ldots, n\}$. On the other hand, the structure of the matrix ensures that $A[\nsubset{k}]\PME B[\nsubset{k}]$ for $k=1,2,3,4$. Thus, we conclude $A[\nsubset{k}]\PME B[\nsubset{k}]$ for all $k\in[n]$, that is, $\kPME{A}{B}{n-1}$. However, as shown earlier, $A\nPME B$. This confirms that even assuming $\kPME{A}{B}{\ell}$ for some $4<\ell < n$ is insufficient to guarantee $A \PME B$ if $A$ does not satisfy property $\prop$.}
\end{document}